\documentclass[a4paper]{article}

\usepackage[top=1in, bottom=1.25in, left=1.25in, right=1.25in]{geometry} 

\usepackage[USenglish]{babel}

\usepackage[utf8]{inputenc}
\usepackage{amsmath}
\usepackage{amssymb}
\usepackage{amsfonts}
\usepackage{amsthm}
\usepackage{tikz} 
\usetikzlibrary{arrows}
\usetikzlibrary{arrows,automata}
\usetikzlibrary{automata,positioning,arrows}
\usetikzlibrary{shapes}
\usetikzlibrary{shapes,backgrounds}
\usetikzlibrary{matrix}
\usetikzlibrary{chains,fit,shapes}
\usepackage{paralist}
\usepackage{wasysym}
\usepackage{url}
\usepackage{hyperref}                
\usepackage{authblk}

\DeclareMathOperator{\DREG}{\mathsf{DREG}}
\DeclareMathOperator{\REG}{\mathsf{REG}}
\DeclareMathOperator{\DFA}{\mathsf{DFA}}
\DeclareMathOperator{\NFA}{\mathsf{NFA}}
\DeclareMathOperator{\eNFA}{\emptyword-\mathsf{NFA}}
\DeclareMathOperator{\DRX}{\mathsf{DRX}}
\DeclareMathOperator{\RX}{\mathsf{RX}}
\DeclareMathOperator{\fDRX}{\mathsf{DRX_{vsf}}}
\DeclareMathOperator{\fRX}{\mathsf{RX_{vsf}}}

\DeclareMathOperator{\TMFA}{\mathsf{TMFA}}
\DeclareMathOperator{\DTMFA}{\mathsf{DTMFA}}
\DeclareMathOperator{\fDTMFA}{\mathsf{DTMFA_{mcf}}}
\DeclareMathOperator{\fTMFA}{\mathsf{TMFA_{mcf}}}
\DeclareMathOperator{\DTMFAacc}{\mathsf{DTMFA^{acc}}}
\DeclareMathOperator{\DTMFArej}{\mathsf{DTMFA^{rej}}}
\DeclareMathOperator{\TMFAacc}{\mathsf{TMFA^{acc}}}
\DeclareMathOperator{\TMFArej}{\mathsf{TMFA^{rej}}}
\DeclareMathOperator{\fTMFArej}{\mathsf{TMFA^{rej}_{mcf}}}
\DeclareMathOperator{\fTMFAacc}{\mathsf{TMFA^{acc}_{mcf}}}
\DeclareMathOperator{\ellDTMFA}{\ell-\DTMFA}
\DeclareMathOperator{\MFA}{\mathsf{MFA}}

\DeclareMathOperator{\PSPACE}{\mathsf{PSPACE}}
\DeclareMathOperator{\NPSPACE}{\mathsf{NPSPACE}}
\DeclareMathOperator{\NP}{\mathsf{NP}}
\DeclareMathOperator{\coNP}{\mathsf{coNP}}

\DeclareMathOperator{\trapstate}{\mathsf{[trap]}}
\DeclareMathOperator{\open}{\mathtt{o}}
\DeclareMathOperator{\close}{\mathtt{c}}
\DeclareMathOperator{\reset}{\mathtt{r}}
\DeclareMathOperator{\unchanged}{\diamond}
\DeclareMathOperator{\opened}{\mathtt{O}}
\DeclareMathOperator{\closed}{\mathtt{C}}
\DeclareMathOperator{\prefequi}{\equiv}
\DeclareMathOperator{\nprefequi}{\not \equiv}
\newcommand{\ta}{\ensuremath{\mathtt{a}}}
\newcommand{\tb}{\ensuremath{\mathtt{b}}}
\newcommand{\tc}{\ensuremath{\mathtt{c}}}
\newcommand{\td}{\ensuremath{\mathtt{d}}}

\DeclareMathOperator{\df}{:=}
\DeclareMathOperator{\ror}{\vee}
\DeclareMathOperator*{\bigror}{\bigvee}

\DeclareMathOperator*{\bigland}{\bigwedge} 
\DeclareMathOperator*{\biglor}{\bigvee}
\DeclareMathOperator{\sepA}{\#}
\DeclareMathOperator{\sepB}{\$}
\DeclareMathOperator{\sepC}{\cent}

\newcommand{\lang}{\mathcal{L}}
\newcommand{\langcl}{\mathcal{L}}
\newcommand{\emptyword}{\varepsilon}
\newcommand{\mdif}{\setminus}
\newcommand{\glush}{\mathcal{M}}
\DeclareMathOperator{\var}{\mathsf{var}}
\newcommand{\bind}[2]{\langle#1\colon #2\rangle}
\newcommand{\rr}[1]{\&#1}
\newcommand{\refl}{\mathcal{R}}
\newcommand{\deref}{\mathcal{D}}
\DeclareMathOperator{\unmark}{\mathsf{unmark}}
\newcommand{\vop}[1]{{[_{#1}}}
\newcommand{\vcl}[1]{{]_{#1}}}
\newcommand{\markpos}[1]{\tilde{#1}}
\newcommand{\mGamma}{\markpos{\Gamma}}
\newcommand{\mSigma}{\markpos{\Sigma}}
\newcommand{\mXi}{\markpos{\Xi}}

\newcommand{\ograph}[1]{G_{\markpos{#1}}}
\newcommand{\onodes}[1]{V_{\markpos{#1}}}
\newcommand{\oedges}[1]{E_{\markpos{#1}}}

\DeclareMathOperator{\src}{\mathsf{src}}
\DeclareMathOperator{\snk}{\mathsf{snk}}
\newcommand{\cheat}{\mathsf{ndet}}
\newcommand{\descitem}[1]{\textbf{\textsf{#1}}}
\newcommand{\ECreg}{\mathsf{EC^{reg}}}

\newcommand{\const}[2]{\mathsf{C}_{#1}(#2)}
\newcommand{\hx}{\hat{x}}
\newcommand{\alpf}[1]{\alpha_{\mathsf{#1}}}
\newcommand{\Lcopy}{L_{\mathsf{copy}}}
\newcommand{\Lex}{L_{\mathsf{ex}}}
\newcommand{\clex}{\overline{\Lex}}
\DeclareMathOperator{\tin}{\mathsf{in}}
\DeclareMathOperator{\tout}{\mathsf{out}}
\newcommand{\alphasq}{\alpha_{\textsf{sq}}}
\newcommand{\fibword}{F_{\omega}}
\newcommand{\netvar}[1]{\mathsf{net}_{#1}}
\DeclareMathOperator{\gfirst}{\mathsf{first}}
\DeclareMathOperator{\gfollow}{\mathsf{follow}}
\DeclareMathOperator{\glast}{\mathsf{last}}
\newcommand{\alphar}{\alpha_{\refl}}
\newcommand{\glab}{\nu}
\DeclareMathOperator{\gmin}{\mathsf{min}}

\definecolor{lboro-purple}{RGB}{81, 45, 109}
\definecolor{lboro-pink}{RGB}{229,0,125}

\newtheorem{lemma}{Lemma}
\newtheorem{theorem}{Theorem}
\newtheorem{proposition}{Proposition}
\newtheorem{definition}{Definition}
\newtheorem{example}{Example}
\newtheorem{corollary}{Corollary}
\newtheorem{remark}{Remark}

\title{Deterministic Regular Expressions With~Back-References\thanks{This work represents an extended version of the paper ``Deterministic Regular Expressions with Back-References'' presented at STACS 2017 and published in LIPICS (\url{http://dx.doi.org/10.4230/LIPIcs.STACS.2017.33}).}}

\author[1]{Dominik D. Freydenberger}
\author[2]{Markus L. Schmid}

\affil[1]{Loughborough University, United Kingdom, \texttt{ddfy@ddfy.de}}
\affil[2]{University of Trier, Germany, \texttt{MLSchmid@MLSchmid.de}}

\begin{document}

\maketitle

\begin{abstract}
Most modern libraries for regular expression matching allow back-references (i.\,e., repetition operators) that substantially increase expressive power, but also lead to intractability. In order to find a better balance between expressiveness and tractability, we combine these with the notion of determinism for regular expressions used in XML DTDs and XML Schema. This includes the definition of a suitable automaton model, and a generalization of the Glushkov construction. We demonstrate that, compared to their non-deterministic superclass, these deterministic regular expressions with back-references have desirable algorithmic properties (i.\,e., efficiently solvable membership problem and some decidable problems in static analysis), while, at the same time, their expressive power exceeds that of deterministic regular expressions without back-references. 
\end{abstract}

%!TEX root=det_SIAM.tex
\section{Introduction}
%Regular expressions were introduced in 1956 by Kleene~\cite{kle:rep}  and quickly found wide use in both theoretical and applied computer science. While the theoretical interpretation of regular expressions remains mostly unchanged (as expressions that describe exactly the class of regular languages), modern applications use variants that vary greatly in expressive power and algorithmic properties. This paper tries to find common ground between two of these variants with opposing approaches to the balance between expressive power and tractability.

Regular expressions were introduced in 1956 by Kleene~\cite{kle:rep} and quickly found wide use in both theoretical and applied computer science, including applications in bioinformatics~\cite{mou:bio}, programming languages~\cite{wal:pro}, model checking~\cite{var:fro}, and XML schema languages~\cite{spe:xml}. While the theoretical interpretation of regular expressions remains mostly unchanged (as expressions that describe exactly the class of regular languages), modern applications use variants that vary greatly in expressive power and algorithmic properties. This paper tries to find common ground between two of these variants with opposing approaches to the balance between expressive power and tractability.

%\todomcom{Changed your smallskip-nonindent-textsf to unnumbered subsubsection (also, paragraph headings in non-serif clashes with section headings in serif).}
\subsubsection*{REGEX}
%\smallskip \noindent \textbf{\textsf{REGEX.}} 
The first variant that we consider are \emph{regex}, regular expressions that are extended with a \emph{back-reference operator}. This operator is used in almost all modern programming languages (like e.\,g.\ Java, PERL, and .NET). For example, the regex $\bind{x}{(\ta\ror\tb)^*}\cdot \rr{x}$ defines $\{ww \mid w\in\{\ta,\tb\}^*\}$,
as $(\ta\ror\tb)^*$ can create a $w\in\{\ta,\tb\}^*$, which is then stored in the variable $x$ and repeated with the reference $\rr{x}$. Hence,  back-references allow to define non-regular languages; but with the side effect that the membership problem is $\NP$-complete (cf.\ Aho~\cite{aho:alg}). 

Regex were first examined from a theoretical point of view by Aho~\cite{aho:alg}, but without fully defining the semantics. There were various proposals for semantics, of which we mention the first by C\^ampeanu, Salomaa, Yu~\cite{cam:afo}, and the recent one by Schmid~\cite{sch:cha}, which is the basis for this paper. 
Apart from defining the semantics, there was work on the expressive power~\cite{cam:afo,car:one,fre:doc},  the static analysis~\cite{car:one,fre:ext,fre:splog}, and the tractability of the membership problem (investigated in terms of a strongly restricted subclass of regex)~\cite{FerSch2015, FerSchVil2015}.
They have also been compared to related models in database theory, e.\,g.\ graph databases~\cite{bar:ext-c,fre:exp} and information extraction~\cite{fag:spa, fre:splog}.

%\smallskip \noindent \textbf{\textsf{Deterministic Regular Expressions.}} 
\subsubsection*{Deterministic Regular Expressions} 
The second variant, \emph{deterministic regular expressions} (also known as \emph{1-unambiguous regular expressions}), uses an opposite approach, and achieves a more efficient membership problem than regular expressions by defining only a strict subclass of the regular languages.

Intuitively, a regular expression is deterministic if, when matching a word from left to right with no lookahead, %each symbol of the word uniquely  
it is always clear where in the expression the next symbol must be matched. 
This property has a characterization via the \emph{Glushkov construction} that converts every regular expression $\alpha$ into a  (potentially non-deterministic) finite automaton $\glush(\alpha)$, by treating each terminal position in $\alpha$ as a state. Then $\alpha$ is deterministic if $\glush(\alpha)$ is deterministic. As a consequence, the membership problem for deterministic regular expressions can be solved more efficiently than for regular expressions in general (more details can be found in~\cite{gro:det}). Hence, in spite of their limited expressive power, deterministic regular expressions are used in actual applications: Originally defined for the ISO~standard for SGML (see Br\"uggemann-Klein and Wood~\cite{bru:one}), they are a central part of the W3C~recommendations on XML~DTDs~\cite{w3c:dtd} and XML~Schema~\cite{w3c:xsd} (see Murata et al.~\cite{mur:tax}). Following the original paper~\cite{bru:one}, deterministic regular expressions have been studied extensively. 
Aspects include computing the Glushkov automaton and deciding the membership problem (e.\,g.~\cite{bru:reg,gro:det,pon:new}), static analysis (cf.~\cite{mar:com}), 
deciding whether a regular language is deterministic (e.\,g.~\cite{cze:dec,gro:det,lu:dec}), 
closure properties and descriptional complexity~\cite{los:clo}, and learning (e.\,g.~\cite{bex:lea}). 
One noteworthy extension are  counter operators (e.\,g.~\cite{gel:reg,gro:det,lat:def}), which we briefly address in Section~\ref{sec:conclusions}.

%\smallskip \noindent \textbf{\textsf{The Best of Two Worlds: Deterministic REGEX.}} 
\subsubsection*{Deterministic REGEX}
%\subsubsection*{The Best of Two Worlds: Deterministic REGEX} 
The goal of this paper is finding common ground between these two variants, by combining the capability of backreferences with the concept of determinism in regular expressions. Generally, our definition of determinism for regex mimics that for classical regular expressions, i.\,e., we define a Glushkov-like conversion from regex into suitable automata and then say that a regex is deterministic if and only if its Glushkov-automaton is. The thus defined class of deterministic regex is denoted by $\DRX$, and $\langcl(\DRX)$ refers to the corresponding language class.\par
The underlying automaton model for this approach is a slight modification of the memory automata ($\MFA$) proposed by Schmid~\cite{sch:cha} as a characterisation for the class of regex-languages. More precisely, we introduce \emph{memory automata with trap-state} ($\TMFA$), for which the deterministic variant, the $\DTMFA$, is better suited for complementation than the deterministic $\MFA$.\par
As indicated by the title of this subsection, it is our hope to preserve, on the one hand, the increased expressive power provided by backreferences, and, on the other hand, the tractability that usually comes with determinism. While it is not surprising that $\DRX$ do not achieve these goals to the full extent, a comprehensive study reveals that their expressive power clearly exceeds that of (deterministic) regular expressions, while, at the same time, being much more tractable than the full class of regex. \par
We shall now outline our main results according to these two aspects, and we start with the algorithmic features of $\DRX$:
%\todo[inline, color=red!30]{MS: Should we also refer to the Theorems of the paper in these summaries of the main results? Not doing this gives us more leeway in stating the results here.}
%\todom{So, not referring is less work, and gives us more leeway? So let's keep things like they are.}
\begin{enumerate}
\item\label{mainAlgoResultsOne} We can decide in time $O(|\Sigma||\alpha|k)$, whether a regex $\alpha$ with $k$ variables and over alphabet $\Sigma$ is deterministic (if so, we can construct its Glushkov-automaton in the same time).
\item\label{mainAlgoResultsTwo} We can decide in time $O(|\Sigma||\alpha|n+ k|w|)$, whether $w \in \Sigma^*$ can be generated by an $\alpha \in \DRX$ with $k$ variables and $n$ occurrences of terminal symbols or variable references.
\item\label{mainAlgoResultsThree} The intersection-emptiness problem for $\DRX$ is undecidable, but in $\PSPACE$ for variable-star-free\footnote{A regex is variable-star-free if each of its sub-regexes under a Kleene-star contains no variable operations (see Section~\ref{sec:static}).} $\DRX$ (as well as the inclusion and equivalence problem). 
%The intersection-emptiness problem, the containment problem and the equivalence problem for variable-star-free\footnote{A regex is variable-star-free if each of its sub-regexes under a Kleene-star contains no variable operations (see Section~\ref{sec:static}).} $\DRX$ are in $\PSPACE$. 
\end{enumerate}

Results~\ref{mainAlgoResultsOne}~and~\ref{mainAlgoResultsTwo} are a consequence of the Glushkov-construction for regex. In view of the $\NP$-hardness of the membership problem for the full class of regex, result~\ref{mainAlgoResultsTwo} demonstrates that the membership problem for $\DRX$ can be solved almost as efficiently as for deterministic regular expression (which is possible in time $O(|\Sigma||\alpha|+|w|)$ \cite{bru:reg, pon:new} or $O(|\alpha|+ |w| \cdot \log \log |\alpha|)$ \cite{gro:det}). The positive results of~\ref{mainAlgoResultsThree} are based on encoding the intersection-emptiness problem in the existential theory of concatenation with regular constraints. With respect to~\ref{mainAlgoResultsThree}, we observe that it is in fact the determinism, which makes the inclusion and equivalence problem for variable-star-free $\DRX$ decidable, since these problems are known to be undecidable for non-deterministic variable-star-free regex (see~\cite{fre:ext}).
%note that the inclusion and equivalence problem are undecidable for non-deterministic variable-star-free regex (see~\cite{fre:ext}). 
Moreover, result~\ref{mainAlgoResultsThree} also yields a $\PSPACE$ minimization algorithm for variable-star-free $\DRX$ (enumerate all smaller candidates and check equivalence).\par
Throughout the paper, there are numerous examples that demonstrate the expressive power of $\DRX$. We also provide a tool for proving non-expressibility, which does not use a pumping argument. In fact, it can be shown that, despite the automata-theoretic characterisation of deterministic regex, $\langcl(\DRX)$ contains infinite languages that cannot be pumped (in the sense as regular languages are ``pumpable''). In addition, we show the following results with respect to $\DRX$'s expressiveness:
\begin{enumerate}
%\item Despite the automata-theoretic characterisation of $\DRX$, $\langcl{\DRX}$ contains infinite languages that are not ``pumpable'' (i.\,e., repeating factors of sufficiently large words arbitrarily often).
\item\label{mainExprResultsOne} There are regular languages that are not in $\langcl(\DRX)$.
\item\label{mainExprResultsTwo} $\langcl(\DRX)$ contains regular languages that are not deterministic regular.
\item\label{mainExprResultsThree} $\langcl(\DRX)$ contains all unary regular languages.
\item\label{mainExprResultsFour} $\langcl(\DRX)$ is not closed under union, concatenation, reversal,  complement, homomorphism, inverse homomorphism, intersection, and intersection with deterministic regular languages.
\end{enumerate}

While result~\ref{mainExprResultsOne} fits to the situation that there are also regular languages that are not deterministic regular languages (and, thus, points out that our definition of determinism restricts regex in a similar way as classical regular expressions), result~\ref{mainExprResultsTwo} points out that in the case of regex this restriction is not as strong. With respect to result~\ref{mainExprResultsThree}, note that not all unary regular languages are deterministic regular (see~\cite{los:clo}). From a technical point of view, it is worth mentioning that in some of our proofs, we use subtleties of the back-reference operator in novel ways. Intuitively speaking, defining and referencing variables under a Kleene-star allows for shifting around factors between different variables (even arbitrarily often in loops), which makes it possible to abuse variables for generating seemingly non-deterministic regular structures in a deterministic way, instead of generating non-regular structures.

As a last strong point in favour of our definition of determinism for regex, we examine a natural   relation of the definition of determinism (i.\,e., requiring determinism only with respect to a constant look-ahead). We prove that checking whether a regex is deterministic under this more general notion is intractable (even for the class of variable-star-free regex).

Summing up, from the perspective of deterministic regular expressions, we propose a natural extension that significantly increases the expressive power, while still having a tractable membership problem. From a regex point of view, we restrict regex to their deterministic core, thus obtaining a tractable subclass. Hence, the authors intend this paper as a starting point for further work, as it opens a new direction on research into making regex tractable. 

%\smallskip \noindent \textbf{\textsf{Organisation of the paper.}} 

\subsubsection*{Structure of the paper} 

In Section~\ref{sec:prelim}, we define some basic concepts and the syntax and semantics of regex; in addition, due to the fact that this aspect is often neglected in the existing literature, which caused misunderstandings, we also provide a thorough discussion of existing variants of regex in theory and practice. Section~\ref{sec:TMFA} is devoted to the definition of memory automata with trap state and their deterministic subclass. In addition, we provide an extensive automata-theoretic toolbox in this section that, besides showing interesting facts about $\TMFA$, shall play an important role for our further results. Next, in Section~\ref{sec:detregex}, we define deterministic regex and provide the respective Glushkov-construction. The expressive power of $\DRX$ (and of related classes resulting from different variants of $\TMFA$) is investigated in Section~\ref{sec:expressive} and the decidability and hardness results of the static analysis of $\DRX$ are provided in Section~\ref{sec:static}. Finally, in Section~\ref{sec:relax}, we discuss the above mentioned relaxation of determinism, and we conclude the paper by giving some conclusions in Section~\ref{sec:conclusions}.

\section{Preliminaries}\label{sec:prelim} 
We use $\emptyword$ to denote the \emph{empty word}. The subset and proper subset relation are denoted by $\subseteq$ and $\subset$, respectively. Let $\Sigma$ be a finite terminal alphabet (unless otherwise noted, we assume $|\Sigma|\geq 2$) and let $\Xi$ be an infinite variable alphabet with $\Xi\cap\Sigma = \emptyset$. For a word $w \in \Sigma^*$ and for every $i$, $1 \leq i \leq |w|$, $w[i]$ denotes the symbol at position $i$ of $w$. We define $w^0 \df \emptyword$ and $w^{i+1}\df w^i \cdot w$ for all $i\geq 0$, and, for $w= a_1\cdots a_n$ with $a_i\in\Sigma$, let  $w^{m+\frac{i}{n}} = w^m \cdot a_1 \cdots a_i$ for all $m\geq 0$ and all $i$ with $0\leq i \leq n$. A $v \in \Sigma^*$ is a \emph{factor} of $w$ if there exist $u_1, u_2\in \Sigma^*$ with $w = u_1 v u_2$. If $u_2=\emptyword$, $v$ is also a $\emph{prefix}$ of $w$. 

We use the notions of deterministic and non-deterministic finite automata ($\DFA$ and $\NFA$) like~\cite{hop:int}. If an $\NFA$ can have $\emptyword$-transitions, we call it an $\eNFA$. Given a class $\mathcal{C}$ of language description mechanisms (e.\,g., a class of automata or regular expressions), we use $\langcl(\mathcal{C})$ to denote the class of all languages $\lang(C)$ with $C\in\mathcal{C}$. The \emph{membership problem for $\mathcal{C}$} is defined as follows: Given a $C\in\mathcal{C}$ and a $w\in\Sigma^*$, is $w\in\lang(C)$?\par
Next, we define the syntax and semantics of regular expressions with backreferences.
%\subsection{Regex}\label{sec:regex}
\begin{definition}[Syntax of regex]\label{def:rx}
The set $\RX$ of \emph{regex} over $\Sigma$ and $\Xi$ is recursively defined as follows:\medskip\\
%	We define $\RX$, the set of \emph{regex} over $\Sigma$ and $\Xi$, recursively:\\
%	\begin{description}
%		\item[Terminals and $\emptyword$:] $a\in\RX$ and $\var(a)=\emptyset$ for every $a\in(\Sigma\cup\{\emptyword\})$.
%		\item[Variable reference:] $\rr{x}\in\RX$ and $\var(\rr{x})= \{x\}$ for every $x\in \Xi$.
%		\item[Concatenation:] $(\alpha\cdot\beta)\in \RX$ and $\var(\alpha\cdot\beta)=\var(\alpha)\cup\var(\beta)$ if $\alpha,\beta\in \RX$.
%		\item[Disjunction:] $(\alpha\ror\beta)\in \RX$ and $\var(\alpha\ror\beta)=\var(\alpha)\cup\var(\beta)$  if $\alpha,\beta\in \RX$. 
%		\item[Kleene plus:]$(\alpha^+)\in \RX$ and $\var(\alpha^+)=\var(\alpha)$ if $\alpha\in\RX$. 
%		\item[Variable binding:]$\bind{x}{\alpha} \in \RX$ and $\var(\bind{x}{\alpha}) = \var(\alpha)\cup\{x\}$  if $\alpha\in \RX$ with $x\in\Xi\mdif\var(\alpha)$.
%	\end{description}
\textbf{Terminals and $\emptyword$:} $a\in\RX$ and $\var(a)=\emptyset$ for every $a\in(\Sigma\cup\{\emptyword\})$.\\
\textbf{Variable reference:} $\rr{x}\in\RX$ and $\var(\rr{x})= \{x\}$ for every $x\in \Xi$.\\
\textbf{Concatenation:} $(\alpha\cdot\beta)\in \RX$ and $\var(\alpha\cdot\beta)=\var(\alpha)\cup\var(\beta)$ if $\alpha,\beta\in \RX$.\\
\textbf{Disjunction:} $(\alpha\ror\beta)\in \RX$ and $\var(\alpha\ror\beta)=\var(\alpha)\cup\var(\beta)$  if $\alpha,\beta\in \RX$. \\
\textbf{Kleene plus:} $(\alpha^+)\in \RX$ and $\var(\alpha^+)=\var(\alpha)$ if $\alpha\in\RX$. \\
\textbf{Variable binding:}$\bind{x}{\alpha} \in \RX$ and $\var(\bind{x}{\alpha}) = \var(\alpha)\cup\{x\}$  if $\alpha\in \RX$ with $x\in\Xi\mdif\var(\alpha)$.\medskip\\
In addition, we allow $\emptyset$ as a regex (with $\var(\emptyset)=\emptyset$), but we do not allow $\emptyset$ to occur in any other regex. An $\alpha\in\RX$ with $\var(\alpha)=\emptyset$ is called a  \emph{proper regular expression}, or just \emph{regular expression}. We use $\REG$ to denote the set of all  regular expressions. We add and omit parentheses freely, as long as the meaning remains clear. We use the Kleene star $\alpha^*$ as shorthand for $\emptyword\ror\alpha^+$, and $A$ as shorthand for $\bigror_{a\in A}a$ for non-empty  $A\subseteq \Sigma$.
\end{definition}

We define the semantics of regex using the \emph{ref-words} (short for \emph{reference words}) by Schmid~\cite{sch:cha}. A ref-word is a word over $(\Sigma\cup\Xi\cup\Gamma)$, where $\Gamma\df\{\vop{x}, \vcl{x},\mid x\in\Xi\}$. Intuitively, the symbols $\vop{x}$ and $\vcl{x}$ mark the beginning and the end of the match that is stored in the variable $x$, while an occurrence of $x$ represents a reference to that variable. Instead of defining the language of a regex~$\alpha$ directly, we first treat $\alpha$ as a generator of ref-words by defining its \emph{ref-language} $\refl(\alpha)$ as follows. 
\begin{itemize}
\item For every $\alpha \in \Sigma\cup\{\emptyword\}$, $\refl(\alpha)\df\{\alpha\}$.
\item For every $x\in \Xi$, $\alpha \in \RX$,
\begin{itemize}
\item $\refl(\rr{x})\df \{x\}$, 
\item $\refl(\bind{x}{\alpha})\df (\vop{x} \refl(\alpha)\vcl{x})$.
\end{itemize}
\item For every $\alpha, \beta \in \RX$, 
\begin{itemize}
\item $\refl(\alpha\cdot\beta)\df\refl(\alpha)\cdot\refl(\beta)$, 
\item $\refl(\alpha\ror\beta)\df\refl(\alpha)\cup\refl(\beta)$, and 
\item $\refl(\alpha^+)\df\refl(\alpha)^+$.
\end{itemize}
\end{itemize}
%If $\alpha\in \Sigma\cup\{\emptyword\}$, $\refl(\alpha)\df\{\alpha\}$; and $\refl(\rr{x})\df \{x\}$ for all $x\in \Xi$. Furthermore, $\refl(\alpha\cdot\beta)\df\refl(\alpha)\cdot\refl(\beta)$, $\refl(\alpha\ror\beta)\df\refl(\alpha)\cup\refl(\beta)$, and $\refl(\alpha^+)\df\refl(\alpha)^+$. Finally, $\refl(\bind{x}{\alpha})\df (\vop{x} \refl(\alpha)\vcl{x})$. 
In particular, if $\alpha$ is a regular expression, then $\lang(\alpha)=\refl(\alpha)$. An alternative definition of the ref-language would be $\refl(\alpha) \df \lang(\alpha_{\refl})$, where $\alpha_{\refl}$ is the proper regular expression obtained from $\alpha$ by replacing each sub-regex $\bind{x}{\beta}$ by $\vop{x} \beta_{\mathcal{R}} \vcl{x}$, and each $\rr{x}$ by $x$.

Intuitively speaking, every occurrence of a variable $x$ in some $r \in \refl(\alpha)$ functions as a pointer to the next factor $\vop{x} v \vcl{x}$ to the left of this occurrence (or to $\emptyword$ if no such factor exists). In this way, a ref-word $r$ compresses a word over $\Sigma$, the so-called dereference $\deref(r)$ of $r$, which can be obtained by replacing every variable occurrence $x$ by the corresponding factor $v$ (note that $v$ might again contain variable occurrences, which need to be replaced as well), and removing all symbols $\vop{x},\vcl{x}\in\Gamma$ afterwards. See~\cite{sch:cha} for a more detailed definition, or the following Example~\ref{refWordExample} for an illustration. Finally, the language of a regex $\alpha$ is defined by $\lang(\alpha)\df \{\deref(r)\mid r\in \refl(\alpha)\}$. 
\begin{example}\label{refWordExample}
Let $\alpha\df \bigl(\bind{x}{(\ta \ror \tb)^+}\rr{x}\big)^+ $. Then 
\begin{equation*}
\refl(\alpha) = \{\vop{x}w_1\vcl{x}\cdot x \cdots \vop{x}w_n\vcl{x}\cdot  x \mid n\geq 1, w_i\in\{\ta, \tb\}^+\}\,,
\end{equation*}
or, equivalently, $\lang(\alpha) = (\Lcopy)^+$, with $\Lcopy \df \{ww\mid w\in\{\ta,\tb\}^+\}$. \par
An interesting example is the regex $\alphasq \df \bigl(\bind{x}{\rr{y}}\bind{y}{\rr{x}\cdot \ta}\bigr)^*$ with ref-language $\refl(\alphasq) = \{r_i \mid i\geq 0\}$, where $r_i\df \bigl(\vop{x}y\vcl{x}\cdot\vop{y}x\cdot \ta\vcl{y}\bigr)^i$. For example, for 
$$r_3 = \vop{x}y\vcl{x}\cdot\vop{y}x\cdot \ta\vcl{y}\cdot\vop{x}y\vcl{x}\cdot\vop{y}x\cdot \ta\vcl{y}\cdot\vop{x}y\vcl{x}\cdot\vop{y}x\cdot \ta\vcl{y},$$ we have $r_3\in \refl(\alphasq)$ with $\deref(r_3) = \ta^{9}$. Using induction, we can verify that $\deref(r_i)=\ta^{i^2}$. Thus, $\lang(\alphasq)=\{\ta^{n^2}\mid n\geq 0\}$.%\todomcom{Included definition of $r_i$, which was missing.}
\end{example}

Hence, unlike  regular expressions, regex can define non-regular languages. The expressive power comes at a price:  their membership problem is $\NP$-complete~(follows from Angluin~\cite{ang:fin}), and various other problems are undecidable (Freydenberger~\cite{fre:ext}).
Starting with Aho~\cite{aho:alg}, there have been various  approaches to specifying syntax and semantics of regex. While~\cite{aho:alg} only sketched the intuition behind the semantics, the first formal definition (using parse trees) was proposed by C\^ampeanu, Salomaa, Yu~\cite{cam:afo}, followed by the ref-words  of Schmid~\cite{sch:cha}. In the following, we provide a more detailed discussion of the different approaches and actual implementations of regex.
% compare these approaches and actual implementations in more detail
%
%For a comparison between these approaches and actual implementations, see the full version~\cite{fre:drx}.

\subsection{Regex in Theory and Practice}\label{sec:regexstuff}
In this section, we motivate the choice of the formalization of regex syntax and semantics that are used in the current paper, in particular in comparison to~\cite{cam:afo}, and then connect these to  the use of back-references in actual implementations. Note that in order to explain how our results and concepts can be adapted to various alternative definitions of syntax and semantics, we anticipate some of the technical content of our paper.
%\todomcom{Thanks. Made some slight changes.}

\subsubsection{Choices behind the definition}\label{sec:choices}
We begin with a discussion of semantics of back-references, which most actual implementations define in terms of the used matching algorithm\footnote{From a theory point of view, this might be considered a rather generous use of the term ``define''.}. For a theoretical analysis, this approach is not satisfactory.

 To the authors' knowledge, the first theoretical analysis of regular expressions is due to Aho~\cite{aho:alg}, who defined the semantics informally.
 C\^ampeanu, Salomaa, Yu~\cite{cam:afo} then proposed a definition using parse trees, which was precise, but rather technical and unwieldy. Schmid~\cite{sch:cha} then introduced the definition with ref-words that we use in the current paper. 
The two  definitions differ only in some semantical particularities, which we discuss further down.
%\todom{Change previous sentence. ``similar up to'' was a little too weak for my taste.}

The most obvious difference in approaches to syntax is that some formalizations, like~\cite{cam:afo}, do not use variables, but numbered back-references. For example,  $\bind{x}{\ta^*}\tb\cdot \rr{x}$ would be written as $(\ta^*)\tb \backslash 1$, where $\backslash 1$ refers to the content of the first pair of parentheses (called the first \emph{group}).

After working with this definition for some time, the authors of the present paper came to the conclusion that using numbered back-references instead of named variables is inconvenient (both when reading and writing regex). The developers of actual implementations seem to agree with this sentiment: While using numbered back-references was well-motivated when considering PERL at the time~\cite{cam:afo} was published,  most current regex  dialects allow the use of \emph{named groups}, which basically act like our variables (depending on the actual dialect, see below). The choice between variables and numbered groups is independent of the choice of semantics, as parse trees can also be used with variables, see~\cite{fre:ext}.
Hence, using variables instead of numbers is a natural choice. 

Building on this, the next question is whether the same variable can be bound at different places in the regex (which is automatically excluded by the use of numbered groups as in~\cite{cam:afo}), i.\,e., whether one allows expressions like 
$$\bigl(\bind{x}{\ta^+}\bind{y}{\tb^+}\bigr)\ror\bigl(\bind{y}{\tb^+}\bind{x}{\ta^+}\bigr)\tc\cdot \rr{x}\cdot \rr{y}.$$

While some implementations that have developed from back-references forbid these constructions to certain degrees 
%(again, see below)
(see Section~\ref{sec:ActImpl} below), there seems to be no particular reason for this decision when approaching this question without this historical baggage. In fact, one can argue from a point of applications that expressions like the following make sense (abstracting away some details that would be needed in actual use):
$$\Sigma^*\Bigl(\bigl(\texttt{Name:}\bind{x}{\Sigma^+} \texttt{ Title:}\bind{y}{\Sigma^+} \bigr) \ror \bigl(\texttt{Title:}\bind{y}{\Sigma^+} \texttt{ Name:}\bind{x}{\Sigma^+} \bigr)  \Bigr)\Sigma^*$$
In fact, these constructions are explicitly allowed in the regex formulas of Fagin et al.~\cite{fag:spa}, that are closely related to regex. In particular, both the semantic definitions (ref-words and parse trees) allow this choice. Thus, there seems to be no particular practical reason to disallow these constructions when considering only the model (instead of its algorithmic properties). 

 In addition to disallowing the repeated binding of the same variable described above, 
the regex definition in~\cite{cam:afo} also includes a  syntactic restriction that changes the expressive power considerably: It requires that a backreference $\backslash n$ can only appear in a regex if it occurs \emph{to the right} of corresponding group number~$n$. In~\cite{cam:afo}, otherwise, the expression is called a ``semi-regex''. 
Consider  $\alphasq=\bigl(\bind{x}{\rr{y}}\bind{y}{\rr{x}\cdot \ta}\bigr)^*$ from Example~\ref{refWordExample}. In the numbered notation of~\cite{cam:afo}, this  would be expressed as   $\beta\df \bigl((_{2}\backslash 3)_{2}   (_{3}\backslash 2 \cdot \ta)_{3}\bigr)^*$, when adding group numbers to the groups to increase readability. But using definitions from~\cite{cam:afo}, $\beta$ is only  a semi-regex,  as the reference $\backslash 3$ occurs to the left of group~3. 

The motivation behind this restriction is not explained in~\cite{cam:afo}. While one might argue that this was chosen to avoid referencing unused groups, the definition of semantics in~\cite{cam:afo} still needs to deal with this problem in regexes like $((_{2}\ta)_{2}\ror (_{3}\ta)_{3})\cdot\backslash 2\cdot\backslash 3$, and handles them by assigning~$\emptyword$ (like the definition from~\cite{sch:cha}, which we use as well). Hence, even on ``semi-regex'', the parse tree semantics behave like the ref-word semantics.

Arguably, the restriction has an advantage from a theoretical point of view, as it allows C\^ampeanu, Salomaa, Yu~\cite{cam:afo}  and Carle and Narendran~\cite{car:one} to define pumping lemmas for this class. Using these, it is possible to show that languages like $\lang(\alphasq)$ from Example~\ref{refWordExample} or the language from Lemma~\ref{lem:fibonacci} cannot be expressed with the regex model from~\cite{cam:afo}. But in other areas, there seems to be no advantage in this choice: Even under this restriction, the membership problem is $\NP$-complete (since it is still possible to describe Angluin's pattern languages~\cite{ang:fin}), the undecidability results from~\cite{fre:ext} on various problems of static analysis are unaffected by this choice, and even the proof of Theorem~\ref{thm:intersectundec} (the undecidability of the disjointness problem for deterministic $\RX$) directly works on this subclass.
In summary, the authors of the current paper see no reason to adapt this restriction.

%For full disclosure, the second author points out that he misinterpreted the regex definition of~\cite{cam:afo} when citing the paper in his own articles~\cite{sch:ins,sch:cha}. 
%Hence, although those papers refer to~\cite{cam:afo} for the full definition of regex, they talk about the language class $\langcl(\RX)$ of the current paper.
%All mentions of  ``regex'' in those papers refer to the expressions that are called ``semi-regex'' in~\cite{cam:afo}. 
%\todo[inline, color=red!30]{MS: I rephrased the next paragraph. I think it is better not to explain what I meant in my papers in terms of semi-regex. I think the term semi-regex does not really cover any class of regex that we are really interested in.}
For full disclosure, the second author points out that in his own articles~\cite{sch:ins,sch:cha}, using~\cite{cam:afo} as reference for a full definition of regex is not entirely correct, since the restrictions of~\cite{cam:afo} discussed above are not used in these papers; instead they talk about the language class $\langcl(\RX)$ of the current paper.
%\todomcom{I though about the issue with the semiregex and the thing the Dortmund people pointed out. For most of this section, we do not really need to go into detail with the difference. But in this paragraph, we should address this in some way. In a way, the problem is more with the claims in your old papers than with the claims in this paper.}

The last choice in the definition that we need to address is how we deal with referencing undefined variables. Both~\cite{cam:afo} and~\cite{sch:cha} default those references to $\emptyword$ (as do others, like~\cite{fre:ext}); but there is also literature, like~\cite{car:one}, that uses~$\emptyset$ as default value (under these semantics, a ref-word that contains a variable that de-references to  $\emptyset$ cannot generate any terminal words; the same holds for a parse tree that contains such a reference). This choice can easily be implemented in both semantics by discarding a ref-word or parse tree that contains such a reference; and a $\TMFA$ (see Section~\ref{sec:TMFA}) can reject if a run encounters a reference to such an undefined memory.

While these ``$\emptyset$-semantics'' are also used in some actual implementations, the authors of the current paper are against this approach. One of the reasons is that using $\emptyset$ as default allows the use of curious synchronization effects that distract from the main message of this paper. For example, let $\Sigma=\{a_1,\ldots,a_n\}$ for some $n\geq 1$, and define 
$$\alpha_n \df \bigl(  \bigror_{i=1}^n (a_i \cdot\bind{x_i}{\emptyword})  \bigr)^n \cdot \rr{x_1} \cdots \rr{x_n}.$$
If unbound variables default to $\emptyset$, this regex generates the language
$$\bigl\{a_{\pi(1)}\cdots a_{\pi(n)} \mid \text{$\pi$ is a permutation of $\{1,\ldots,n\}$} \bigr\},$$
as every variable $x_i$ needs to be assigned $\emptyword$ exactly once (otherwise, a reference would return $\emptyset$ and block). Hence, using this semantics, even variables that are bound only to $\emptyword$ can be used for synchronization effects. While this can lead to interesting constructions, the authors think that it provides more insight to study the effects of back-references on lower bounds without relying on these additional features. This way, there is no question whether the  hardness of the examined problems  is due to the effects of the $\emptyset$-semantics. 

Furthermore, all examples in the present paper can be adapted from the used $\emptyword$-semantics to $\emptyset$-semantics: Given an $\alpha\in\RX$ with the variables $x_1,\ldots,x_k$, define $\alpha' \df \bind{x_1}{\emptyword}\cdots\bind{x_k}{\emptyword}\cdot \alpha$. First, we observe that the language that is defined by $\alpha'$ under $\emptyset$-semantics is the same language that $\alpha$ defines under $\emptyword$-semantics. Furthermore, note that if $\alpha$ is deterministic (in the sense as shall be defined in Section~\ref{sec:detregex}), $\alpha'$ is also deterministic. The analogous construction can be used for $\TMFA$ (and $\DTMFA$).

While it is possible to adapt most of the results in the current paper directly to this alternative semantics, the authors chose to keep the paper focused on $\emptyword$-semantics.

\subsubsection{Actual implementations}\label{sec:ActImpl} We now give a brief overview of how back-references are used in some actual implementations. For a good introduction on various dialects, the authors recommend~\cite{goy:reg}, in particular the section on back-references and named groups.  As this behavior is often under-defined, badly documented, and implementation dependent, this can only be a very short and superficial summary of some behavior. 

Before we go into details, we address why back-references are used, in spite of the resulting $\NP$-hard membership problem: Most regex libraries use a backtracking algorithm that can have exponential running time, even on many proper regular expressions  (see Cox~\cite{cox:reg}). From this point of view, back-references can be added with little implementation effort and without changing the efficiency of the program.

Most modern dialects of regex not only support numbered back-references as used by~\cite{cam:afo}, but also  named capture groups, which work like our variables. In some dialects like e.\,g.\ Python, PERL, and PCRE, these act as aliases for back-references with numbers; hence, $\bind{x}{\ta^*}\tb \rr{x}$ would be interpreted as $(\ta^*)\tb \backslash 1$.  As a consequence, each name resolves to a well-defined number.  As some of these dialects assign the empty set as default value of unbound back-references (or group names), the resulting behavior is similar implicitly requiring the restriction from~\cite{cam:afo}. This implementation of named capture groups seems to be mostly for historical reasons (as back-references were introduced earlier).

In contrast to this, there are other dialects that use numbered back-references and explicitly allow references to access groups that occur to their right in the expression. For example, the W3C recommendation for XPath and XQuery functions and operators~\cite{w3c:xfun} defines regular expressions with back-references for the use in \texttt{fn:matches}. There, it is possible to refer to capture groups that occur to the right of the reference (although only for the capture groups 1 to 9, but not for 10 to 99, which might be considered a peculiar decision). As this dialect defaults unbound references to $\emptyword$, it is possible to directly express $\alphasq$ by renaming the variable references to back-references.

Furthermore, {.NET} allows the same name to be used for different groups,  for example $((\bind{x}{\ta^+}\ror \bind{x}{\tb^+}\tc)\rr{x})^*$. While .NET  defaults unset variables to $\emptyset$, it is possible to express $\lang(\alphasq)$, by using an expression like $\bind{x}{\emptyword}\cdot\alphasq$.  In the same way, every regex in the sense of our paper can be converted into an equivalent .NET regex. 

Finally, in 2007 (just four years after the publication of~\cite{cam:afo}),  PERL 5.10 introduced   \emph{branch reset groups} (which were also adapted in PCRE). These reset the numbering inside disjunctions, and  allow expressions that behave like the expression $((\bind{x}{\ta^+}\ror \bind{x}{\tb^+}\tc)\rr{x})^*$.  This allows PERL regex to replicate a large part of the behavior of .NET regex. 

In conclusion, it seems that almost every formalization of regex syntax and semantics can be justified by finding the right dialect; but every restriction might be superseded by the continual evolution of regex dialects. Hence, the current paper attempts to avoid restrictions; and when in doubt, we choose natural definitions over trying to directly emulate a single dialect. Therefore, we use variables instead of numbered back-references, and allow multiple uses of the same variable name. 

The authors acknowledge that most actual implementations of ``regular expressions'' allow additional operators. Common features are \emph{counters}, which allow constructions like e.\,g. $\ta^{2,5}$ that define the language $\{\ta^i\mid i\in \{2,\ldots,5\}\}$, \emph{character classes} and \emph{ranges}, which are shortcuts for sets like ``all alphanumeric symbols'' or ``all letters from \texttt{b} to \texttt{y}'', and \emph{look ahead} and \emph{look behind}, which can be understood as allowing the expression to call another expressions as a kind of subroutine. 

While these operators are outside of the scope of the current paper, we briefly address the issue of counters. These are used in XML~DTDs and XML~Schema, and were studied in connection to determinism. In particular, Gelade, Gyssens, Martens~\cite{gel:reg} described how counters can be added to finite automata and proposed an appropriate extension of determinism and Glushkov construction to this model.  Although the current paper does not address this matter (in order to keep the paper focussed), the $\TMFA$ that we introduce in Section~\ref{sec:TMFA} can also be extended with counters (like the extension to $\NFA$ in~\cite{gel:reg}). Likewise, the Glushkov constructions of~\cite{gel:reg} and the current paper can be combined, as can the notions of determinism. The membership problem for the resulting class of deterministic regex with counters can then be solved as efficiently as for deterministic regex (see Theorem~\ref{thm:drxmembership}).
%!TEX root=det_SIAM.tex
\section{Memory Automata with Trap State}\label{sec:TMFA}

In this section, we define memory automata with trap-state, the deterministic variant of which will be the algorithmic foundation for deterministic regex. Before moving on to the actual definition of deterministic regex in Section~\ref{sec:detregex} and the applications of memory automata with trap-state, we subject this automaton model to a thorough theoretical analysis. Most of the thus obtained insights will have immediate consequences and applications for proving the main results regarding deterministic regex, while others have the mere purpose of supporting our understanding of memory automata (and therefore the important class of regex languages).\par
%\todo[inline]{MS: Reading this section feels a bit like reading a classical automata theory paper. I think we should avoid this impression, if possible. Maybe we can add some paragraph that explains that in this section we mainly define and proof tools that we need later and everything else formal-languages-like that we get on the way are just interesting additions.}
\emph{Memory automata} \cite{sch:cha} are a simple automaton model that characterizes $\langcl(\RX)$. Intuitively speaking, these are classical finite automata that can record consumed factors in memories, which can be recalled later on in order to consume the same factor again. However, for our applications, we need to slightly adapt this model to \emph{memory automata with trap-state}.
\begin{definition}\label{TMFADefinition}
For every $k \in \mathbb{N}$, a \emph{$k$-memory automaton with trap-state}, denoted by $\TMFA(k)$, is a tuple $M = (Q, \Sigma, \delta, q_0, F)$, where $Q$ is a finite set of \emph{states} that contains the \emph{trap-state} $\trapstate$, $\Sigma$ is a finite \emph{alphabet}, $q_0 \in Q$ is the \emph{initial state}, $F \subseteq Q$ is the set of \emph{final states} and $\delta\colon Q \times (\Sigma \cup \{\varepsilon\} \cup \{1, 2, \ldots, k\}) \rightarrow \mathcal{P}(Q  \times \{\open, \close, \reset, \unchanged\}^k)$ is the \emph{transition function} (where $\mathcal{P}(A)$ denotes the power set of a set $A$), which satisfies $\delta(\trapstate, b) = \{(\trapstate, \unchanged, \unchanged, \ldots, \unchanged)\}$, for every $b \in \Sigma \cup \{\varepsilon\}$, and $\delta(\trapstate, i) = \emptyset$, for every $i$, $1 \leq i \leq k$. The elements $\open$, $\close$, $\reset$ and $\unchanged$ are called \emph{memory instructions} (they stand for \underline{\textsf{o}}pening, \underline{\textsf{c}}losing and \underline{\textsf{r}}eseting a memory, respectively, and $\unchanged$ leaves the memory unchanged).

A \emph{configuration} of $M$ is a tuple $(q, w, (u_1, r_1), \ldots, (u_k, r_k))$, where $q \in Q$ is the \emph{current state}, $w$ is the \emph{remaining input} and, for every $i$, $1 \leq i \leq k$, $(u_i, r_i)$ is the \emph{configuration of memory $i$}, where $u_i \in \Sigma^*$ is the \emph{content of memory $i$} and $r_i \in \{\opened, \closed\}$ is the \emph{status of memory $i$} (i.\,e., $r_i = \opened$ means that memory $i$ is open and $r_i = \closed$ means that it is closed). The \emph{initial configuration} of $M$ (\emph{on input $w$}) is the configuration $(q_0, w, (\varepsilon, \closed), \ldots, (\varepsilon, \closed))$, a configuration $(q, w, (u_1, r_1), \ldots, (u_k, r_k))$ is an \emph{accepting configuration}  if $w = \varepsilon$ and  $q \in F$.

$M$ can change from a configuration $c = (q, v w, (u_1, r_1), \ldots, (u_k, r_k))$ to a configuration $c' = (p, w, (u'_1, r'_1), \ldots, (u'_k, r'_k))$, denoted by $c \vdash_{M} c'$, if there exists a transition $\delta(q, b) \ni (p, s_1, \ldots, s_k)$ with either ($b \in (\Sigma \cup \{\varepsilon\})$ and $v = b$) or ($b \in \{1, 2, \ldots, k\}$, $s_b = \close$ and $v = u_b$), and, for every $i$, $1 \leq i \leq k$, 
\begin{itemize}
\item $(s_i = \unchanged) \wedge (r_i = \opened) \Rightarrow (u'_i, r'_i) = (u_i v, r_i)$, 
\item $(s_i = \unchanged) \wedge (r_i = \closed) \Rightarrow (u'_i, r'_i) = (u_i, r_i)$,
\item $s_i = \open \Rightarrow (u'_i, r'_i) = (v, \opened)$, 
\item $s_i = \close \Rightarrow (u'_i, r'_i) = (u_i, \closed)$,
\item $s_i = \reset \Rightarrow (u'_i, r'_i) = (\varepsilon, \closed)$.
\end{itemize}
%\begin{tabular}{llllll}
%$s_i = \unchanged \wedge r_i = \opened$ & $\Rightarrow$ & $(u'_i, r'_i) = (u_i v, r_i)$, & $s_i = \unchanged \wedge r_i = \closed$ & $\Rightarrow$ & $(u'_i, r'_i) = (u_i, r_i)$,\\
%$s_i = \open$ & $\Rightarrow$ & $(u'_i, r'_i) = (v, \opened)$, & $s_i = \close$ & $\Rightarrow$ & $(u'_i, r'_i) = (u_i, \closed)$,\\
%$s_i = \reset$ & $\Rightarrow$ & $(u'_i, r'_i) = (\varepsilon, \closed)$. & & &
%\end{tabular}\\
%\begin{inparaenum}\todo{D: Compactified}
%	\item $s_i = \unchanged$ and $r_i = \opened$  implies  $(u'_i, r'_i) = (u_i v, r_i)$,
%	\item $s_i = \unchanged$ and $r_i = \closed$ implies $(u'_i, r'_i) = (u_i, r_i)$,
%	\item $s_i = \open$  implies  $(u'_i, r'_i) = (v, \opened)$,
%	\item $s_i = \close$  implies  $(u'_i, r'_i) = (u_i, \closed)$,
%	\item $s_i = \reset$  implies  $(u'_i, r'_i) = (\varepsilon, \closed)$.
%\end{inparaenum}
Furthermore, $M$ can change from a configuration $(q, v w, (u_1, r_1), \ldots, (u_k, r_k))$ to the configuration $(\trapstate, w, (u_1, r_1), \ldots, (u_k, r_k))$, if $\delta(q, b) \ni (p, s_1, \ldots, s_k)$ for some $p \in Q$, $b \in \{1, 2, \ldots, k\}$ and $s_b = \close$, such that $u_b = v v'$ with $v' \neq \varepsilon$ and $v'[1] \neq w[1]$.

A transition $\delta(q, b) \ni (p, s_1, s_2, \ldots, s_k)$ is an \emph{$\varepsilon$-transition} if $b = \varepsilon$ and is called \emph{consuming}, otherwise (if all transitions are consuming, then $M$ is called \emph{$\varepsilon$-free}). If $b \in \{1, 2, \ldots, k\}$, it is called a \emph{memory recall transition} and the situation that a memory recall transition leads to the state $\trapstate$, is called a \emph{memory recall failure}. 

The symbol $\vdash_{M}^*$ denotes the reflexive and transitive closure of $\vdash_{M}$. A $w \in \Sigma^*$ is \emph{accepted} by $M$ if $c_{\text{init}} \vdash^*_{M} c_f$, where $c_{\text{init}}$ is the initial configuration of $M$ on $w$ and $c_f$ is an accepting configuration. The set of words accepted by $M$ is denoted by $\lang(M)$.
\end{definition}
Note that executing the open action $\open$ on a memory that already contains some word discards the previous contents of that memory. A crucial part of $\TMFA$ is the trap-state $\trapstate$, in which computations terminate, if a memory recall failure happens. If $\trapstate$ is not accepting, then $\TMFA$ are (apart from  negligible formal differences) identical to the memory automata introduced in \cite{sch:cha}, which characterize the class of regex language.
If, on the other hand, $\trapstate$ is accepting, then every computation with a memory recall failure is accepting (independent from the remaining input). While it seems counter-intuitive to define the words of a language via ``failed'' back-references, the possibility of having an accepting trap-state yields closure under complement for \emph{deterministic} $\TMFA$ (see Theorem~\ref{closureComplementTheorem}). It will be convenient to consider the partition of $\TMFA$ into $\TMFArej$ and $\TMFAacc$ (having a rejecting and an accepting trap-state, respectively).\par
Next, we illustrate the concept of memory automata with trap state by some examples\footnote{For the sake of convenience, we present $\TMFA$ in the form of the usual automata diagrams (initial states are marked by an unlabeled incoming arc, accepting states by an additional circle and arcs are labelled with the transition tuples).} (further illustrations can be found in~\cite{sch:cha}).\par
Intuitively speaking, in a single step of a computation of a $\TMFA$, we first change the memory statuses according to the memory instructions $s_i$, $1 \leq i \leq k$, and then a (possibly empty) prefix $v$ of the remaining input ($v$ is either from $\Sigma \cup \{\varepsilon\}$ or it equals the content of some memory that, according to the definition, has been closed by the same transition) is consumed and appended to the content of every memory that is currently open (note that here the new statuses after applying the memory instructions count). The changes of memory configurations caused by a transition are illustrated in Figure~\ref{fig:MemConfChanges}.

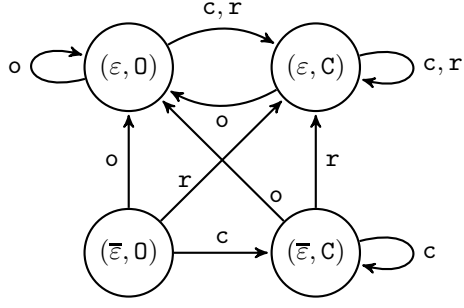
\begin{figure}
\begin{center}
\begin{tikzpicture}[->,>=stealth',shorten >=1pt,auto,node distance=7em, thick]
  \tikzstyle{every state}=[fill=none,draw=black,text=black]
  \node[state] (emptyopened)   [fill= none, draw=black,text=black]    {$(\varepsilon, \opened)$};
  \node[state] (nonemptyopened) [below of=emptyopened]       					{$(\overline\varepsilon, \opened)$};
  \node[state]         (emptyclosed)   [right of=emptyopened]                           {$(\varepsilon, \closed)$}; 
  \node[state]         (nonemptyclosed)   [right of=nonemptyopened]                           {$(\overline\varepsilon, \closed)$};
  \path (emptyopened) edge [bend left] node {$\close, \reset$} (emptyclosed);
  \path (emptyopened) edge [loop left] node {$\open$} (emptyopened);
  \path (nonemptyopened) edge node {$\close$} (nonemptyclosed);
  \path (nonemptyopened) edge node {$\open$} (emptyopened);
  \path (nonemptyopened) edge [pos=0.2, above] node {$\reset$} (emptyclosed);
  \path (emptyclosed) edge [loop right] node {$\close, \reset$} (emptyclosed);
  \path (emptyclosed) edge [bend left] node {$\open$} (emptyopened);
  \path (nonemptyclosed) edge [loop right] node {$\close$} (nonemptyclosed);
  \path (nonemptyclosed) edge [pos=0.2, right] node {$\open$} (emptyopened);
  \path (nonemptyclosed) edge [right] node {$\reset$} (emptyclosed);
\end{tikzpicture}
\end{center}
\caption{Possible configuration changes of a fixed memory. Note that by $\varepsilon$ and $\overline\varepsilon$, we denote an empty or non-empty memory content, respectively; the instruction $\unchanged$ is omitted. Moreover, the diagram only shows configuration changes caused by memory instructions (in particular, $\varepsilon$ can only change into $\overline\varepsilon$ by consuming transitions).}
\label{fig:MemConfChanges}
\end{figure}
\begin{example}\label{ex:tmfa1}
	Consider the following  $\TMFArej$~$M$ with two memories over $\Sigma=\{\ta,\tb\}$:
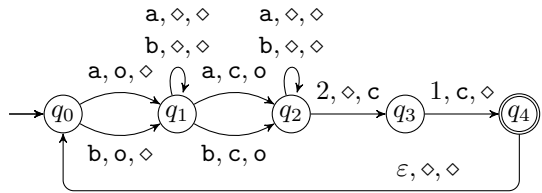
		\begin{center}
			\begin{tikzpicture}[node distance=15mm,on grid,>=stealth',auto, 
			state/.style={circle,draw=black,inner sep=1pt,minimum size=5mm}]
			
			\node[state,initial,initial by arrow,initial text={}]         (q0)                 				{$q_0$};
			\node[state]  (q1)   [ right=of q0]   {$q_1$};
			\node[state]  (q2)   [ right=of q1]   {$q_2$};
			\node[state]  (q3)   [ right=of q2]   {$q_3$};
			\node[state,accepting]	(q4)	[right = of q3] {$q_4$};	 	
			\path[->] 
			(q0) edge[bend left] node {$\ta,\open,\unchanged$} (q1)
			(q0) edge[bend right] node[below] {$\tb,\open,\unchanged$} (q1)
			(q1) edge[loop above] node[align=center] {$\ta,\unchanged,\unchanged$\\$\tb,\unchanged,\unchanged$} (q1)
			(q1) edge[bend left] node {$\ta,\close,\open$} (q2)
			(q1) edge[bend right] node[below] {$\tb,\close,\open$} (q2)
			(q2) edge[loop above] node[align=center] {$\ta,\unchanged,\unchanged$\\$\tb,\unchanged,\unchanged$} (q2)
			(q2) edge node[above] {$2,\unchanged,\close$} (q3)		
			(q3) edge node[above] {$1,\close,\unchanged$} (q4)		
			;
			\draw[->, rounded corners] (q4.south) -- +(0,-.75) -| node [pos=0.1,above]{$\emptyword,\unchanged,\unchanged$} (q0.south);
			\end{tikzpicture}
		\end{center}
This $\TMFA$ works as follows. First, in state $q_1$, we record a non-empty word over $\{\ta, \tb\}$ in the first memory, then, in state $q_2$, a non-empty word over $\{\ta, \tb\}$ in the second memory, and then, by moving through states $q_3$ and $q_4$, these words are repeated in reverse order by first recalling the second and then the first memory (note that in the transition from $q_3$ to $q_4$, an already closed memory is closed again, since according to Definition~\ref{TMFADefinition}, every memory that is recalled must be closed in the same transition). Due to the $\emptyword$-transition from $q_4$ to $q_0$, $M$ describes the Kleene-plus of such words, i.\,e., $\lang(M) = \lang(\alpha)$, where $\alpha = (\bind{x}{(\ta\ror\tb)^+}\bind{y}{(\ta\ror\tb)^+}\cdot\rr{y}\cdot\rr{x})^+ = (\{uvvu \mid u, v \in \{\ta, \tb\}^+\})^+$. 

		Note that each of the two memory recall transitions closes the respective memory. This is required by definition, as a transition can only recall a memory if it ensures that it is closed.

\end{example}

\begin{example}\label{ex:tmfa2}
	Consider the following  $\TMFArej$~$M$ with two memories over $\Sigma=\{\ta,\tb,\td\}$:
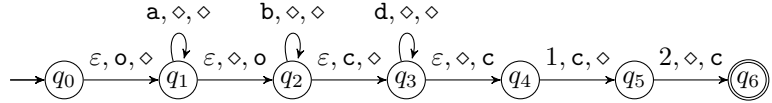
	\begin{center}
		\begin{tikzpicture}[node distance=15mm,on grid,>=stealth',auto, 
		state/.style={circle,draw=black,inner sep=1pt,minimum size=5mm}]
		
		\node[state,initial,initial by arrow,initial text={}]         (q0)                 				{$q_0$};
		\node[state]  (q1)   [ right=of q0]   {$q_1$};
		\node[state]  (q2)   [ right=of q1]   {$q_2$};
		\node[state]  (q3)   [ right=of q2]   {$q_3$};
		\node[state]  (q4)   [ right=of q3]   {$q_4$};
		\node[state]  (q5)   [ right=of q4]   {$q_5$};
		\node[state,accepting]	(q6)	[right = of q5] {$q_6$};	 	
		\path[->] 
		(q0) edge node[above] {$\emptyword,\open,\unchanged$} (q1)
		(q1) edge[loop above] node[] {$\ta,\unchanged,\unchanged$} (q1)
		(q1) edge node[above] {$\emptyword,\unchanged,\open$} (q2)
		(q2) edge[loop above] node[] {$\tb,\unchanged,\unchanged$} (q2)
		(q2) edge node[above] {$\emptyword,\close,\unchanged$} (q3)
		(q3) edge[loop above] node[] {$\td,\unchanged,\unchanged$} (q3)
		(q3) edge node[above] {$\emptyword,\unchanged,\close$} (q4)		
		(q4) edge node[above] {$1,\close,\unchanged$} (q5)		
		(q5) edge node[above] {$2,\unchanged,\close$} (q6)		
		;
		\end{tikzpicture}
	\end{center}
%	Clearly, $M$ is not deterministic, as the three states $q_1$ to $q_3$ each have a terminal transition and an $\emptyword$-transition.
	The behavior of $M$ can be described as follows: First, $M$ opens memory~1 and reads $\ta^i$, $i\geq 0$. After that, $M$ opens the second memory, reads $\tb^j$ (which is stored in \emph{both} memories), $j\geq 0$, closes the first memory, reads $\td^k$, $k\geq 0$, and closes the second memory. Hence, after reading $\ta^i \tb^j \td^k$, the first memory contains $\ta^i\tb^j$, and the second $\tb^j \td^k$. Finally, $M$ recalls memory~1 and then~2. Hence, $\lang(M) = \{ \ta^i \tb^j \td^k\ta^i \tb^{2j} \td^k \mid i,j,k\geq 0  \}$.
	
	Now, note that in each input word $w$, memory~2 is opened and closed after memory~1. Hence, if $j>0$, the areas in $w$ where the two memories are open overlap, instead of being nested. This cannot happen in a regex, as it is ensured from the syntax of variable bindings that these ``areas'' in the word are properly nested. For this reason, it seems impossible to express $\lang(M)$ with a regex with only two variables. But this does not mean that $\lang(M)$ is not a regex language, as $\lang(M)=\lang(\alpha)$ for  $\alpha = \bind{x}{\ta^*}\bind{y}{\tb^*}\bind{z}{\td^*}\cdot \rr{x}\rr{y}\cdot\rr{y}\rr{z}$. In other words, the key idea is expressing each memory with two variables (one for the overlapping parts of the memories, and one for each rest).\footnote{Proving that these overlaps can always be resolved is the main step in showing the equivalence of $\langcl(\RX)$ and $\langcl(\TMFA)$, which is provided in~\cite{sch:cha} (see also the discussion at the end of the proof of Theorem~\ref{thm:TMFAisMFA}).}
\end{example}

%Clearly, the $\TMFA$ of Examples~\ref{ex:tmfa1}~and~\ref{ex:tmfa2} are not deterministic (in Example~\ref{ex:tmfa1}, there are different transitions for the same state that consume the same symbol and in Example~\ref{ex:tmfa2}, there are states for which $\emptyword$-transitions exists in addition to other transitions). By minor changes of the $\TMFA$ of Example~\ref{ex:tmfa1}, a $\DTMFA$ can be easily constructed for the language $(\{u \td v \td v u \mid u, v \in \{\ta, \tb\}^+\})^+$, the details are left to the reader.

Next, we shall see that every $\TMFAacc$ can be transformed into an equivalent $\TMFArej$, which implies $\langcl(\TMFA) = \langcl(\TMFArej)$; thus, it follows from \cite{sch:cha} that $\TMFA$ characterize $\langcl(\RX)$.
The idea of this construction is as follows. Every memory $i$ is simulated by two memories $(i, 1)$ and $(i, 2)$, which store a (nondeterministically guessed) factorisation of the content of memory $i$. This allows us to guess and verify if a memory recall failure occurs, i.\,e., $(i, 1)$ stores the longest prefix that can be matched 
and $(i, 2)$ starts with the first mismatch. For correctness, it is crucial that every possible factorisation of the content of a memory $i$ can be guessed. \par
We first need the following definition. An $M \in \TMFA$ is in \emph{normal form} if no empty memory is recalled, no open memory is opened, no memory is reset, and, for every transition $\delta(q, b) \ni (p, s_1, \ldots, s_k)$, 
\begin{itemize}
\item if $b \neq \varepsilon$, then $s_i = \unchanged$, $1 \leq i \leq k$,
\item if $s_i \neq \unchanged$, for some $i$, $1 \leq i \leq k$, then $b = \varepsilon$ and $s_j = \unchanged$, for every $j$, $1 \leq j \leq k$, $i \neq j$.
\end{itemize}

\begin{proposition}
Any $\TMFA$ can be transformed into an equivalent $\TMFA$ in normal form.
\end{proposition}

\begin{proof}
An arbitrary $\TMFA$ can be changed into an equivalent one in normal form as follows. By introducing $\varepsilon$-transitions, we can make sure that every transition is of the form stated in the proposition. Furthermore, by adding states, we can keep track of the memory configurations (i.\,e., their status and whether or not they are empty; this simple technique is also explained in more detail in the proof of Theorem~\ref{closureComplementTheorem}). This allows us to replace transitions that are recalling an empty memory by $\varepsilon$-transitions. Furthermore,  transitions that open an open memory $i$ are replaced by transitions applying the memory instructions $\close$ and $\open$ in this order to memory $i$, and transitions that reset a memory $i$ are replaced by transitions applying the memory instructions $\close$, $\open$ and $\close$ in this order to memory $i$ (the correctness of this can be easily checked with the help of Figure~\ref{fig:MemConfChanges}). The $\TMFA$ is then in normal form and, by definition, these modifications do not change the accepted language. 
\end{proof}

Now, we can formally prove the claimed characterisation.

\begin{theorem}\label{thm:TMFAisMFA}
$\langcl(\TMFA) = \langcl(\TMFArej) = \langcl(\RX)$.
\end{theorem}

%\begin{theorem}\label{thm:TMFAisMFA}
%$\langcl(\TMFA) = \langcl(\TMFArej)$.
%\end{theorem}

\begin{proof}
We first note that $\langcl(\TMFArej) = \langcl(\RX)$ follows from \cite{sch:cha} (we briefly discuss this at the end of the proof). Since $\langcl(\TMFArej) \subseteq \langcl(\TMFA)$ and $\TMFA = \TMFArej \cup \TMFAacc$, it only remains to prove $\langcl(\TMFAacc) \subseteq \langcl(\TMFArej)$. To this end, let $M$ be a $\TMFAacc$ in normal form. 
%
%
%, it only remains to show $\langcl(\TMFArej) = \langcl(\TMFArej)$
%
%Let $M$ be a $\TMFA$ in normal form. If the state $\trapstate$ of $M$ is not accepting, then we can obtain an $\MFA$ $M'$ by removing the state $\trapstate$. Obviously, every computation of $M'$ on some $w \in \Sigma^*$ that leads to an accepting configuration can be carried out by $M$ in the same way and, furthermore, in every computation of $M$ on some $w \in \Sigma^*$ that leads to an accepting configuration, the state $\trapstate$ cannot be visited; thus, this computation can be carried out by $M'$ in the same way. This implies $L(M) = L(M')$.\par
%If $\trapstate$ is accepting, then the construction is more complicated. 
First, we replace every memory $i$, $1 \leq i \leq k$, by two memories $(i, 1)$ and $(i, 2)$ and we implement in the finite state control a list $(x_1, x_2, \ldots, x_k)$ with entries from $\Sigma \cup \{\varepsilon\}$, which initially satisfies $x_i = \varepsilon$, $1 \leq i \leq k$. Then, we change the transitions of $M$ such that the new memories $(i, 1)$ and $(i, 2)$ simulate the old memory $i$, i.\,e., memory $i$ stores some word $u$ if and only if memories $(i, 1)$ and $(i, 2)$ store $u_1$ and $u_2$, respectively, with $u = u_1u_2$. Moreover, the element $x_i$ always equals the first symbol of the content of memory $(i, 2)$. More precisely, this can be done as follows. Let $\delta(q, b) \ni (p, s_1, \ldots, s_k)$ be an original transition of~$M$.
\begin{itemize}
\item If $s_i = \open$ or $s_i = \close$, for some $i$, $1 \leq i \leq k$, then instead we open memory $(i, 1)$ or close memory $(i, 2)$, respectively.
\item If $b \in \Sigma$, then, for every open memory $(i, 1)$, we nondeterministically choose to close it and open memory $(i, 2)$ instead and set $x_i = b$. Then we read $b$ from the input and change to state $p$.
\item If $b \in \{1, 2, \ldots, k\}$, then we first recall memory $(b, 1)$ and then, for every open memory $(i, 1)$, we nondeterministically choose to close it and open memory $(i, 2)$ instead and set $x_i = x_b$. Then, we recall memory $(b, 2)$ and change to state $p$. 
\end{itemize}
All these modifications can be done by introducing intermediate states and using $\varepsilon$-transitions and the accepted language of $M$ does not change. \par
The automaton $M$ now stores some content $u$ of an original memory $i$ factorised into two factors $u_1$ and $u_2$ in the memories $(i, 1)$ and $(i, 2)$, respectively. For the sake of convenience, we simply say that $u$ is stored in $(i, 1) \cdot (i, 2)$ in order to describe this situation. Next, we show that if $u$ is stored in $(i, 1) \cdot (i, 2)$, then any way of how $u$ is factorised into the content of $(i, 1)$ and $(i, 2)$ is possible. More precisely, we show that, for every $w, u, u_1, u_2 \in \Sigma^*$ with $u_1 u_2 = u$, $M$ can reach state $p$ by consuming $w$ with $u$ stored in $(i, 1) \cdot (i, 2)$ if and only if $M$ can reach state $p$ by consuming $w$ with $u_1$ and $u_2$ stored in $(i, 1)$ and $(i, 2)$, respectively .\par
The \emph{if} part of this statement is trivial. We now assume that $M$ can reach state $p$ by consuming $w$ with $u$ stored in $(i, 1) \cdot (i, 2)$. This implies that we reach the situation that $(i, 1)$ is open, currently stores $u'_1$ and the next consuming transition consumes $u''_1 u'_2$, where $u_1 = u'_1 u''_1$ and $u_2 = u'_2 u''_2$ with $u'_2 \neq \varepsilon$. If $u''_1 = \varepsilon$, then $M$ can choose to close $(i, 1)$ and then open $(i, 2)$, which results in $u_1$ and $u_2$ being stored in $(i, 1)$ and $(i, 2)$, respectively. If, on the other hand, $u''_1 \neq \varepsilon$, then the next transition recalls memories $(j, 1)$ and $(j, 2)$ such that $u''_1 u'_2$ is stored in $(j, 1) \cdot (j, 2)$. If $u''_1$ and $u'_2$ are stored in $(j, 1)$ and $(j, 2)$, respectively, then $M$ first recalls $(j, 1)$, chooses to close $(i, 1)$ and open $(i, 2)$, and then recalls $(j, 2)$, which results in $u_1$ and $u_2$ being stored in $(i, 1)$ and $(i, 2)$. Consequently, we have to repeat this argument for memories $(j, 1)$ and $(j, 2)$, i.\,e., we have to show that it is possible that $u''_1 u'_2$ is stored in $(j, 1) \cdot (j, 2)$ in such a way that $u''_1$ is stored in $(j, 1)$ and $u'_2$ is stored in $(j, 2)$. Repeating this argument, we will eventually arrive at a memory that is not filled by any memory recalls; thus, we necessarily have the case $u''_1 = \varepsilon$.\par
Now, we turn $M$ into a $\TMFArej$ $M'$, i.\,e., the state $\trapstate$ becomes non-accepting, and, in addition, we add a new accepting state $q_t$ (simulating the old accepting $\trapstate$) with $\delta(q_t, x) = \{(q_t, \unchanged, \ldots, \unchanged)\}$, $x \in \Sigma$, and we change all ordinary transitions (i.\,e., transitions that are not recall failure transitions) of the former $M$ that lead to $\trapstate$ such that they now lead to $q_t$. 
%Finally, we add a new accepting state $q_t$, i.\,e., all the transitions leading to $\trapstate$ now lead to $q_t$, with $\delta(q_t, x) = \{(q_t, \unchanged, \ldots, \unchanged)\}$, $x \in \Sigma$, we can turn $M$ into an $\MFA$ $M'$.
%By replacing $\trapstate$ by a new accepting state $q_t$, i.\,e., all the transitions leading to $\trapstate$ now lead to $q_t$, with $\delta(q_t, x) = \{(q_t, \unchanged, \ldots, \unchanged)\}$, $x \in \Sigma$, we can turn $M$ into an $\MFA$ $M'$. 
Furthermore, we change this $M'$ such that for every memory recall, there is also the nondeterministic choice to only recall $(i, 1)$, then check whether $x_i$ does not equal the next symbol on the input and, if this is the case, enter state $q_t$. 
Obviously, this simulates the memory recall failure of $M$. \par
Every word accepted by $M$ without memory recall failures can be accepted by $M'$ in the same way, every word accepted by $M$ due to a recall failure can be accepted by $M'$ by guessing and simulating this memory recall failure. On the other hand, if $M'$ accepts a word with a simulated memory recall failure, then $M$ will accept this word by a proper memory recall failure, and if $M'$ accepts a word without a simulated memory recall failure, then, since $M' \in \TMFArej$, there is no memory recall failure in the computation and $M$ can accept the word by the same computation.\par
This completes the proof of $\langcl(\TMFAacc) \subseteq \langcl(\TMFArej)$.\par
We shall conclude this proof by briefly sketching why $\langcl(\TMFArej) = \langcl(\RX)$ holds. For an $\alpha \in \RX$, it is straightforward to obtain an equivalent $\TMFArej$: Transform $\alpha$ into a proper regular expression $\alpha_{\mathcal{R}}$ with $\lang(\alpha_{\mathcal{R}}) = \refl(\alpha)$ (by just renaming variable bindings and references), then transform $\alpha_{\mathcal{R}}$ into an equivalent $\NFA$~$M$, and finally interpret $M$ as a $\TMFArej$ by interpreting transition labels $[_{x}$, $]_{x}$ as memory instructions and transition labels $x$ as memory recalls. The other direction relies on first resolving overlaps of memories (i.\,e., the case that two memories store factors that overlap in the input word, see also Example~\ref{ex:tmfa2}) and then transforming the $\TMFArej$ $M$ into a proper regular expression for a ref-language that dereferences to $\lang(M)$, which can then directly be interpreted as a regex (due to the non-overlapping property of memories, which translates into a well-formed nesting of the parentheses $[_x$, $]_x$). This works in the same way as for the case of memory automata without trap-states (see~\cite{sch:cha} for details).
\end{proof}

A consequence of the proof is that $\TMFA$ inherits the $\NP$-hardness of the membership problem from $\RX$. We do not devote more attention to this, as we  focus on deterministic $\TMFA$.

\subsection{Deterministic $\TMFA$}\label{sec:DTMFA}

A $\TMFA$ is \emph{deterministic} (or a $\DTMFA$, for short) if $\delta$ satisfies $|\delta(q, b)| \leq 1$, for every $q \in Q$ and $b \in \Sigma \cup \{\varepsilon\} \cup \{1, 2, \ldots, k\}$ (for the sake of convenience, we then interpret $\delta$ as a partial function with range $Q  \times \{\open, \close, \reset, \unchanged\}^k$), and, furthermore, for every $q \in Q$, if $\delta(q, x)$ is defined for some $x \in \{1, 2, \ldots, k\} \cup \{\varepsilon\}$, then, for every $y \in (\Sigma \cup \{\varepsilon\} \cup \{1, 2, \ldots, k\}) \setminus \{x\}$, $\delta(q, y)$ is undefined.\footnote{Note that in~\cite{sch:cha} deterministic memory automata \emph{without} trap-state are considered.} Analogously to $\TMFA$, we partition  $\DTMFA$ into  $\DTMFAacc$ and $\DTMFArej$.\par
Clearly, the $\TMFA$ of Examples~\ref{ex:tmfa1}~and~\ref{ex:tmfa2} are not deterministic (in Example~\ref{ex:tmfa1}, there are different transitions for the same state that consume the same symbol and in Example~\ref{ex:tmfa2}, there are states for which $\emptyword$-transitions exists in addition to other transitions). By minor changes of the $\TMFA$ of Example~\ref{ex:tmfa1}, a $\DTMFA$ can be easily constructed for the language $(\{u \td v \td v u \mid u, v \in \{\ta, \tb\}^+\})^+$, the details are left to the reader.\par
The algorithmically most important feature of $\DTMFA$ is that their membership can be solved efficiently by running the automaton on the input word. However, for each processed input symbol, there might be a delay of at most $|Q|$ steps, due to $\emptyword$-transitions and recalls of empty memories, which leads to $O(|Q||w|)$. Removing such non-consuming transitions first, is possible, but problematic. In particular, recalls of empty memories depend on the specific input word and could only be determined beforehand by storing for each memory whether it is empty, which is too expensive. However, by an $O(|Q|^2)$ preprocessing, we can compute the information that is needed in order to determine in $O(k)$ where to jump if certain memories are empty, and which memories are currently empty can be determined on-the-fly while processing the input. This leads to a delay of only $k$, the number of memories:

%\begin{theorem}\label{thm:dtmfamembership}
%Given $M \in \DTMFA$ with $n$ states and $k$ memories, and $w \in \Sigma^*$, we can decide in time $O(n^2 + k|w|)$, whether or not $w \in \lang(M)$.
%\end{theorem}

%\todo[inline, color=blue!30]{MS: I restated this theorem in order to highlight the ``preprocessing'' aspect a little (see also the following corollary). I have no strong opinion here, so change back if you like the old version better.}
%\todomcom{How about this?}
%\begin{theorem}\label{thm:dtmfamembership}
%Given $M \in \DTMFA$ with $n$ states and $k$ memories, and $w \in \Sigma^*$, we can decide in time $O(n|w|)$ (or in time $O(k|w|)$ after an $O(n^2)$ preprocessing) whether or not $w \in \lang(M)$.
%\end{theorem}
\begin{theorem}\label{thm:dtmfamembership}
	Given $M \in \DTMFA$ with $n$ states and $k$ memories, and $w \in \Sigma^*$, we can decide  whether or not $w \in \lang(M)$
	\begin{itemize}
		\item in time $O(n|w|)$ without preprocessing, or
		\item in time $O(k|w|)$ after an $O(n^2)$ preprocessing.
	\end{itemize}
\end{theorem}
%\subsection{Proof of Theorem~\ref{thm:dtmfamembership}}
\begin{proof}
We first modify $M$ with respect to its $\emptyword$-transitions as follows. 
%For every $\emptyword$-transition that is followed by at least another $\emptyword$-transition as follows.
Let $p \in Q$ be a state with an $\emptyword$-transition that is followed by another $\emptyword$-transition. If $p$ is contained in a cycle $q_1, q_2, \ldots, q_n$ of $\emptyword$-transitions, we simply replace this cycle by a single state $q'$ (i.\,e., all incoming edges of any $q_i$, $1 \leq i \leq n$, then point to $q'$) that is accepting if and only if some $q_i$, $1 \leq i \leq n$, is (note that, since $M$ is deterministic, no $q_i$ has any other transition). Otherwise, there are states $q_1, q_2, \ldots, q_n$, $p = q_1$, with transitions $\delta(q_i, \emptyword) = (q_{i + 1}, s_{i, 1}, \ldots, s_{i, k})$, $1 \leq i \leq n-1$, such that $q_n$ has no $\emptyword$-transition. We can now remove the transition $\delta(q_1, \emptyword) = (q_{2}, s_{1, 1}, \ldots, s_{1, k})$ and add a transition $\delta(q_1, \emptyword) = (q_{n}, t_{1}, \ldots, t_{k})$, where, for every $j$, $1 \leq j \leq k$, $t_j$ is a memory instruction that has the same effect as applying instructions $s_{1, j}, s_{2, j}, \dots, s_{n-1, j}$ in this order. Moreover, if, for some $i$, $2 \leq i \leq n-1$, $q_i \in F$, then we define $q_1$ as accepting. By	applying this modification for every $\emptyword$-transition that is followed by another $\emptyword$-transition, we can modify $M$ such that no $\emptyword$-transition is followed by another $\emptyword$-transition. Hence, since $M$ is deterministic, there are at most $|Q|$ $\emptyword$-transitions and for each, we have to determine the states $q_1, q_2, \ldots, q_n$ and perform the modifications described above, which can be done in time $O(|Q|)$, as well. Consequently, the whole procedure can be carried out in $O(|Q|^2)$. 
	
	Next, we consider states with a memory recall transition. Similar as for states with $\emptyword$-transition, such states are followed by a (possibly empty) sequence of consecutive memory recall or $\emptyword$-transitions that either ends in a state with neither memory recall nor $\emptyword$-transition or eventually forms a loop. We first consider the case, where this sequence does not contain any $\emptyword$-transitions and does not form a loop. 
%	without any $\emptyword$-transition and without loop. 
Let $q_1$ be the state with memory recall transition and let $(q_1, \ell_1), (q_2, \ell_2), \ldots, (q_n, \ell_n), q_{n+1}$ be the sequence of the following states with consecutive memory recall transitions along with the memory that is recalled. More precisely, the transition from $q_{i}$ to $q_{i+1}$, $1 \leq i \leq n$, recalls $\ell_{i}$ and the last element $q_{n+1}$ is the first state without memory recall transition (and, by assumption, also without $\emptyword$-transition). We now contract this list by the following algorithm. Initially, let $A = \emptyset$. Then we move through the list from left to right and for every element $(q_i, \ell_i)$ (except for $q_{n+1}$), we proceed as follows. If $\ell_i \in A$, then we remove $(q_i, \ell_i)$ and if $\ell_i \notin A$, then we keep $(q_i, \ell_i)$ and add $\ell_i$ to $A$. Obviously, this results in a list $(p_1, r_1), \ldots, (p_{n'}, r_{n'}), q_{n+1}$ with $n' \leq k$. The idea is that if we move from left to right through this new list, it tells us which state to enter if the memory of the current memory recall is empty, i.\,e., if memory $r_1$ is non-empty, we recall it in state $p_1$, if memory $r_1$ is empty, we can directly jump to state $p_2$ and either recall $r_2$, if it is non-empty, or jump to $p_3$ otherwise, and so on. If all memories (that occur somewhere in the list) are empty, we end up in state $q_{n+1}$.\par
	In the presence of $\emptyword$-transitions, we simply ignore these and always only consider the next transition that recalls a memory, i.\,e., it is possible that for elements $(q_i, \ell_i)$ and $(q_{i + 1}, \ell_{i + 1})$ of the non-contracted list, there is an intermediate state $p$ with recall transition from $q_i$ to $p$ and $\emptyword$-transitions from $p$ to $q_{i + 1}$ (note that due to the construction from above, there are no consecutive $\emptyword$-transitions), but the contraction works in the same way. Moreover, if $(q_i, \ell_i)$ and $(q_j, \ell_j)$ (or $q_{n+1}$, the last element) with $i < j$ are consecutive elements of the contracted list (i.\,e., all elements $(q_{r}, \ell_{r})$, $i + 1 \leq r \leq j-1$, have been deleted by the algorithm), then we replace $(q_i, \ell_i)$ by $(q^{\textsc{acc}}_i, \ell_i)$ (where the marker $\textsc{acc}$ means that we can accept), if for some $r$, $i + 1 \leq r \leq j$, $q_{r} \in F$. Note that this is analogous to the modification from above, where we define states as accepting, if they are connected to an accepting state by a sequence of $\emptyword$-transitions, but here we cannot change acceptance of the actual states, since it depends on the current contents of memories, whether we can reach an accepting state by only recalls of empty memories or $\emptyword$-transitions.	\par
	If the sequence of memory recall transitions enters a loop, we construct the list only up to the first time a state is repeated, say $p$, and have $(p, \textsc{loop})$ as the last element of the list. Then we apply the contraction in the same way as before, where $(p, \textsc{loop})$ plays the role of $q_{n+1}$. 
	%As before, if all memories are empty, then we reach the element $(p, \textsc{loop})$, which tells us that the automaton has entered an infinite $\emptyword$-loop. 
	Similarly as before, we mark elements $(q_i, \ell_i)$ as accepting if a pair was removed that contained an accepting state.\par
	In addition to the states, we also store in the list the memory instructions that have to be applied in order to  jump to the next state (this can be done similiar as for contracting the $\emptyword$-transitions above). We construct such a list for every state with a memory recall transition. Every single list can be constructed in time $O(|Q|)$, so we need time $O(|Q|^2)$ in total. \par
	Now we check whether or not $w \in \lang(M)$ by running (the modified) $M$ on input $w$ in a special way. We first initialise a list $(1, \closed, \emptyword), (2, \closed, \emptyword), \ldots, (k, \closed, \emptyword)$ indicating that every memory is closed and empty. Then we simulate $M$ on input $w$ as follows. Every transition that consumes a single symbol as well as every $\emptyword$-transition is just carried out. Whenever a memory status is changed, we store this in the list and we also store whether a memory is currently empty or not (note that we have to know the current statuses in order to do this). When a memory is recalled in state $q$, then we move through the list stored for state $q$ until we find a recall of a memory that is currently non-empty, jump in the automaton to the corresponding state and apply the memory instructions. Whenever we reach an element $(q^{\textsc{acc}}_i, \ell_i)$ in the list, then we check whether the input has been fully consumed and if yes, we conclude $w \in \lang(M)$. If we reach in a list an element $(p, \textsc{loop})$, then we conclude $w \in \lang(M)$, if $p \in F$ and $w \notin \lang(M)$ otherwise. If in the computation the input has been completely consumed, then we conclude $w \in \lang(M)$ if and only if $M$ is in an accepting state. \par
	Since there are no consecutive $\emptyword$-transitions, every consumption of a single symbol from the input by a transition is done in constant time. Every consumption by a memory recall transition requires time $O(k)$, since we have to move through a list of size $O(k)$. Consequently, the total running time is $O(|Q|^2 + k|w|)$.
\end{proof}

Note that the preprocessing in the proof of Theorem~\ref{thm:dtmfamembership} is only required once, which implies the following corollary.
%so we can solve the membership for several words $w_i$ in $O(n^2 + k\sum|w_i|)$. 

\begin{corollary}
Given $M \in \DTMFA$ with $n$ states and $k$ memories, and words $w_i \in \Sigma^*$, $1 \leq i \leq \ell$, we can decide whether or not $w_i \in \lang(M)$, $1 \leq i \leq \ell$, in total time $O(n^2 + k\sum|w_i|)$. 
\end{corollary}

Moreover, if it is guaranteed that no empty memories are recalled, then membership can be solved in $O(n + |w|)$ (where $O(n)$ is needed in order to remove $\emptyword$-transitions). \par
Similar to $\DFA$, it is possible to complement $\DTMFA$ by toggling the acceptance of states. However, for $\DTMFA$, we have to remove $\emptyword$-transitions and recalls of empty memories. In particular, our construction uses the finite control to store whether memories are empty or not, which causes a blow-up that is exponential in the number of memories.

%\subsection{Proof of Theorem~\ref{closureComplementTheorem}}\label{app:complement}
We first extend the notion of completeness from $\DFA$ to $\DTMFA$, by saying that a $\DTMFA$ is \emph{complete} if, for every $q \in Q$, either $\delta(q, x)$ is defined, for every $x \in \Sigma$, or $\delta(q, i)$ is defined, for some $i$, $1 \leq i \leq k$, or $\delta(q, \varepsilon)$ is defined. This means that a complete $\DTMFA$ has, for every state, either exactly $|\Sigma|$ transitions (which are all consuming transitions, but not memory recall transitions), exactly one memory recall transition, or exactly one $\varepsilon$-transition. 

For deterministic automata, it is usually possible to apply the state complementation technique (i.\,e., toggling acceptance of states) in order to show closure under complement. However, we also need completeness and $\emptyword$-freeness, since otherwise it may happen that a word is not accepted because its computation gets stuck or enters an infinite $\emptyword$-loop and therefore is not entirely processed, which leads to a word which is accepted neither by the original nor by the complement automaton. The requirement of completeness and $\emptyword$-freeness is not a restriction for $\DTMFA$, since these properties can be achieved by classical techniques. However, recalling empty memories, which are special cases of $\emptyword$-transition, can cause the same problems and therefore we have to get rid of them as well. This can be done by storing in the finite-state control whether the memories are currently empty or non-empty and then treating recalls of empty memories as $\emptyword$-transitions and remove them along with the other $\emptyword$-transition in the classical way (note that the trick of handling empty memories that has been used in the context of Theorem~\ref{thm:dtmfamembership} cannot applied here, since the automaton needs to store the information for all possible runs on input words). 

 %(all these transformations are formally carried out in a detailed way in the Appendix).

We need a few more definitions: Let $\Gamma = \{\open, \close, \reset, \unchanged\}$ and let $\circledcirc$ be a binary operator on $\Gamma$ defined by $x \circledcirc y = y$, if $y \neq \unchanged$ and $x \circledcirc y = x$, if $y = \unchanged$. Furthermore, we extend $\circledcirc$ to $\Gamma^k$ by $(x_1, \ldots, x_k) \circledcirc (y_1, \ldots, y_k) = (x_1 \circledcirc y_1, \ldots, x_k \circledcirc y_k)$. We note that $\circledcirc$ is associative and some memory instructions $s_1, s_2, \ldots, s_n \in \Gamma$ applied to some memory in this order have the same result as the memory instruction $s_1 \circledcirc s_2 \circledcirc \ldots \circledcirc s_n$ (this can be easily verified with the help of Figure~\ref{fig:MemConfChanges}).\par
Next, we prove a sequence of propositions (that are all proved in a straightforward way by applying classical automata constructions):

\begin{proposition}\label{uniqueComputationProposition}
Let $M \in \DTMFA$. For every $w \in \Sigma^*$ and every configuration $c$ for $M$, there exists at most one configuration $c'$ with $c \vdash_M c'$. 
\end{proposition}

\begin{proof}
Let $c = (q, v, (u_1, r_1), \ldots, (u_k, r_k))$. If no $\delta(q, i)$, $1 \leq i \leq k$, is defined, then there is obviously at most one $c'$ with $c \vdash_m c'$. If $\delta(q, i) = (p, s_1, \ldots, s_k)$, for some $i$, $1 \leq i \leq k$, then either $v = u_i v'$, which implies that $c \vdash_M (p, v', (u'_1, r'_1), \ldots, (u'_k, r'_k))$, where the $(u'_j, r'_j)$, $1 \leq j \leq k$, are uniquely determined by $u_i$ and the $s_j$, $1 \leq j \leq k$, or $u_i$ is not a prefix of $v$, which implies that $c \vdash_M (\trapstate, v'', (u_1, r_1), \ldots, (u_k, r_k))$, where $v = v' v''$ and $v'$ is the largest common prefix of $v$ and $u_i$. In both cases, there is at most one configuration $c'$ with $c \vdash_M c'$.
\end{proof}

\begin{proposition}\label{removeEpsilonTransitionsProposition}
For every $M \in \DTMFA$ there exists an $\varepsilon$-free $M' \in \DTMFA$ with $\lang(M) = \lang(M')$.
\end{proposition} 

\begin{proof}
Let $M = (Q, \Sigma, \delta, q_0, F)$. For every $p \in Q$, if, for some $q \in Q$, $\delta(p, \varepsilon) = (q, s_1, \ldots, s_k)$, then we define $\mathcal{S}_{\varepsilon, 1}(p) = q$ and $\mathcal{M}(p, q) = (s_1, \ldots, s_k)$. For every $p \in Q$ and every $i$, $2 \leq i \leq |Q| - 1$, we define $\mathcal{S}_{\varepsilon, i}(p) = \mathcal{S}_{\varepsilon, 1}(\mathcal{S}_{\varepsilon, i-1}(p))$ and, if $\mathcal{S}_{\varepsilon, i}(p)$ is defined, we define (or redefine) $\mathcal{M}(p, \mathcal{S}_{\varepsilon, i}(p)) = \mathcal{M}(p, \mathcal{S}_{\varepsilon, i - 1}(p)) \circledcirc (s_1, \ldots, s_k)$, where $\delta(\mathcal{S}_{\varepsilon, i - 1}(p), \varepsilon) = (\mathcal{S}_{\varepsilon, i}(p), s_1, \ldots, s_k)$.\par
For every $p \in Q$ with $\delta(p, \varepsilon)$ defined, we now remove the $\varepsilon$-transitions as follows. Let $i$, $1 \leq i \leq |Q|-1$, be such that $\mathcal{S}_{\varepsilon, i}(p) = q$ and $\mathcal{S}_{\varepsilon, i + 1}(p)$ is undefined. Furthermore, let $\delta(q, x_j) = (t_j, s_{j,1}, \ldots, s_{j,k})$, $1 \leq j \leq \ell$, for some $\ell$ with $0 \leq \ell \leq |\Sigma|$, be all the transitions from $q$ (note that $\ell = 1$ and $x_1 \in \{1, 2, \ldots, k\}$ covers the case of a single memory recall transition and, furthermore, $x_j = \varepsilon$ is by definition not possible). We now add new transitions $\delta(p, x_j) = (t_j, s'_1 \circledcirc s_{j,1}, \ldots, s'_k \circledcirc s_{j,k})$, where $\mathcal{M}(p, q) = (s'_1, \ldots, s'_k)$. Then, we simply delete all $\varepsilon$-transitions (note that this may produce states that are not reachable anymore, which are deleted as well). It can be easily verified that this results in an $M' \in \DTMFA$ with $\lang(M) = \lang(M')$. 
\end{proof}

\begin{proposition}\label{makeCompleteProposition}
For every $M \in \DTMFA$ there exists a complete $M' \in \DTMFA$ with $\lang(M) = \lang(M')$.
\end{proposition} 

\begin{proof}
Let $M = (Q, \Sigma, \delta, q_0, F)$. We transform $M$ into $M'$ by adding a new non-accepting state $t$ with $\delta(t, x) = (t, \unchanged, \ldots, \unchanged)$, for every $x \in \Sigma$, and we add transitions for every state $q \in Q$ as follows. If $\delta(q, i)$ is undefined, for every $i$, $1 \leq i \leq k$, and $\delta(q, \varepsilon)$ is undefined, then, for every $x \in \Sigma$ with $\delta(q, x)$ undefined, we set $\delta(q, x) = (t, \unchanged, \ldots, \unchanged)$. On the other hand, if $\delta(q, i)$ is defined, for some $i$, $1 \leq i \leq k$, or $\delta(q, \varepsilon)$ is defined, then we do not add any transition. By definition, $M'$ is complete and, since $t$ is non-accepting, $\lang(M) = \lang(M')$. 
\end{proof}

\begin{remark}\label{epsilonfreeCompletenessRemark}
We note that the construction of the proof of Proposition~\ref{removeEpsilonTransitionsProposition} preserves completeness, i.\,e., if $M$ is a complete $\DTMFA$, then we obtain an equivalent complete $\DTMFA$ without $\varepsilon$-transitions. Moreover, the construction of the proof of Proposition~\ref{makeCompleteProposition} does not introduce $\varepsilon$-transitions; thus, it turns an $\varepsilon$-free $\DTMFA$ into an equivalent complete $\DTMFA$ that is still $\varepsilon$-free.
\end{remark}

We are now ready to show closure of $\langcl(\DTMFA)$ under complementation.

\begin{theorem}\label{closureComplementTheorem}
$\langcl(\DTMFA)$ is closed under complement.
\end{theorem}

\begin{proof}
Let $M = (Q, \Sigma, \delta, q_0, F) \in \DTMFA$. By Proposition~\ref{makeCompleteProposition}, we can assume that $M$ is complete. Due to Proposition~\ref{uniqueComputationProposition}, for any input $w$, there is a unique computation of $M$ on $w$. Hence, the idea is now to toggle the acceptance of all the states of $M$ in order to obtain a $\DTMFA$ that accepts $\overline{\lang(M)}$.
%The general idea is now to toggle the acceptance of all the states of $M$ in order to obtain a $\DTMFA$ that accepts $\overline{L(M)}$. 
However, this only works if $M$ is $\varepsilon$-free, since otherwise it is possible that some word $w \in \Sigma^*$ cannot be fully consumed by $M$ (for example, if it leads into a loop in which all transitions are $\varepsilon$-transitions and no state is accepting); thus, $w$ is neither accepted by $M$ nor by the $\DTMFA$ obtained by toggling the acceptance of states. While we can remove $\varepsilon$-transitions due to Proposition~\ref{removeEpsilonTransitionsProposition}, we encounter the problem that a memory recall transition with respect to an empty memory behaves just like an $\varepsilon$-transition and, thus, can cause the same problems. Hence, we first have to transform such memory recall transition into ordinary $\varepsilon$-transitions, which can then be removed according to Proposition~\ref{removeEpsilonTransitionsProposition}. \par
To this end, we modify $M$ such that the finite state control stores, for every $i$, $1 \leq i \leq k$, whether or not memory $i$ is open and whether or not memory $i$ stores the empty word. More precisely, we obtain an $M_1 \in \DTMFA$ by modifying $M$ as follows. Every state $q$ is replaced by $2^{2k}$ new states $[q, (r_1, c_1), \ldots, (r_k, c_k)]$, where $r_i \in \{\closed, \opened\}$, $c_i \in \{\varepsilon, \overline{\varepsilon}\}$, $1 \leq i \leq k$, and we change the transitions such that if $M_1$ reaches a configuration with state $[q, (r_1, c_1), \ldots, (r_k, c_k)]$, then, in the current configuration, for every $i$, $1 \leq i \leq k$, $r_i$ is the status of memory $i$ and memory $i$ is empty if and only if $c_i = \varepsilon$. For example, if $M_1$ is in state $[p, (r_1, c_1), \ldots, (r_k, c_k)]$ with $(r_i, c_i) = (\closed, \overline{\varepsilon})$ and $\delta(p, x) = (q, s_1, \ldots, s_k)$ with $x \in \Sigma$ and $s_i = \open$, then, if $x$ is the next symbol of the input, $M_1$ changes to a state $[q, (r'_1, c'_1), \ldots, (r'_k, c'_k)]$ with $(r'_i, c'_i) = (\opened, \overline{\varepsilon})$. We note that $M_1$ is still complete and deterministic. \par
Next, we change $M_1$ into $M_2$ by replacing, for every $i$, $1 \leq i \leq k$, every transition of the form $$\delta([p, (r_1, c_1), \ldots, (r_k, c_k)], i) = ([q, (r'_1, c'_1), \ldots, (r'_k, c'_k)], s_1, \ldots, s_k)$$ with $c_i = \varepsilon$ by an $\varepsilon$-transition $$\delta([p, (r_1, c_1), \ldots, (r_k, c_k)], \varepsilon) = ([q, (r'_1, c'_1), \ldots, (r'_k, c'_k)], s_1, \ldots, s_k)\,.$$ We note that $\lang(M_1) = \lang(M_2)$ and, since $M_1$ is deterministic, this only introduces $\varepsilon$-transitions, such that if $\delta(p, \varepsilon)$ is defined then, for every $y \in (\Sigma \cup \{1, 2, \ldots, k\})$, $\delta(q, y)$ is undefined. Consequently, $M_2$ is still deterministic and it never happens that an empty memory is recalled. Next, by Proposition~\ref{removeEpsilonTransitionsProposition}, we can transform $M_2$ into a complete $M_3 \in \DTMFA$ without $\varepsilon$-transitions (see Remark~\ref{epsilonfreeCompletenessRemark}) that still has the property that no empty memories are recalled. \par
Let $\overline{M} \in \DTMFA$ be obtained from $M_3$ by toggling the acceptance of the states, i.\,e., if $Q_3$ and $F_3$ are the sets of states and accepting states, respectively, of $M_3$, then $\overline{M}$ is obtained from $M_3$ by replacing $F_3$ by $Q_3 \setminus F_3$. Obviously, for every $w \in \Sigma^*$, both $M_3$ and $\overline{M}$, on input $w$, reach the same state and completely consume the input. This directly implies $\lang(\overline{M}) = \overline{\lang(M_3)}$.
\end{proof}

We next discuss expressive power: If there is a constant upper bound on the lengths of contents of memories that are recalled in accepting computations of an $M \in \DTMFA$, then memories can be simulated by the finite state control; thus, $\lang(M) \in\langcl(\REG)$. Consequently, if $\lang(M) \notin \langcl(\REG)$, there is a word $u v w$ that is accepted by recalling some memory with an arbitrarily large content $v$. Moreover, if $\trapstate$ is non-accepting, then no word can be accepted that contains $u$ as a prefix, but not $u v$, since this will cause a memory recall failure. Intuitively speaking, a $\DTMFArej$ for a non-regular language makes arbitrarily large ``jumps'': 
\begin{lemma}[Jumping Lemma]\label{lem:jump}
	Let $L\in\langcl(\DTMFArej)$. Then either $L$ is regular, or for every $m\geq 0$, there exist  $n\geq m$ and $p_n,v_n\in\Sigma^+$ such that
%	\begin{inparaenum}
%		\item  $|v_n|=n$,
%		\item $v_n$ is a factor of $p_n$,
%		\item $p_n v_n$ is a prefix of a word from $L$,
%		\item for all $u\in\Sigma^+$, $p_n u\in L$ only if $v_n$ is  a prefix of $u$.
%	\end{inparaenum}
\begin{enumerate}
	\item  $|v_n|=n$,
	\item $v_n$ is a factor of $p_n$,
	\item $p_n v_n$ is a prefix of a word from $L$,
	\item for all $u\in\Sigma^+$, $p_n u\in L$ only if $v_n$ is  a prefix of $u$.
\end{enumerate}
\end{lemma}

%\subsection{Proof of Lemma~\ref{lem:jump}}

%\begin{lemma}\label{lem:jump}
%	Let $L\in\langcl(\DTMFArej)$. Then either $L$ is regular, or for every $m\geq 0$, there exist an $n\geq m$, a $p_n\in\Sigma^*$ with $|p_n|\geq n$, and a $v_n\in\Sigma^+$ with $|v_n|=n$ such that
%	\begin{enumerate}
%		\item $v_n$ is a factor of $p_n$,
%		\item $p_nv_n$ is a prefix of a word from $L$,
%		\item $p_nu\notin L$ for all $u\in\Sigma^+$ such that $v_n$ is not a prefix of $u$.
%	\end{enumerate}
%\end{lemma}

\begin{proof}
	As $L\in\langcl(\DTMFArej)$, there exists an $M\in\DTMFArej$ with $\lang(M)=L$. If there is an $m\geq 0$ such that in every accepting run of $M$, each memory stores only a word of length at most $m$, then $L$ is regular (as we can rewrite $M$ into a $\DFA$ that stores the contents of the memories in its states). Likewise, if memories can store words of unbounded length, but are then never recalled, these memories can be eliminated, which also allows us to turn $M$ into a $\DFA$ for $L$. 
	
	Hence, if $L$ is not regular, $M$ has at least one memory $x$ such that for every $m\geq 0$, there is an accepting run of $M$ on a word $w$ during which $x$ stores a word of length $n\geq m$, and this memory is recalled with this content. Let $p_n$ be the part of the accepting run that $M$ has processed up to a state $q$ where it recalls $x$ at a point where this memory contains a word $v_n$ of length $n$. 
	%Let $v_n$ be this content of $x$ (hence, $|v_n|=n$). 
	As $v_n$ must have been consumed while processing $p_n$, $|p_n|\geq n$ holds, and $v_n$ must be a factor of $p_n$.
	
	If $M$ succeeds at recalling $x$ at this point (i.\,e., it consumes $v_n$), it can continue to accept $w$, which means that $p_nv_n$ is a prefix of $w\in L$. On the other hand, on an input $p_n u$ for some $u\in\Sigma^+$ such that $v_n$ is not a prefix of $u$, $M$ encounters a memory recall failure and rejects. As $M$ is deterministic, the recall transition for $x$ must be the only transition that leaves the state $q$. Hence, $p_n u\notin L$ for $u$ that do not have $v_n$ as prefix.
\end{proof}

The Jumping Lemma is a convenient tool for proving that languages cannot be accepted by a $\DTMFArej$, which shall be illustrated by some examples.

\begin{example}\label{ex:jump}
	Let $L\df \{ww\mid w\in\Sigma^*\}$ with $|\Sigma| \geq 2$, which is well-known to be not regular. Assume  $L\in\langcl(\DTMFArej)$ and choose $m\df 1$. Then there exist $n\geq 1$ and $p_n,v_n\in\Sigma^*$  that satisfy the conditions of Lemma~\ref{lem:jump}. Choose  $a\in\Sigma$ that is not the first letter of $v_n$, and define $u\df a p_n a$. Then $v_n$ is not a prefix of $u$, but  $p_n u = (p_n a)^2\in L$, which is a contradiction. 
\end{example}
\begin{example}\label{ex:jump2}
	Let $L\df\{\ta^i\tb\ta^j \mid i >j\geq 0\}$. Using textbook methods, it is easily shown that $L$ is not regular. Now, assuming that $L\in\langcl(\DTMFArej)$, choose $m\df 4$. Then there exist $n\geq 4$ and $p_n,v_n\in\Sigma^+$  that satisfy the conditions of Lemma~\ref{lem:jump}. As $p_n v_n$ is a prefix of a word in $L$, either $p_n = \ta^i$ or $p_n = \ta^i\tb\ta^j$ with $i,j\geq 0$ (and $i\geq 4$ or $i+j\geq 3$). In the first case, consider $u\df \tb\ta$. Then $p_n u=\ta^i\tb\ta$ with $i\geq 4$; hence, $p_n u \in L$. But $u$ starts with $\tb$, and $v_n$ is a factor of $p_n = \ta^i$, which leads to a contradiction, as $v_n$ cannot be a prefix of $u$. For the second case,  let $u\df \ta$. As $p_n v_n$ is a prefix of a word in $L$, and as $|v_n|=n$, $i > j + n \geq j+4$ must hold. Hence, $p_n u = \ta^i \tb \ta^{j+1}$, and $p_n u \in L$, which, as $v_n$ is not a prefix of $u$, leads again to a contradiction.
\end{example}

For unary languages, there is an alternative to Lemma~\ref{lem:jump} that is easier to apply and that characterizes unary $\DTMFArej$-languages. It is built on the following definition: A  language $L\subseteq \{\ta\}^*$ is an \emph{infinite arithmetic progression} if  $L=\{\ta^{bi+c} \mid i\geq 0\}$  for some $b\geq 1$, $c\geq 0$. 
\begin{lemma}\label{lem:finap}
	Let $L\in\langcl(\DTMFArej)$ be an infinite language with $L\subseteq\{\ta\}^*$.  The following conditions are equivalent:
%	\begin{inparaenum}
%		\item\label{finap1} $L$ is regular.
%		\item\label{finap2} $L$ contains an  infinite arithmetic progression.
%		\item\label{finap3} There is $b\geq 1$ such that, for every $n\geq 0$, $\ta^{bi+c_n} \in L$ for some $c_n \geq 0$ and all $0\leq i \leq n$. 
%	\end{inparaenum}
	\begin{enumerate}
		\item\label{finap1} $L$ is regular.
		\item\label{finap2} $L$ contains an  infinite arithmetic progression.
		\item\label{finap3} There is $b\geq 1$ such that, for every $n\geq 0$, there exists some  $c_n \geq 0$ with  $\ta^{bi+c_n} \in L$ for  all $0\leq i \leq n$. 
	\end{enumerate}
\end{lemma}
%\todomcom{Made the last condition clearer.}

%\subsection{Proof of Lemma~\ref{lem:finap}}

%\begin{lemma}\label{lem:finap}
%	Let $L\in\langcl(\DTMFArej)$ be an infinite language with $L\subseteq\{\ta\}^*$. Then the following conditions are equivalent:
%	\begin{enumerate}
%		\item\label{finap1} $L$ is regular.
%		\item\label{finap2} $L$ contains an  infinite arithmetic progression.
%		\item\label{finap3} There is a $b\geq 1$ such that, for every $n\geq 0$, $\ta^{bi+c_n} \in L$ for some $c_n \geq 0$ and $0\leq i \leq n$. 
%	\end{enumerate}
%\end{lemma}
\begin{proof}
	We show that~\ref{finap1} implies~\ref{finap2}, which implies~\ref{finap3}, which implies~\ref{finap1}. The first two of these steps are simple: Assume that $L$ is regular. Every regular language over a single letter alphabet can be expressed as a finite union of arithmetic progressions (cf., e.\,g., Chrobak~\cite{chr:fin,chr:fin-err}). As $L$ is infinite, it must contain an infinite arithmetic progression. But if $L$ contains an infinite arithmetic progression $\ta^{ib+c}$, then the third condition is satisfied by definition.
	
	The step from~\ref{finap3} to~\ref{finap1} is more involved. Before we prove this, note that there are unary languages (which are not $\DTMFA$-languages), for which condition~\ref{finap3} does not imply the existence of an infinite arithmetic progression, see Example~\ref{ex:lex} below.
	
	Assume that $L\subseteq\{\ta\}^*$ is infinite, $L\in\langcl(\DTMFArej)$, and condition~\ref{finap3} is met for some $b\geq 1$. By definition, there is an $M\in\DTMFArej$ with $\lang(M)=L$. As $M$ is deterministic, each of its states can have at most one outgoing transition; and as $L$ is infinite, each state must have exactly one outgoing transition. Hence, like a $\DFA$ for a unary language (see e.\,g.\ the proof of Theorem~\ref{thm:drxunary}, in particular Figure~\ref{lollipopFigure}), $M$ consists of a chain and a cycle. Let $m$ be the number of accepting states on the cycle, and let $k$ be the number of memories that are accessed in the cycle (by recalling them, or by performing memory instructions).
	
%	\todomcom{Fixed the issue you remarked on.}
	Now consider an $n> m(b+1)^k$ such that there exists some $c_n$ with $w_i\df \ta^{bi+c_n} \in L$ for all $0\leq i \leq n$, and reading $w_0=\ta^{c_n}$ takes $M$ into the cycle (as condition~\ref{finap3} holds for all $n$, such a $c_n$ exists for every $n$ that is sufficiently large). Let $q$ be the accepting state that is reached by $w_0$. 

In the following, by an \emph{iteration} of the cycle, we mean the situation that $M$ is in state $q$ and then consumes input symbols until it reaches $q$ for the next time. The iteration of the cycle that starts after having fully consumed $w_0$ is called iteration $1$. Now, for every $j \geq 1$, we define a function $\vec{v}_j\colon \{1,\ldots,k\}\to \mathbb{N}$ that describes the content of each memory after completing iteration $j$. \par
In the remainder of the proof, we show that there is a constant upper bound for the values $\vec{v}_j(x)$, $1 \leq x \leq k$, $j \geq 1$. Note that if the length of the content of each memory is bounded, then $M$ can be rewritten into an equivalent $\DFA$ that simulates all memories in its states. Hence, $L$ must be a regular language, which shows that condition~\ref{finap3} implies condition~\ref{finap1}.\par
As $M$ has to accept all words $w_i$ with $0\leq i\leq n$, and as each iteration of the cycle can accept only $m$ words, we know that $M$ has to perform at least $I\df \frac{n}{m}> (b+1)^k$ iterations of the cycle in order to accept $w_n$. During these iterations, $M$ cannot consume more than $\ta^b$ between each pair of accepting states -- otherwise, $M$ would skip at least one of the $w_i$ (as $M$ is deterministic, the run for $w_n$ must be an extension of each run for a $w_i$ with $i<n$). In particular, this means that each memory that is recalled during these iterations cannot contain more than $\ta^b$; thus, there are only $b+1$ possible contents for each memory. Furthermore, as $M$ is deterministic, we know that each memory that is not recalled during these iterations will not be recalled during any later iterations of the cycle, which means that it can be removed from the cycle (and, as the chain is of finite length, it can also be removed from the chain). Hence, without loss of generality, we can assume that $\vec{v}_j(x) \leq b$, $1 \leq x \leq k'$, $1 \leq j \leq I$, where the cycle contains exactly the memories $1,\ldots,k$.

As $I>(b+1)^k$ and as there are only $(b+1)^k$ possible choices of $\vec{v}_j$, there exist $j,j'$ with $0\leq j' < j' \leq I$ and $\vec{v}_j = \vec{v}_{j'}$. As $M$ is deterministic, this allows us to conclude 	$\vec{v}_{j+l} = \vec{v}_{j'+l}$ for all $l \geq 0$. In other words, the sequence of transitions from iteration $j$ to iteration $j'$ will be repeated forever, using exactly the same memory contents, which means that $\vec{v}_l(x) \leq b$ for all $l\geq 0$ and all $1\leq x \leq k$. As explained above, this concludes the proof.
\end{proof}

Just like the Jumping Lemma, Lemma~\ref{lem:finap} can be used to prove  $\DTMFArej$-inexpressibility for unary languages. See the following examples.

\begin{example}\label{ex:lex}
	We define a $\Lex\subset\{\ta\}^*$ together with its complement $\clex$ in the following way: First, add $\ta$ to $\Lex$, then add the two words $\ta^2 $ and  $\ta^3$ to $\clex$, and the three words $\ta^4$ to $\ta^6$ to $\Lex$, and so on. In other words, in each step $i$, we add the next $i$ words of $\{\ta\}^*$ to one of the languages; namely $\Lex$ if $i$ is odd, and $\clex$ if $i$ is even. Then $\Lex$ satisfies condition~\ref{finap3} of Lemma~\ref{lem:finap}, but it does not contain any  infinite arithmetic progression. Hence, $\Lex\notin\langcl(\DTMFArej)$; and  $\clex\notin\langcl(\DTMFArej)$ follows analogously.
\end{example}

\begin{example}
	Let $\alpha\df \bind{x}{\ta\ta^+}(\rr{x})^+$ (this regex is also known as ``Abigail's expression''~\cite{abigail} in the PERL community). Then $\lang(\alpha)=\{\ta^{mn}\mid m,n\geq 2\}$. In other words, $\alpha$ generates the language of all $\ta^{i}$ such that $i$ is a composite number (i.\,e., not a prime number). As $\lang(\alpha)$ is not regular and contains the arithmetic progression $2i+4$,  Lemma~\ref{lem:finap} yields $\lang(\alpha)\notin\langcl(\DTMFArej)$. 
\end{example}
The following result is a curious consequence of Lemma~\ref{lem:finap}:
\begin{proposition}\label{prop:crash}
Over unary alphabets, we have $$\langcl(\DTMFArej)\cap\langcl(\DTMFAacc)=\langcl(\REG).$$
\end{proposition}

%\subsection{Proof of Proposition~\ref{prop:crash}}	
\begin{proof}
	We first observe that $\langcl(\REG) \subseteq \langcl(\DTMFArej)\cap\langcl(\DTMFAacc)$ holds by definition. Next, we assume that there is a non-regular language $L\in(\langcl(\DTMFArej)\cap\langcl(\DTMFAacc))$ over $\{\ta\}^*$. In particular, this implies that both $L$ and its complement $\overline{L}\df \{\ta\}^*\mdif L$ are infinite and, furthermore, by Theorem~\ref{closureComplementTheorem}, $L\in\langcl(\DTMFAacc)$ implies $\overline{L}\in\langcl(\DTMFArej)$. Since $\overline{L}\in\langcl(\DTMFArej)$ is a non-regular $\DTMFArej$ language, Lemma~\ref{lem:jump} allows us to conclude that for every $m\geq 0$, there exist an $n\geq m$ and a $p_n\geq n$ such that $\ta^i \in \overline{L}$ for all $p_n\leq i < p_n+n$. Hence, $\overline{L}$ contains finite arithmetic progressions of unbounded length; and as $\overline{L}$  is infinite, Lemma~\ref{lem:finap} states that $\overline{L}$ is regular, which is a contradiction. 
	%
	%	
	%	Due to Theorem~\ref{closureComplementTheorem}, the assumption $L\in\langcl(\DTMFAacc)$ implies $\overline{L}\in\langcl(\DTMFArej)$. Lemma~\ref{lem:jump} allows us to conclude that for every $m\geq 0$, there exist an $n\geq m$ and a $p_n\geq n$ such that $\ta^i\notin L$ for all $p_n\leq i < p_n+n$. Hence, $\overline{L}$ contains  finite arithmetic progressions of unbounded length; and as  $\overline{L}$  is infinite, Lemma~\ref{lem:finap} states that $\overline{L}$ is regular. Hence, $L$ is also regular.
\end{proof}

%!TEX root=det_SIAM.tex
\section{Deterministic Regex}\label{sec:detregex}
In order to define deterministic regex as an extension of deterministic regular expressions, we first extend the notion of a \emph{marked alphabet} that is commonly used for the latter: For every alphabet $A$, let $\markpos{A}\df \{a_{(n)}\mid a\in A, n\geq 1\}$. For every $\alpha\in\RX$, we define $\markpos{\alpha}$ as a regex that is obtained by taking $\alphar$ (the proper regular expression over $\Sigma\cup\Xi\cup\Gamma$ that generates the ref-language $\refl(\alpha)$), and  marking  each occurrence of $\chi\in(\Sigma\cup \Xi\cup \Gamma)$  by a unique number (to make this well-defined, we assume that the markings start at 1 and are increased stepwise). For example, if $\alpha\df \bind{y}{(\ta\ror\rr{x})^* \cdot(\emptyword\ror \tb \cdot \ta)}\cdot\rr{y}$, then  $\markpos{\alpha}=\vop{y}_{(1)}(\ta_{(2)}\ror x_{(3)})^* \cdot(\emptyword\ror \tb_{(4)} \cdot \ta_{(5)})	\vcl{y}_{(6)} \cdot y_{(7)}$. We also use these markings in the  ref-words: For example, $\vop{y}_{(1)}\ta_{(2)}\ta_{(2)} x_{(3)} \ta_{(2)}\vcl{y}_{(6)} y_{(7)}\in\refl(\markpos{\alpha})$. 

Before we explain this definition and use it to define deterministic regex, we first discuss the special case of deterministic  regular expressions: A proper regular expression $\alpha$ is \emph{not} deterministic if there exist words $u,v_1,v_2\in \markpos{\Sigma}^*$, a terminal $a\in\Sigma$ and positions $i\neq j$ such that $u a_{(i)}v_1$ and $u a_{(j)} v_2$ are elements of $\lang(\markpos{\alpha})$ (see e.\,g.\ \cite{bru:one,gro:det}). Otherwise, it is a \emph{deterministic proper regular expression} (or, for short, just \emph{deterministic regular expression}).

The intuition behind this definition is based on the Glushkov construction for the conversion of regular expressions into finite automata, as a  regular expression $\alpha$ is deterministic if and only if its \emph{Glushkov automaton} $\glush(\alpha)$ is deterministic. Given a regular expression $\alpha$, we define  $\glush(\alpha)$ in the following way: First,  we use the marked regular expression $\markpos{\alpha}$ to construct its \emph{occurrence graph} $G_{\markpos{\alpha}}$, a directed graph that has a source node $\src$, a sink node $\snk$, and one node for each $a_{(i)}$ in $\markpos{\alpha}$.\footnote{Most literature, like~\cite{bru:one}, defines the occurrence graph only implicitly by using sets $\gfirst$, $\glast$, and $\gfollow$, which correspond to  the edge from $\src$, the edges to $\snk$, or to the other edges of the graph, respectively. The explicit use of a graph is taken from the  $k$-occurrence automata by Bex et al.~\cite{bex:lea}. We shall see that an  advantage of graphs  is that they can be easily extended by adding memory actions to the edges.} The edges are constructed in the following way: Each node $a_{(i)}$ has an incoming edge from $\src$ if $a_{(i)}$ can be the first letter of a word in $\lang(\markpos{\alpha})$, and an outgoing edge to $\snk$ if it can be the last letter of such a word. Furthermore, for each factor $a_{(i)}b_{(j)}$ that occurs in a word of $\lang(\markpos{\alpha})$, there is an edge from $a_{(i)}$ to $b_{(j)}$. As a consequence, there is a one-to-one-correspondence between marked words in $\lang(\markpos{\alpha})$ and paths from $\src$ to $\snk$ in $G_{\markpos{\alpha}}$. To obtain $\glush(\alpha)$, we directly interpret $G_{\markpos{\alpha}}$ as $\NFA$ over $\Sigma$: The source $\src$ is the starting state, each node $a_{(i)}$ is a state $q_i$, and an edge from $a_{(i)}$ to $b_{(j)}$ corresponds to a transition from $q_i$ to $q_j$ when reading $b$. The sink $\snk$ does not become a state; instead, each node with an edge to $\snk$ is a final state (hence, $\glush(\alpha)$ contains the source state, and  one state for every terminal in $\alpha$). This interpretation allows us to treat occurrence graphs as an alternative notation for a subclass of $\NFA$ (namely those where the starting state is not reachable from other states, and for each state~$q$, there is a characteristic terminal $a_q$ such that all transitions to $q$ read $a_q$). When doing so, we usually omit the occurrence markings on the nodes in graphical representations. 

Intuitively, $\glush(\alpha)$ treats each terminal of $\alpha$ as a state. Recall that $\alpha$ is not deterministic if there exists words  $u a_{(i)}v_1$ and $u a_{(j)} v_2$ in $\lang(\markpos{\alpha})$ with $i\neq j$. This corresponds to the situation where, after reading $u$, $\glush(\alpha)$ has to decide between states $a_{(i)}$ and $a_{(j)}$ for the input letter $a$. 
\begin{example}
Let $\alpha\df (\emptyword\ror ((\ta\ror\tb)^+\ta))$. Then $\markpos{\alpha}=(\emptyword\ror ((\ta_{1}\ror\tb_{2})^+\ta_{3}))$, and the Glushkov automaton $\glush(\alpha)$ of $\alpha$ is defined as follows:
	
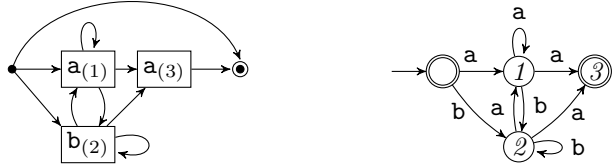
\begin{center}
		\begin{tikzpicture}[node distance=10mm,on grid,>=stealth',auto, 
		state/.style={rectangle,draw=black,inner sep=2pt,minimum size=4mm}]
		\node         (start)                 				{};
		\node[state]  (a1)     [right=of start] 				{\vphantom{$\tb$}$\ta_{(1)}$};
		\node[state]  (b2)     [below=of a1]     	{$\tb_{(2)}$};
		\node[state]  (a3)     [right=of a1]		 	{\vphantom{$\tb$}$\ta_{(3)}$};
		\node 				(end)   [right=of a3]	   	{};
		
		\draw[fill=black] (start) circle (0.5mm);
		\draw[fill=black] (end) circle (0.5mm);
		\draw 						(end) circle (1mm);
		
		\draw[->]  (start.center) -- (a1);
		\draw[->]  (start.center) -- (b2.north west);
		\draw[->]  (a1) edge[bend left] (b2);
		\draw[->]  (b2) edge[bend left] (a1);
		\draw[->] (a1) -- (a3);
		\draw[->] (b2) -- (a3);
		\draw[->]  (a3) -- (end);
		\draw[->] (a1) edge [loop above] (a1);
		\draw[->] (b2) edge [loop right] (b2);
		\path[->]  (start.center) edge [bend left=70] (end.north);
		\end{tikzpicture}\hspace{15mm}
		\begin{tikzpicture}[node distance=10mm,on grid,>=stealth',auto, 
		state/.style={circle,draw=black,inner sep=0pt,minimum size=4mm}]
		\node[state,initial by arrow,initial text={},accepting] (q_0){}; 
		\node[state](q_1) [right=of q_0] {1};
		\node[state,accepting](q_3) [right=of q_1] {3}; 
		\node[state](q_2) [below=of q_1] {2};	
		
		\path[->] 
		(q_0) edge node[near start] {$\ta$} (q_1)
		(q_0) edge[bend right=10] node[below,very near start] {$\tb$} (q_2)
		(q_1) edge[bend left=10] node[right, pos=0.4]  {$\tb$} (q_2)
		(q_2) edge[bend left=10] node[ pos=0.5] {$\ta$} (q_1)
		(q_2) edge[bend right=10] node[right, pos=0.5] {$\ta$} (q_3)
		(q_1) edge[] node {$\ta$} (q_3)
		(q_1) edge[loop above] node[] {$\ta$} (q_1)
		(q_2) edge[loop right] node[right] {$\tb$} (q_2)
		;
		\end{tikzpicture}
	\end{center}	
	To the left, $\glush(\alpha)$ is represented as an occurrence graph, to the right in standard $\NFA$ notation. Then $\glush(\alpha)$ and $\alpha$ are both not deterministic: For $\glush(\alpha)$, consider state~1; for $\alpha$, consider $u = \ta_{(1)}$, $v_1 = \ta_{(3)}$, $v_2 = \emptyword$, and the words $u \ta_{(1)} v_1$ and $u \ta_{(3)} v_2$.
\end{example}
As shown in~\cite{bru:one}, $\langcl(\DREG)\subset\langcl(\REG)$ (also see \cite{cze:dec, lu:dec}, or Lemma~\ref{lem:drxvsdrmfa} below). Like for determinism of  regular expressions, the key idea behind our definition of deterministic regex  is that a matcher for the expression treats terminals (and variable references) as states. Then an expression is deterministic if the current symbol of the input word always uniquely determines the next state and all necessary variable actions. For  regular expressions, non-determinism can only occur when the matcher has to decide between two occurrences of the same terminal symbol; but as regex also need to account for non-determinism that is caused by variable operations or references, their definition of  non-determinism is  more complicated.
\begin{definition}\label{def:drx}
		An $\alpha\in\RX$ is \emph{not deterministic} if there exist $\rho_1, \rho_2\in \refl(\markpos{\alpha})$ such that any of the following conditions is met for some  $r,s_1,s_2\in(\markpos{\Sigma}\cup\markpos{\Xi}\cup\markpos{\Gamma})^*$ and $\gamma_1,\gamma_2\in\markpos{\Gamma}^*$:
		\begin{compactenum}
			\item\label{def:drx:c1} $\rho_1 = r\cdot \gamma_1\cdot  a_{(i)}\cdot  s_1$ and  $\rho_2 = r\cdot  \gamma_2\cdot  a_{(j)} \cdot s_2$  with $a\in\Sigma$ and $i\neq j$, 
			\item\label{def:drx:c2} $\rho_1 = r\cdot  \gamma_1\cdot  x_{(i)}\cdot  s_1$ and $\rho_2 = r\cdot \gamma_2\cdot  \chi_{(j)}\cdot s_2$ with $x\in\Xi$, $\chi\in(\Sigma\cup\Xi)$ and $i\neq j$, 
			\item\label{def:drx:c3} $\rho_1=r\cdot \gamma_1\cdot \chi_{(i)}\cdot  s_1$ and $\rho_2=r\cdot \gamma_2\cdot \chi_{(i)}\cdot  s_2$ with  $\chi\in(\Sigma\cup\Xi)$ and $\gamma_1\neq \gamma_2$,
			\item\label{def:drx:c4} $\rho_1=r\cdot \gamma_1$ and  $\rho_2=r\cdot \gamma_2$ with $\gamma_1\neq \gamma_2$.
		\end{compactenum}
	Otherwise, $\alpha$ is \emph{deterministic}. We use $\DRX$ to denote the set of all deterministic regex, and define $\DREG\df \DRX\cap\REG$ as the set of deterministic regular expressions (as discussed below, Condition~\ref{def:drx:c1} of  Definition~\ref{def:drx} covers the case of deterministic proper regular expression).
\end{definition}
\begin{example}\label{ex:drx}
	We define $\alpha_1\df (\bind{x}{\ta}\ror\ta)$, $\alpha_2\df (\ta\ror\rr{x})$, $\alpha_3\df (\bind{x}{\emptyword}\ror\emptyword)\ta$ and $\alpha_4\df (\bind{x}{\emptyword}\ror\emptyword)$. None of these regex are deterministic, as each $\alpha_i$ meets the $i$-th condition of Definition~\ref{def:drx}. We discuss this for $\alpha_1$: Observe $\markpos{\alpha}_1=(\vop{x}_{(1)}\ta_{(2)}\vcl{x}_{(3)})\ror\ta_{(4)}$. Then choosing $\rho_1 = \vop{x}_{(1)}\ta_{(2)}\vcl{x}_{(3)}$ and $\rho_2 = \ta_{(4)}$, with $r=\emptyword$, $\gamma_1 = \vop{x}_{(1)}$, $s_1 = \vcl{x}_{(3)}$, and $\gamma_2=s_2=\emptyword$ shows  the condition is met.
	
	Let $\beta_1\df \bind{x}{(\ta\ror\tb)^*}\tc\cdot \rr{x}$ and $\beta_2\df \bigl(\bind{x}{\rr{y}}\bind{y}{\rr{x}\cdot \ta}\bigr)^*$. Both regex are deterministic, with 
	$\lang(\beta_1)\df \{w\tc w\mid w\in\{\ta,\tb\}^*\}$ and  $\lang(\beta_2)=\{\ta^{n^2}\mid n\geq 0\}$ (see Example~\ref{refWordExample}). 
\end{example}

Condition~\ref{def:drx:c1} of  Definition~\ref{def:drx} describes cases where non-determinism is caused by two occurrences of the same terminal ($\gamma_1$ and $\gamma_2$ are included for cases like $\alpha_1$ in Example~\ref{ex:drx}). If restricted to  regular expressions,  it is equivalent to the usual definition of deterministic  regular expressions. 
Condition~\ref{def:drx:c2} expresses that the matcher has to decide between a variable reference and any other symbol (this may be a terminal, a different variable or the same variable, but with a different index); while in condition~\ref{def:drx:c3}, the symbol is unique, but there is a non-deterministic choice between variable operations. 
Finally, condition~\ref{def:drx:c4} describes cases where the behavior of variables is non-deterministic after the end of the word (while one could consider this edge case deterministic, this choice simplifies recursive definitions). In conditions~\ref{def:drx:c3} and~\ref{def:drx:c4}, the definition not only requires that it is clear which variables are reset, but also  that it is clear which part of the regex acts on the variables. Hence,  $(\bind{x}{\emptyword}\ror\bind{x}{\emptyword})$ is also not deterministic. This is similar to the notion of strong determinism for regular expressions, see~\cite{gel:reg}. As one might expect, some non-deterministic regexes define $\DRX$-languages:
\begin{example}\label{ex:chunk}
	Let $\Sigma =\{\mathtt{0},\mathtt{1}\}$ and $\alpha\df \mathtt{1}^+ \bind{x}{\mathtt{0}^*} (\mathtt{1}^+ \rr{x})^* \mathtt{1}^+$. This regex was introduced by Fagin et al.~\cite{fag:spa}, who call its language the ``uniform-0-chunk language''. Obviously, $\alpha$ is not deterministic (in fact, it satisfies conditions~\ref{def:drx:c1}, \ref{def:drx:c2}, and \ref{def:drx:c3} of Definition~\ref{def:drx}). Nonetheless, it is possible to express $\lang(\alpha)$ with the deterministic regex  $\mathtt{1} \bigl(\mathtt{1}^+ \ror \bigl(0\bind{x}{0^*} 1^+(0\cdot\rr{x}\cdot1^+)^*\bigr)\bigr)$. 
\end{example}
We now discuss the conversion from $\DRX$ to $\DTMFArej$, which generalizes the Glushkov construction of $\glush(\alpha)$ for  regular expressions. The core idea is extending the occurrence graph to a \emph{memory occurrence graph} $G_{\markpos{\alpha}}$, which has two crucial differences: First, instead of only considering terminals, each terminal and each variable reference of a regex $\alpha$ becomes a node. Second, each edge is labeled with a ref-word from $\mGamma^*$ that describes the memory actions (hence, there can be multiple edges from one node to another). In analogy to the occurrence graph,  each memory occurrence graph can be directly interpreted as an $\emptyword$-free~$\TMFArej$. 
\begin{theorem}\label{thm:glushkov}
Let $\alpha\in\RX$, and let  $n$ denote the number of occurrences of terminals and variable references in $\alpha$.
We can construct an $n+2$ state $\TMFArej$~$\glush(\alpha)$ with $\lang(\glush(\alpha))=\lang(\alpha)$ that is deterministic if and only if $\alpha$ is deterministic. 
In  time $O(|\Sigma||\alpha|n)$, the algorithm either
\begin{inparaenum}
	\item computes $\glush(\alpha)$   if $\alpha$ is deterministic, or 
	\item detects  that $\alpha$ is not deterministic.
\end{inparaenum}
\end{theorem}

%\subsection{Proof of Theorem~\ref{thm:glushkov}}
\begin{proof}
	
	We construct $\glush(\alpha)$ by first constructing a graph $G_{\markpos{\alpha}}$ from the marked regex $\markpos{\alpha}$. As $G_{\markpos{\alpha}}$ is a generalization of the occurrence graphs for proper regular expressions, we call this the \emph{memory occurrence graph}. Analogously to proper regular expressions, this graph can be directly interpreted as an $\glush(\alpha)\in\TMFArej$ that is deterministic if and only if $\alpha$ is deterministic.

	\subparagraph*{Memory occurrence graph $\ograph{\alpha}$:} 	Given a marked regex $\markpos{\alpha}$, we define a memory occurrence graph $\ograph{\alpha}\df(\onodes{\alpha},\oedges{\alpha})$ with a source node $\src$, a sink node  $\snk$, and one node for each marked variable reference or terminal. The labeled edges are of the form $(u,\glab,v)$, where $u,v\in \onodes{\alpha}$, and each label $\glab$ is a marked ref-word $\glab\in\mGamma^*$. We use marked ref-words instead of unmarked ref-words to fulfill the promise that $\glush(\alpha)$ is deterministic if and only if $\alpha$ is deterministic. If $\alpha$ has $n$ occurrences of variable references and terminals, $\glush(\alpha)$ has $n+2$ states: the initial state, the state~$\trapstate$ for memory recall failures, and one state for each of the~$n$ occurrences in $\alpha$.
	
	If we only want to construct an algorithm that turns a deterministic regex into a $\DTMFArej$ and rejects non-deterministic regexes, we can use  unmarked edge labels instead (see the section at the end of this proof).
	
	When interpreting $\ograph{\alpha}$ as a $\TMFA$~$\glush(\alpha)$, we first remove the markings from the edge labels, and interpret these  as memory actions of a $\TMFA$, e.\,g., $\vop{x}$ corresponds to opening the memory for $x$. In order to simplify the construction, we take into account that  different ref-words over $\Gamma$ can have the same net effect on variables, and can be represented by the same single transition in  a $\TMFA$. For example, $\vcl{x}\vop{x}\vcl{x}\vop{x}$ and $\vcl{x}\vop{x}$ and $\vop{x}$ all have the same effect as performing $\open$ on the memory for $x$. Following this intuition, given a ref-word $\glab\in\Gamma^*$, we define the \emph{net variable action of $\glab$} as a function $\netvar{\glab}\colon \Xi\to\{\open,\close,\reset,\unchanged\}$, where for each $x\in\Xi$, $\netvar{\glab}(x)\df \unchanged$ if no element of  $\Gamma_x \df \{\vop{x}, \vcl{x}\}$ occurs in $\glab$, and $\netvar{\glab}(x)\df  \open$ if the rightmost occurrence of an element of $\Gamma_x$ is a $\vop{x}$. Furthermore, if the rightmost occurrence of an element of $\Gamma_x$ is $\vcl{x}$, we define $\netvar{\glab}(x)\df \reset$ if $\glab$ contains $\vop{x}$, and $\netvar{\glab}(x)\df \close$ otherwise. In the construction further down, we also consider concatenations of labels. We observe the following for all $\glab,\glab_1,\glab_2\in\Gamma^*$ and all $x\in \Xi$: If $\netvar{\glab}(x)=\unchanged$, then $\netvar{\glab_1\cdot\glab}(x)=\netvar{\glab_1}(x)$ and $\netvar{\glab\cdot\glab_2}=\netvar{\glab_2}(x)$. If $\netvar{\glab_2}(x)\in\{\open,\reset\}$, then $\netvar{\glab_1\cdot \glab_2}(x)=\netvar{\glab_2}(x)$.
	
	We also use the following notion of minimal representations: For all $\glab\in\Gamma^*$ and $x\in \Xi$, we define  $\gmin_x(\glab)\in \Gamma^*$ by $\gmin_x(\glab)\df\vop{x} $ if $\netvar{\glab}(x)=\open$, 
	$\gmin_x(\glab)\df\vcl{x} $ if $\netvar{\glab}(x)=\close$,
	$\gmin_x(\glab)\df\vop{x}\vcl{x} $ if $\netvar{\glab}(x)=\reset$, and 
	$\gmin_x(\glab)\df\emptyword $ if $\netvar{\glab}(x)=\unchanged$.
	% $$
	% \gmin_x(\glab)\df \begin{cases}
	% \vop{x} & \text{ if $\netvar{\glab}(x)=\open$},\\
	% \vcl{x} & \text{ if $\netvar{\glab}(x)=\close$},\\
	% \vop{x}\vcl{x} & \text{ if $\netvar{\glab}(x)=\reset$},\\
	% \emptyword & \text{ if $\netvar{\glab}(x)=\unchanged$,}
	% \end{cases}
	% $$
	For any $\glab\in\Gamma^*$, its minimal representation $\gmin(\glab)$ is defined as any concatenation of all $\gmin_x(\glab)$ for all $x\in \Xi$ (as $\netvar{\glab}\neq\unchanged$ holds only for finitely many $x\in\Xi$, this is not problematic). In other words, for each $\glab\in\Gamma^*$, $\gmin(\glab)$ is one of the shortest words in $\Gamma^*$ that satisfies $\netvar{\gmin(\glab)}=\netvar{\glab}$. %For any set $G\subseteq \Gamma^*$, we define $\gmin(G)\df \{\gmin(\glab)\mid \glab\in G\}$. 
By $\unmark\colon(\mSigma\cup\mXi\cup\mGamma)\to(\Sigma\cup\Xi\cup\Gamma)^*$, we denote the morphism that removes the markings from marked letters.
		
	By using $\netvar{}$, we can directly interpret a memory occurrence graph $\ograph{\alpha}$ as a $\TMFA$ $\glush(\alpha) \df (Q,$ $\Sigma, \delta, \src, F)$, analogously to the occurrence graph for proper regular expressions. The components of $\glush(\alpha)$ are obtained as follows: First, we rename the variables such that $\ograph{\alpha}$ contains exactly the variables  $\{1,\ldots,k\}$ for some $k\geq 0$ (hence, for each $1\leq i\leq k$, there is a variable $x_i\in\var(\alpha)$ such that $x_i$ is represented by $i$). We then define 
	
\begin{align*}
		Q&\df (\onodes{\alpha}\mdif\{\snk\})\cup\{\trapstate\},\\
		F&\df \{u\in Q \mid (u,\glab,\snk)\in \oedges{\alpha} \text{ for some $\glab$} \}.
	\end{align*}
	In other words, all nodes except $\snk$ are states, and all nodes that have an edge to $\snk$ are final states (as in the occurrence graph).  Following this intuition, each edge $(u,\glab,v)$ with $v \neq \snk$ corresponds to a transition from state $u$ to state $v$, while performing the memory actions of $\netvar{\glab}(x)$ on each $x\in \var(\alpha)$. In order to allow recursive applications of the construction, each edge $(u,\glab,\snk)$ not only marks that $u$ is an accepting state, but also that the memory actions of $\glab$ need to be performed before accepting.
	Formally, we define $\delta$ to include exactly the following transitions:
	\begin{enumerate}
		\item If $(u,\glab,a_{(i)})\in \oedges{\alpha}$ with $a\in\Sigma$, then $(a_{(i)},s_1,\ldots,s_k)\in\delta(u,a)$.
		\item If $(u,\glab,x_{(i)})\in \oedges{\alpha}$ with $x\in \var(\alpha)$, then $(x_{(i)},s_1,\ldots,s_k)\in\delta(u,x)$,
	\end{enumerate}
	where for each $1\leq i \leq k$, $s_i \df \netvar{\unmark(\glab)}(x_i)$ (unless the transition recalls memory $i$; then we choose $s_i \df \close$ as required by Definition~\ref{TMFADefinition}). 
%Recall that $\unmark\colon(\mSigma\cup\mXi\cup\mGamma)\to(\Sigma\cup\Xi\cup\Gamma)^*$\todo{MS: $\unmark$ is not previously defined!} is the morphism that removes the markings from marked letters.
 As we shall see, in order to satisfy the condition that $\glush(\alpha)$ is deterministic \emph{only if} $\alpha$ is deterministic, we need to slightly adapt this definition.
	
	Following this interpretation, we say that a memory occurrence graph $\ograph{\alpha}$ is \emph{not deterministic} if there exists a $u\in \onodes{\alpha}$  such that any of the following conditions is met:
	\begin{enumerate}
		\item\label{dmoa:c1} $\oedges{\alpha}$ contains edges $(u,\glab_1,a_{(i)})$ and $(u,\glab_2,a_{(j)})$ with $i\neq j$ and $a\in\Sigma$.
		\item\label{dmoa:c2} $\oedges{\alpha}$ contains edges $(u,\glab_1,x_{(i)})$ and $(u,\glab_2,\chi_{(j)})$ with $i\neq j$, $x\in \Xi$, $\chi\in(\Xi\cup\Sigma)$,
		\item\label{dmoa:c3} $\oedges{\alpha}$ contains edges $(u,\glab_1,\chi_{(i)})$ and $(u,\glab_2,\chi_{(i)})$ with $\glab_1\neq \glab_2$ and $\chi\in(\Xi\cup\Sigma)$,
		\item\label{dmoa:c4} $\oedges{\alpha}$ contains edges $(u,\glab_1,\snk)$ and $(u,\glab_2,\snk)$ with $\glab_1\neq \glab_2$.				
	\end{enumerate}
	Otherwise, we call $\ograph{\alpha}$ \emph{deterministic}. It is easily seen that if $\ograph{\alpha}$ is deterministic, $\glush(\alpha)$ is also deterministic. For the other direction, we need to account for two problems: First, it is possible that two labeled ref-words $\glab_1$ and $\glab_2$ map to the same memory action $\netvar{\unmark(\glab_1)}=\netvar{\unmark(\glab_2)}$, e.\,g., $\glab_1 = \vop{x}_{(1)}\vcl{x}_{(2)}$ and $\glab_2= \vop{x}_{(3)}\vcl{x}_{(4)}\vop{x}_{(5)}\vcl{x}_{(6)}$, which  can occur in regex like $\alpha_1\df \bigl(\bind{x}{\emptyword} \ror (\bind{x}{\emptyword}\bind{x}{\emptyword})\bigr)\ta$.  Second, as $\glush(\alpha)$ has no $\emptyword$-transitions, it does not model the difference between distinct edges to $\snk$, as they appear when converting regex like $\alpha_2\df(\emptyword\ror \bind{x}{\emptyword})$. As $\DTMFA$ cannot detect explicitly that the end of the input has been reached, they cannot simulate the memory actions of edges to $\snk$, which means that the construction ignores this. 
	
	In both cases, the accepted language is correct; but this has the side effect that the resulting  $\glush(\alpha)$ is deterministic, although $\alpha$ and $\ograph{\alpha}$ are not. 
	Hence, to ensure that $\glush(\alpha)$ is deterministic only if $\alpha$ is deterministic, we proceed as follows: If $\ograph{\alpha}$ contains any of these edges, we pick any transition $\delta(q, b) \ni (p, s_1, \ldots, s_k)$ with $b \in \Sigma \cup \{1, 2, \ldots, k\}$. We then  add a new state $p_{\cheat}$, a transition $\delta(q, b) \ni (p_{\cheat}, s_1, \ldots, s_k)$, and $p_{\cheat}$ has the same outgoing transitions as $p$. If we want to construct an algorithm that rejects non-deterministic regex, we can simply omit this technical crutch, and detect these cases in the construction of $\ograph{\alpha}$ as discussed below.
	
	\subparagraph*{Constructing $G_{\markpos{\alpha}}$:}
	We now define $G_{\markpos{\alpha}}=(\onodes{\alpha},\oedges{\alpha})$ recursively. %During the construction, we also observe that the application of recursive rules preserves the following invariant: For each edge $(u,\glab,v)$, $u=\src$ implies $\netvar{\glab}(x)\neq \close$ and $v=\snk$ implies $\netvar{\glab}(x)\neq\open$ for all $x\in \Xi$. In particular, this implies that for each edge $(\src,\glab,\snk)$, $\netvar{\glab}(x)\in\{\reset,\unchanged\}$ for all $x\in \Xi$. 
	\begin{enumerate}
		\item \descitem{Empty word:} If $\markpos{\alpha} = \emptyword$, we define 
		\begin{align*}
			\onodes{\alpha}&\df\{\src,\snk\},\\ 
			\oedges{\alpha}&\df\{(\src,\emptyword,\snk)\}.	
		\end{align*}
		
		%			\begin{tikzpicture}[node distance=10mm,on grid,>=stealth',auto, 
		%			state/.style={rectangle,draw=black,inner sep=0pt,minimum size=4mm}]
		%			\node         (start)                 				{};
		%			\node 				(end)   [right=of start]	   	{};
		%			
		%			\draw[fill=black] (start) circle (0.5mm);
		%			\draw[fill=black] (end) circle (0.5mm);
		%			\draw 						(end) circle (1mm);
		%			
		%			\draw[->]  (start.center) -- (end.west);
		%			\end{tikzpicture}
		This case is completely straightforward: An edge from $\src$ to $\snk$ is how occurrence graphs model $\emptyword$, and the marking $\emptyword$ means that this transition performs no memory actions.
		
		\item \descitem{Terminals and variable references:} If $\markpos{\alpha} = \chi_{(i)}$  with $\chi\in(\Sigma\cup\Xi)$, we define
		\begin{align*}
			\onodes{\alpha}&\df\{\src,\chi_{(i)},\snk\},\\
			\oedges{\alpha}&\df\{(\src,\emptyword,\chi_{(i)}),(\chi_{(i)},\emptyword,\snk)\}.
		\end{align*}
		
		Similar to the case for $\emptyword$, this models that the terminal is read, or that a variable  reference is processed, by recalling the appropriate memory.
		
		\item \descitem{Variable bindings:} If $\markpos{\alpha} = (\vop{x}_{(i)}\markpos{\beta}\vcl{x}_{(j)})$ with $x\in \Xi$, we define $	\onodes{\alpha}\df V_{\markpos{\beta}}$ and
		\begin{align*}
			\oedges{\alpha} \df 
			\:&\{(\src,\vop{x}_{(i)}\cdot\glab_{\tin},v) \mid (\src,\glab_{\tin},v)\in E_{\markpos{\beta}}, v\neq \snk\}\\
			&\cup 	\{(u,\glab,v) \mid (u,\glab,v)\in E_{\markpos{\beta}}, u\neq \src, v\neq \snk\}\\
			&\cup 	\{(u,\glab_{\tout}\cdot \vcl{x}_{(j)},\snk) \mid (u,\glab_{\tout},\snk)\in E_{\markpos{\beta}}, u\neq \src\}\\
			&\cup \{(\src,\vop{x}_{(i)}\cdot\glab_{\emptyword}\cdot\vcl{x}_{(j)},\snk) \mid (\src,\glab_{\emptyword},\snk)\in E_{\markpos{\beta}}\},
		\end{align*}
		Less formally, we take the memory occurrence graph for $\beta$ and add opening (and closing) of $x$ to all edges from $\src$ (and to $\snk$, respectively); while all other edges remain unchanged.
		Note that for edges from $\src$ to $\snk$, we could also use $\glab_{\emptyword}\cdot\vop{x}\vcl{x}$ or $\vop{x}\vcl{x}\cdot \glab_{\emptyword}$, as  by Definition~\ref{def:rx}, $\bind{x}{\beta}$ is only a regex if $x\notin\var(\beta)$, which implies that no marked $\vop{x}$ or $\vcl{x}$ occurs in $\glab_{\emptyword}$.
		\item \descitem{Disjunction:} If $\markpos{\alpha} = (\markpos{\beta} \ror \markpos{\gamma})$, we define $\onodes{\alpha}\df V_{\markpos{\beta}}\cup V_{\markpos{\gamma}}$ and $\oedges{\alpha}\df E_{\markpos{\beta}}\cup E_{\markpos{\gamma}}.$ 
		
		As the markings define  a one to one correspondence between the nodes in $\onodes{\alpha}$ and the terminals and the variable references in $\alpha$, we know that $V_{\markpos{\beta}}\cap V_{\markpos{\gamma}}=\{\src,\snk\}$. Therefore, the resulting memory occurrence graph $\ograph{\alpha}$ computes the union of $G_{\markpos{\beta}}$ and $G_{\markpos{\gamma}}$. %Furthermore, as there are no new edges, the invariant remains unaffected.
		
		\item \descitem{Concatenation:} If  $\markpos{\alpha} = (\markpos{\beta} \cdot \markpos{\gamma})$, we define
		\begin{align*}
			\onodes{\alpha}&\df V_{\markpos{\beta}}\cup V_{\markpos{\gamma}},\\
			\oedges{\alpha}&\df 
			\{(u,\glab,v)\mid (u,\glab,v)\in E_{\markpos{\beta}}, v\neq \snk\} \\
			&\cup \{(u,\glab,v)\mid (u,\glab,v)\in E_{\markpos{\gamma}}, u\neq \src\} \\
			&\cup \{(u,(\glab_{1}\cdot \glab_{2}),v)\mid (u,\glab_1,\snk)\in E_{\markpos{\beta}}, (\src,\glab_2,v)\in E_{\markpos{\gamma}}\},
		\end{align*}
		Again, we use the fact that $V_{\markpos{\beta}}\cap V_{\markpos{\gamma}}=\{\src,\snk\}$. The memory occurrence graph $G_{\markpos{\alpha}}$ first simulates $G_{\markpos{\beta}}$, until the latter would accept by processing an edge $(u,\glab_1,\snk)\in E_{\markpos{\beta}}$. Instead of following this edge to $\snk$, $G_{\markpos{\alpha}}$ then starts its simulation of $G_{\markpos{\gamma}}$, by picking any edge $(\src,\glab_2,v)\in E_{\markpos{\gamma}}$, which is merged with $(u,\glab_1,\snk)$ 
		into a single edge from $u$ to $v$, and its label is $(\glab_1\cdot \glab_2)$. Hence, it is easy to see that $G_{\markpos{\alpha}}$ computes the concatenation of  $G_{\markpos{\beta}}$ and  $G_{\markpos{\gamma}}$. 
		
		%The only edges that could affect the invariant are the new edges that were formed from edges $(u,\glab_1,\snk)\in E_{\markpos{\beta}}$ and  $(\src,\glab_2,v)\in E_{\markpos{\gamma}}$; and this only if $u=\src$ or $v=\snk$. If $u=\src$, the invariant requires $\netvar{\glab_1}(x)\in\{\reset,\unchanged\}$ for all $x\in \Xi$. As it also demands $\netvar{\glab_2}(x)\neq \close$, we obtain $\netvar{\glab_1\cdot\glab_2}(x)\in\{\open,\reset,\unchanged\}$ (i.\,e., $\netvar{\glab_1\cdot\glab_2}(x)\neq\close$) for all $x\in \Xi$. Likewise, if $v=\snk$, we can conclude $\netvar{\glab_1\cdot \glab_2}(x)\neq \open$ for all $x\in \Xi$. Hence, the invariant is preserved.
		
		\item \descitem{Kleene plus:} Assume $\markpos{\alpha} = \markpos{\beta}^+$. 
		This case requires some additional definitions.  
		Let  $N_{\emptyword}$ denote the set of all $\glab$ with $(\src,\glab,\snk)\in E_{\markpos{\beta}}$, and let $N^{(*)} \df \{\gmin(\glab)\mid \glab\in N^*_{\emptyword}\}$, where we assume that the elements of $N^{(*)}$ have some arbitrary markings (as we shall see, this definition matters only for non-deterministic regex, which means that we do not need markings to detect non-determinism). 
		
		Note that, as $N_{\emptyword}$ is finite, there are only finitely many $x\in\Xi$ such that $\netvar{\glab}(x)\neq \unchanged$ for a $\glab\in N_{\emptyword}$, which implies that $N^{(*)}$ is finite. In order to avoid hiding non-determinism in some very special cases, we assume that $N^{(*)}$ always contains at least two elements (this is possible without loss of generality, as we can always add some $\glab^2$ for a $\glab\in N^{(*)}$ without changing the behavior).
		We now define $\onodes{\alpha}\df V_{\markpos{\beta}}$, as well as
		\begin{align*}
			\oedges{\alpha}\df\:&E_{\markpos{\beta}} 
			\cup \{(\src,\hat{\glab}\cdot \glab_{\tin},v)\mid (\src,\glab_{\tin},v)\in E_{\markpos{\beta}}, \hat{\glab}\in N^{(*)}\} \\
			&\cup \{(u,\glab_{\tout}\cdot \hat{\glab},\snk)\mid (u,\glab_{\tout},\snk)\in E_{\markpos{\beta}}, \hat{\glab}\in N^{(*)} \} \\
			&\cup \{(u,\glab_{\tout}\cdot \hat{\glab}\cdot \glab_{\tin},v)\mid (u,\glab_{\tout},\snk)\in E_{\markpos{\beta}}, (\src,\glab_{\tin},v)\in E_{\markpos{\beta}}, \hat{\glab} \in N^{(*)}\}.
		\end{align*}
		Similar to the construction for concatenation, the idea is that $G_{\markpos{\alpha}}$ simulates $G_{\markpos{\beta}}$; and whenever the latter could accept by taking an edge to $\snk$, the former can loop back to the beginning. The only difficult part is when $G_{\markpos{\beta}}$ contains edges from $\src$ to $\snk$ with memory actions. As the Kleene plus allows us to use an arbitrary amount of these edges before taking an edge from $\src$ or to $\snk$, we need to include $N_{\emptyword}^*$ in the functions. This set is generally infinite; but it can be compacted to the finite set $N^{(*)}$.  
		
		%	As $\netvar{\glab}(x)\in\{\reset,\unchanged\}$ for all $x\in \Xi$ and all $\glab\in N_{\emptyword}$, the invariant is preserved. 
		
		For deterministic regex, this construction collapses to a far simpler case that does not use $N^{(*)}$: First, note that if $\alpha$ is deterministic and contains $\beta^+$,  $(\src,\glab,\snk)\in \oedges{\beta}$  implies $\glab=\emptyword$, as  otherwise, $\beta$ would satisfy condition~\ref{def:drx:c4} of Definition~\ref{def:drx}, and $\alpha$ would satisfy condition~\ref{def:drx:c3} or~\ref{def:drx:c4}. Hence, if $\alpha$ is deterministic, we can assume that $N_{\emptyword}=\{\emptyword\}$ or $N_{\emptyword}=\emptyset$; both cases lead to $N^{(*)}= \{\emptyword\}$. This allows us to use the following simplified definition:
		\begin{align*}
			\oedges{\alpha}&\df E_{\markpos{\beta}} \cup \{(u,\glab_{\tout}\cdot\glab_{\tin},v)\mid (u,\glab_{\tout},\snk)\in E_{\markpos{\beta}}, (\src,\glab_{\tin},v)\in E_{\markpos{\beta}}\}.
		\end{align*}
		Hence, when constructing $\ograph{\alpha}$ inductively, we first check if $\oedges{\beta}$ contains an edge $(u,\glab,v)$ with $\glab\neq \emptyword$. If this is the case, we can reject $\alpha$ as not deterministic. Otherwise, we use this simplified definition.
	\end{enumerate}
	
	\subparagraph*{Correctness and determinism:} The correctness of the construction is easily seen by a lengthy but straightforward induction, using the explanations provided with the definitions above. In particular, note that if $\alpha$ does not contain a Kleene plus (or contains a Kleene plus and is deterministic), each path from $\src$ to $\snk$ through $\ograph{\alpha}$ corresponds to a marked ref-word from $\refl(\alpha)$, and vice versa. If $\alpha$ is not deterministic and contains a Kleene plus, the correspondence is a little bit less strict, as ref-words $\gamma\in\Gamma^+$ are compressed to the equivalent $\gmin(\gamma)$. 
	
	To see that $\glush(\alpha)$ is deterministic if and only if $\alpha$ is deterministic, recall that we established above that $\glush(\alpha)$ is deterministic if and only if $\ograph{\alpha}$ is deterministic. Hence, it suffices to show that determinism in $\ograph{\alpha}$ is equivalent to determinism in $\alpha$. 	But this follows immediately from our observation that there is a one-to-one correspondence between paths  in $\ograph{\alpha}$ and the marked ref-words in $\refl(\alpha)$, and the fact that each node $\chi_{(i)}\in\onodes{\alpha}\mdif\{\src,\snk\}$ corresponds to the same $\chi_{(i)}$ in $\markpos{\alpha}$. Thus, if $\ograph{\alpha}$ satisfies a condition $i$ for non-determinism, $\alpha$ satisfies the same condition $i$ in Definition~\ref{def:drx}, and vice versa.
	
	\subparagraph*{Complexity:} Given a regex $\alpha$, let $n$ denote the number of occurrences of terminals and variable references in $\alpha$. We examine two steps of the computation: Computing $\ograph{\alpha}$, and converting it to $\glush(\alpha)$.
	
	For the first step, observe that  $\ograph{\alpha}$ has $n+2$ nodes, and if $\alpha$ is deterministic, each node has at most $\min(n,|\Sigma|)$ outgoing edges, which means that we can bound this number with $|\Sigma|$.  Hence, if $\alpha$ is deterministic, $\ograph{\alpha}$ can be computed in time $O(|\Sigma||\alpha| n)$ by directly following the recursive definition of $\ograph{\alpha}$: If $\alpha$ is represented as a tree, it has at most $|\alpha|$ nodes, which means that the recursive rules have to be applied $O(|\alpha|)$ times. Each rule application requires the creation of at most $O(|\Sigma|n)$ edges, each of which uses a concatenation.
	
	For the conversion, we need to process each edge $(u,\glab,v)\in\oedges{\alpha}$, and compute its function $\netvar{\unmark(\glab)}$. From the recursive definition, we can immediately conclude that $|\glab|\in O(|\alpha|)$ (as $\alpha$ is deterministic, we do not even need to take into account that the definition for Kleene plus uses $\gmin$). Hence, each edge can be turned into a transition in time $O(|\alpha|)$. As $\alpha$ is deterministic, there are $O(|\Sigma|n)$ edges, which gives us a total time of $O(|\Sigma||\alpha|n)$ for this step.
	
	As we have the same estimation for both steps, we conclude that the total running time is $O(|\Sigma||\alpha| n)$.
	
	If $\alpha$ is not deterministic, this can be discovered during the construction, as soon as the recursive definition computes a non-deterministic memory occurrence graph $\ograph{\beta}$ for a non-deterministic subexpression $\beta$ of $\alpha$, or if hidden non-determinism is detected.

	\subparagraph*{Unmarked edge labels:} 
	As mentioned above, if the goal is not to construct an $\glush(\alpha)$ that is deterministic if and only if $\alpha$ is deterministic, but to turn every deterministic $\alpha$ in a deterministic $\glush(\alpha)$ and to reject non-deterministic $\alpha$, we can construct $\ograph{\alpha}$ by using unmarked ref-words on the labels. The only cases where using unmarked ref-words can hide non-determinism (in the sense that $\ograph{\alpha}$ is deterministic, although $\alpha$ is not) is in the rule for union. For example, consider $\alpha \df \bind{x}{\emptyword} \ror \bind{x}{\emptyword}$, which satisfies condition~\ref{def:drx:c4} of Definition~\ref{def:drx}, as $\lang(\markpos{\alpha})$ contains $\vop{x}_{(1)} \vcl{x}_{(2)}$ and $\vop{x}_{(3)} \vcl{x}_{(4)}$, due to  $\markpos{\alpha} = (\vop{x}_{(1)} \vcl{x}_{(2)})\ror(\vop{x}_{(3)} \vcl{x}_{(4)})$. If we use unmarked ref-words, $\ograph{\alpha}$ consists only of a single edge from $\src$ to $\snk$ with label $\vop{x}\vcl{x}$, which is clearly deterministic. Nonetheless, we can detect this hidden non-determinism when recursively constructing $\ograph{\alpha}$, by checking whether there exist edges $(\src,\glab_1,\snk)\in \oedges{\beta}$ and $(\src,\glab_2,\snk)\in\oedges{\gamma}$ with $\glab_1\neq \emptyword$ or $\glab_2\neq\emptyword$. Hence, if the conversion algorithm encounters this case, it can reject the regex as non-deterministic.
	
	Note that 	concatenation  cannot hide non-determinism: For  $u\in\onodes{\beta}$ and $v\in\onodes{\gamma}$, define $N_u \df \{\glab \mid (u,\glab,\snk)\in\oedges{\beta}\}$ and $N_v \df \{\glab \mid (\src,\glab,v)\in\oedges{\gamma}\}$. Assume that at least one of the two sets $N_u$ and $N_v$ contains more than one element. Then $\oedges{\alpha}$ contains at least two edges from $u$ to $v$, which means that $\ograph{\alpha}$ is not deterministic. Finally, Kleene star is also unaffected by this change, as the presence of any edge  $(\src,\glab,\snk)$ with $\glab\neq\emptyword$ causes non-determinism regardless of whether $\glab$ is marked or not.
	
	Furthermore, note that an  implementation of this construction can also represent each label $\glab$ in reduced form as $\netvar{\gmin(\glab)}$, if it ensures  that no hidden non-determinism is present. 
\end{proof}

Let us illustrate the construction of the proof of Theorem~\ref{thm:glushkov} by an example.

\begin{example}\label{ex:glush}
	Consider the deterministic regex $\alpha\df \bind{x}{(\ta\ror\tb)^+}\cdot \td\cdot \rr{x}$. Applying the markings yields $\markpos{\alpha}\df \vop{x}_{(1)}(\ta_{(2)}\ror\tb_{(3)})^+\vcl{x}_{(4)}\cdot \td_{(5)}\cdot x_{(6)}$, and $\glush(\alpha)$ is the following automaton: 
	
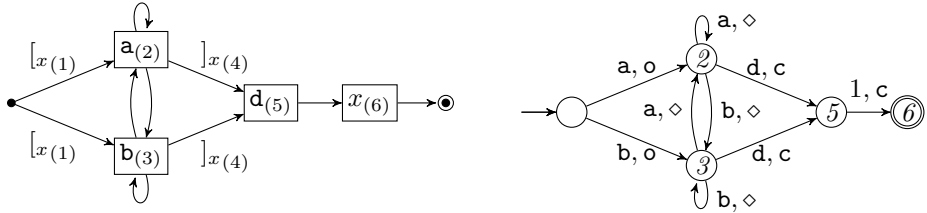
\begin{center}
		\begin{tikzpicture}[node distance=10mm,on grid,>=stealth',auto, 
		state/.style={rectangle,draw=black,inner sep=2pt,minimum size=4mm}]
		\node         (start)                 				{};
		\node[state]  (a1)     [above right=of start,xshift=10mm] 				{\vphantom{$\tb$}$\ta_{(2)}$};
		\node[state]  (b)     [below right=of start,xshift=10mm]     	{$\tb_{(3)}$};
		\node[state]	(d)	[below right=of a1,xshift=10mm]		 	{$\td_{(5)}$};
		\node[state]	(x)	[right = of d,xshift=3mm] {\vphantom{$\tb$}$x_{(6)}$};	 	
		\node 				(end)   [right=of x]	   	{};
		
		\draw[fill=black] (start) circle (0.5mm);
		\draw[fill=black] (end) circle (0.5mm);
		\draw 						(end) circle (1mm);
		
		\path[->]
		(start.center) edge node[above, left=5pt,very near end,yshift=1mm] {$\vop{x}_{(1)}$} (a1)
		(start.center) edge node[below, left=5pt,very near end,yshift=-1mm] {$\vop{x}_{(1)}$} (b)
		(a1) edge[loop above] node {} (a1)
		(b) edge[loop below] node {} (b)
		(a1) edge[bend left=15] node {} (b)
		(b) edge[bend left=15] node {} (a1)
		(a1) edge node[ above, right=5pt, very near start,yshift=1mm] {$\vcl{x}_{(4)}$} (d)
		(b) edge node[below, right=5pt,  very near start,yshift=-2mm] {$\vcl{x}_{(4)}$} (d)
		(d) edge node {} (x)
		(x) edge node{} (end.west);
		\end{tikzpicture}\hspace{5mm}
		\begin{tikzpicture}[node distance=10mm,on grid,>=stealth',auto, 
		state/.style={circle,draw=black,inner sep=0pt,minimum size=4mm}]
		
		\node[state,initial,initial by arrow,initial text={}]         (start)                 				{};
		\node[state]  (a1)     [above right=of start,xshift=1cm]   {2};
		\node[state]  (b)     [below right=of start,xshift=1cm]    {3};
		\node[state]	(d)	[below right=of a1,xshift=1cm]		 	{5};
		\node[state,accepting]	(x)	[right = of d] {6};	 	
		
		\path[->] 
		(start) edge node[above, left=3pt, very near end] {$\ta,\open$} (a1)
		(start) edge node[below, left=3pt,very near end] {$\tb,\open$} (b)
		(a1) edge[loop above] node[right,near end] {$\ta,\unchanged$} (a1)
		(b) edge[loop below] node[right,near start] {$\tb,\unchanged$} (b)
		(a1) edge[bend left=15] node[right,midway] {$\tb,\unchanged$} (b)
		(b) edge[bend left=15] node[left,midway] {$\ta,\unchanged$} (a1)
		(a1) edge node[ above, right=4pt, pos=0.1] {$\td,\close$} (d)
		(b) edge node[below, right=6pt,  pos=0.1] {$\td,\close$} (d)
		(d) edge node {$1,\close$} (x)
		;
		\end{tikzpicture}
	\end{center}	
	To the left, $\glush(\alpha)$ is represented as the memory occurrence graph $\ograph{\alpha}$, to the right as the  $\DTMFA$ that can be directly derived from this graph (which uses memory 1 for~$x$). 
\end{example}
The construction from the proof of Theorem~\ref{thm:glushkov} behaves like the Glushkov construction for  regular expressions, with one important difference: On regex that are not deterministic, its  running time may be exponential in the number of variables; as there are non-deterministic regex where conversion into a $\TMFA$ without $\emptyword$-transitions requires an exponential amount of transitions.  For example, for $k\geq 1$, let $\alpha\df \ta\cdot (\emptyword\ror\bind{x_1}{\emptyword})\cdots(\emptyword\ror\bind{x_k}{\emptyword})\cdot\tb$ and $\beta \df \ta \bigl(\bigror_{1\leq i\leq k}\bind{x_i}{\emptyword}\bigr)^* \tb$. An automaton that is derived with  a Glushkov style conversion then contains states $q_{1}$ and $q_2$ that correspond to the terminals; and between these two states, there must be $2^k$ different transitions to account for all possible combinations of actions on the variables.
This suggests that converting a regex into a $\TMFA$ without $\emptyword$-edges is only efficient for deterministic regex; while in  general, it is probably advisable to use a construction with $\emptyword$-edges.

By combining Theorems~\ref{thm:glushkov} and~\ref{thm:dtmfamembership}, due to $n\leq |\alpha|$, we immediately obtain the following:
\begin{theorem}\label{thm:drxmembership}
	Given  $\alpha \in \DRX$ with $n$ occurrences of terminal symbols or variable references and $k$ variables, and  $w \in \Sigma^*$, we can decide in time $O(|\Sigma||\alpha|n+ k|w|)$, whether  $w \in \lang(\alpha)$.
\end{theorem}	
If we ensure that recalled variables  never contain~$\emptyword$ (or that only a bounded number of variables references are possible in a row), we can even drop the factor $k$. For comparison, the membership problem for $\DREG$ can be decided in time $O(|\Sigma||\alpha|+|w|)$ when using optimized versions of the Glushkov construction (see \cite{bru:reg,pon:new}), and in  $O(|\alpha|+ |w|\cdot \log \log |\alpha|)$ with the algorithm by Groz and Maneth~\cite{gro:det} that  does not compute an automaton.
%!TEX root=det_SIAM.tex
\section{Expressive Power}\label{sec:expressive}
Although C\^ampeanu, Salomaa, Yu~\cite{cam:afo} as well as  Carle and Narendran~\cite{car:one} state pumping lemmas for a class of regex, these do not apply to regex as defined in this paper (see Section~\ref{sec:choices}). However, Lemmas~\ref{lem:jump}~and~\ref{lem:finap}, introduced in Section~\ref{sec:TMFA}, shall be helpful for proving inexpressibility. A consequence of Lemma~\ref{lem:finap} is that there are infinite unary $\DTMFArej$-languages that are not pumpable (in the sense that certain factors can be repeated arbitrarily often), as this would always lead to an arithmetic progression. It is also possible to demonstrate this phenomenon on larger alphabets, without relying on a trivial modification of the unary case.

%\begin{example}\label{ex:fibonacci}
%\todo[inline, color=blue!30]{MS: Some restructuring w.r.t. the Fibonacci stuff was done here, since the full proof is now in the paper.}
%\todomcom{Thanks. But next time, please check whether this breaks references. ;) I~made some changes.}
The Fibonacci word $\fibword$ is the infinite word that is the limit of the sequence of words $F_{0}\df \tb$, $F_{1}\df \ta$, and $F_{n+2}\df F_{n+1}\cdot F_n$ for all $n\geq 0$. The Fibonacci word has a number of curious properties. In particular, it is \emph{cube-free}, which means that it does not contain cubes (i.\,e., factors~$www$, with $w\neq \emptyword$). This and various other properties are explained throughout Lothaire~\cite{lot:alg}. In the following, we demonstrate that an infinite subset $L$ of $\{F_i \mid i \geq 0\}$ can be generated by a deterministic regex. Since the cube-freeness of $\fibword$ particularly implies that all $F_i$, $i \geq 0$, are cube-free, this demonstrates that $L$ is a $\DRX$-language that cannot be pumped by repeating factors of sufficiently large words arbitrarily often. This is a rather counter-intuitive phenomenon with respect to a class of formal languages that is characterised by a natural extension of classical finite automata.
%	
%	This and various other properties are explained throughout Lothaire~\cite{lot:alg}.\par
%
%
%	We define the following regex $\alpha$:
%
%\begin{align*}
%\alpha \df  &\ta \bind{x_0}{\tb}\bind{x_1}{\ta}\\
%&\bigl( \bind{x_2}{\rr{x_1}\rr{x_0}} \bind{x_3}{\rr{x_1}\rr{x_0}\rr{x_1}} \bind{x_0}{\rr{x_3}\rr{x_2}} \bind{x_1}{\rr{x_3}\rr{x_2}\rr{x_3}}       \bigr)^*
%\end{align*}
%	
%	$$\alpha \df  \ta \bind{x_0}{\tb}\bind{x_1}{\ta}\bigl( \bind{x_2}{\rr{x_1}\rr{x_0}} \bind{x_3}{\rr{x_1}\rr{x_0}\rr{x_1}} \bind{x_0}{\rr{x_3}\rr{x_2}} \bind{x_1}{\rr{x_3}\rr{x_2}\rr{x_3}}       \bigr)^*.$$
%	Then $\lang(\alpha)= \{F_{4i+3}\mid i\geq 0\}$. Hence, the words of $\lang(\alpha)$ converge towards $\fibword$. The  proof of this equivalence is straightforward, but  long. It uses that $F_{n+3} = F_{n+1}\cdot  F_n \cdot F_{n+1}$ holds for all $n\geq 0$
%	As $\fibword$ contains no cube, the same applies to all $F_n$. Thus, $\lang(\alpha)$  is a $\DRX$-language that cannot be pumped by repeating factors of sufficiently large words arbitrarily often.
%\end{example}

%\subsection{Lemma~\ref{lem:fibonacci} (for Example~\ref{ex:fibonacci})}

\begin{lemma}\label{lem:fibonacci}
	Let $L\df \{F_{4i+3}\mid i\geq 0\}$. Then $L=\lang(\beta)$ holds for the deterministic regex
%	\begin{multline*}
%\beta \df  \ta \bind{x_0}{\tb}\bind{x_1}{\ta}\\\cdot\bigl( \bind{x_2}{\rr{x_1}\rr{x_0}} \bind{x_3}{\rr{x_1}\rr{x_0}\rr{x_1}} \bind{x_0}{\rr{x_3}\rr{x_2}} \bind{x_1}{\rr{x_3}\rr{x_2}\rr{x_3}}       \bigr)^*.
%	\end{multline*}
\begin{align*}
\beta &\df  \ta \bind{x_0}{\tb}\bind{x_1}{\ta} \cdot (\beta_{\mathsf{shift}})^*\,,\\
\beta_{\mathsf{shift}} &\df \bind{x_2}{\rr{x_1}\rr{x_0}} \bind{x_3}{\rr{x_1}\rr{x_0}\rr{x_1}} \bind{x_0}{\rr{x_3}\rr{x_2}} \bind{x_1}{\rr{x_3}\rr{x_2}\rr{x_3}}\,.
\end{align*}

\end{lemma}
\begin{proof}
%	First, we observe that $\refl(\beta)=\{r_i\mid i\geq 0\}$, where the ref-words $r_i$ are defined by 
%	\begin{align*}
%		r_0 &\df \ta \vop{x_0}\tb\vcl{x_0} \vop{x_1} \ta\vcl{x_1} ,\\
%		\hat{r} &\df \vop{x_2} x_1  x_0 \vcl{x_2} \vop{x_3} x_1  x_0  x_1 \vcl{x_3} \vop{x_0} x_3  x_2 \vcl{x_0} \vop{x_1} x_3  x_2  x_3 \vcl{x_1},
%	\end{align*}
%	and $r_{i+1} \df r_i \cdot \hat{r}$ for all $i\geq 0$. 
First, we observe that $\refl(\beta)=\{r_i\mid i\geq 0\}$, where the ref-words $r_i$ are defined by $r_0 \df \ta \vop{x_0}\tb\vcl{x_0} \vop{x_1} \ta\vcl{x_1}$ and $r_{i+1} \df r_i \cdot \hat{r}$ for all $i\geq 0$, where
\begin{equation*}
\hat{r} \df \vop{x_2} x_1  x_0 \vcl{x_2} \vop{x_3} x_1  x_0  x_1 \vcl{x_3} \vop{x_0} x_3  x_2 \vcl{x_0} \vop{x_1} x_3  x_2  x_3 \vcl{x_1}\,.
\end{equation*}
We now prove by induction that, for each $i\geq 0$, $\deref(r_i) = F_{4i+3}$, and the rightmost values that are assigned to $x_0$ and $x_1$ are $F_{4i}$ and $F_{4i+1}$, respectively. For $i=0$, this is obviously true: $\deref(r_0) = \mathtt{aba} = F_3$, $x_0$ is assigned $\tb=F_0$, and $x_1$ is assigned $\ta = F_1$.
	
	Now assume that the claim holds for some $i\geq 0$, and consider $r_i$. Then we can observe that 
	$\deref(r_{i+1}) = \deref(r_i \cdot \hat{r}) = \deref(r_i) \cdot \deref (s)$, 
	where the ref-word $s$ is obtained from $\hat{r}$ by replacing $x_0$ and $x_1$ with their respective values $F_{4i}$ and $F_{4i+1}$. Hence,
	\begin{align*}
		s &= \vop{x_2} F_{4i+1} \cdot F_{4i} \vcl{x_2} \vop{x_3} F_{4i+1} \cdot  F_{4i}\cdot  F_{4i+1} \vcl{x_3} \vop{x_0} x_3  x_2 \vcl{x_0} \vop{x_1} x_3  x_2  x_3 \vcl{x_1}\\
		&= \vop{x_2} F_{4i+2} \vcl{x_2} \vop{x_3} F_{4i+3} \vcl{x_3} \vop{x_0} x_3  x_2 \vcl{x_0} \vop{x_1} x_3  x_2  x_3 \vcl{x_1}.
	\end{align*}
	The second part of this equation uses that $F_{n+3} = F_{n+2}\cdot  F_{n+1} = F_{n+1}\cdot  F_n \cdot F_{n+1}$ holds for all $n\geq 0$. We now construct a ref-word $t$ by replacing the variables $x_2$ and $x_3$ in $s$ with their respective values. Hence,
	\begin{align*}
		t &=  \vop{x_2} F_{4i+2} \vcl{x_2} \vop{x_3} F_{4i+3} \vcl{x_3} \vop{x_0} F_{4i+3}\cdot  F_{4i+2}  \vcl{x_0} \vop{x_1} F_{4i+3} \cdot  F_{4i+2} \cdot   F_{4i+3} \vcl{x_1}\\
		&=  \vop{x_2} F_{4i+2} \vcl{x_2} \vop{x_3} F_{4i+3} \vcl{x_3} \vop{x_0} F_{4i+4}  \vcl{x_0} \vop{x_1} F_{4i+5} \vcl{x_1}.
	\end{align*}
	Then $\deref(r_{i+1}) = \deref(r_i \cdot t)$ holds, and $r_{i+1}$ assigns $x_0$ and $x_1$ as $t$ does. Hence, $x_0$ is assigned $F_{4(i+1)}$, and $x_1$ is assigned $F_{4(i+1)+1}$, as required by the claim. To see that $\deref(r_{i+1})=F_{4(i+1)+3}$, we observe that
	\begin{align*}
		\deref(r_{i+1}) &= \deref(r_i \cdot t) \\
		&= F_{4i+3} \cdot \deref(\vop{x_2} F_{4i+2} \vcl{x_2} \vop{x_3} F_{4i+3} \vcl{x_3} \vop{x_0} F_{4i+4}  \vcl{x_0} \vop{x_1} F_{4i+5} \vcl{x_1})\\
		&= F_{4i+3} \cdot  F_{4i+2} \cdot F_{4i+3} \cdot  F_{4i+4}\cdot   F_{4i+5} \\
		&= F_{4i+5} \cdot  F_{4i+4}\cdot   F_{4i+5} \\
		&= F_{4i+7}=   F_{4(i+1)+3}.
	\end{align*}
	This concludes the proof. 
\end{proof}

%\todomcom{Made comments regarding $M$  more explicit, please check whether you agree.}

Also note that we can construct an $M\in\DTMFArej$ with $\lang(M)=L_F\df \{F_n\mid n\geq 1\}$ from $\glush(\beta)$. More specifically, this can be achieved by making the states that have an outgoing transition that opens a variable accepting. In terms of regex, this corresponds to being able to accept before each variable binding (the correctness of this construction follows from the proof of Lemma~\ref{lem:fibonacci}). Note that the cycle in $M$ has multiple accepting states. The authors conjecture that there is no $\DTMFA$ for $L_F$ that has a cycle with exactly one accepting state, and that $L_F$ is not a $\DRX$-language.

For further separations, we use the following language:
\begin{example}\label{ex:drxCrazy}	
	Let $\alpha\df \ta^2\cdot \bind{x}{\ta^2} \cdot\bigl( \bind{y}{\rr{x}\cdot \rr{x}}\cdot \bind{x}{\rr{y}\cdot\rr{y}}  \bigr)^*$. Then $\lang(\alpha)=\{\ta^{4^i}\mid i\geq 1 \}$. This can be proven with a straightforward induction that is left to the reader. 
\end{example}
From this, we define an  $L\in\langcl(\TMFA)$ with neither  $L\in\langcl(\DTMFArej)$, nor  $L\in\langcl(\DTMFAacc)$:
\begin{lemma}\label{lem:unarySep}
	Let 
	$L\df \{\ta^{4i+1}\mid i\geq 0\}\cup \{\ta^{4^i}\mid i\geq 1\}$. Then $L\in\langcl(\TMFA)\mdif\langcl(\DTMFA)$.
\end{lemma}

%\subsection{Proof of Lemma~\ref{lem:unarySep}}
\begin{proof}
	To show that $L\in\langcl(\TMFA)$, we construct a regex $\alpha$ for $L$, by $\alpha\df \alpha_1 \ror \alpha_2$, where $\alpha_1\df \ta(\ta^4)^*$, and $\alpha_2$ is the deterministic regex with $\lang(\alpha_2)=\{\ta^{4^i}\mid i\geq 1 \}$ from Example~\ref{ex:drxCrazy}.
	
	Next, observe that $L$ contains the arithmetic progression $\{\ta^{4i+1}\mid i\geq 0\}$, and $\overline{L}\df \{\ta\}^*\mdif L$ contains the arithmetic progression $\{\ta^{4i+2}\mid i\geq 0\}$. Assume that $L$ is a $\DTMFA$-language. Then there is an $A\in\DTMFArej$ that accepts $L$ or $\overline{L}$. Moreover, by Lemma~\ref{lem:finap}, $L$ is regular (note that in case $\lang(A)=\overline{L}$, we also use that the class of regular languages is closed under complementation). This leads to a contradiction, since, as we shall see next, $L$ is not regular. 
%	But as $L$ is not regular, this is a contradiction. 
To show that $L$ is not regular, first assume the contrary. Then $\lang(\alpha_2)=L\cap \{\ta^4\}^*$ would be regular, as the class of regular languages is closed under intersection. But $\lang(\alpha_2)$ is not regular, as for every pair $i \neq j$, $\ta^{4^i}$ and $\ta^{4^j}$ are not Nerode-equivalent. 
\end{proof}

While  inexpressibility through $\DTMFArej$ provides us with a powerful sufficient criterion for $\DRX$-inexpressibility, it is not powerful enough to cover all cases of $\DRX$-inexpressibility. 
In particular, there are even regular languages that are no $\DRX$-languages:
\begin{lemma}\label{lem:drxvsdrmfa}
	Let $L\df \lang\bigl((\ta\tb)^*(\ta\ror\emptyword)\bigr)=\{(\mathtt{ab})^{\frac{1}{2}i}\mid i\geq 0\}$. Then $L\in \langcl(\REG)\mdif \langcl(\DRX)$. 
\end{lemma}

%\subsection{Proof of Lemma~\ref{lem:drxvsdrmfa}}
\begin{proof}
	Before we assume the existence of an $\alpha\in\DRX$ with $\lang(\alpha)=L$ (and use this to obtain a contradiction), we first examine the structure of any $M\in\DTMFArej$ with $\lang(M)=L$. We observe that, with the exemption of states that are unreachable or cannot reach an accepting state, $M$ must consist of a chain (which might be empty) that is followed by a cycle that contains at least one final state (like a $\DFA$ for a unary language, see the proof of Theorem~\ref{thm:drxunary}, in particular Figure~\ref{lollipopFigure}).
	
	This is for the following reason: First, like for every $\DTMFA$, each state of $M$ that has an outgoing memory recall transition cannot have any other outgoing transitions. The same holds for $\emptyword$-transitions. Furthermore, due to the structure of $L$, in $M$ no state can have an outgoing transition that consumes $\ta$ and an outgoing transition that consumes $\tb$ at the same time. For $\DTMFA$, this is not problematic. In fact, as $L$ is regular, we can interpret any $\DFA$ for $L$ as a $\DTMFA$ for $L$. For example, consider the following minimal incomplete $\DFA$ for $L$, and its corresponding notation as an occurrence graph (without markings):
	
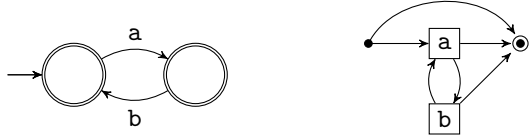
\begin{center}
		\begin{tikzpicture}[node distance=16mm,on grid,>=stealth',auto, 
		state/.style={circle,draw=black,inner sep=1pt,minimum size=8mm}]
		
		\node[state,initial,initial by arrow,initial text={},accepting]         (qb) {};
		\node[state,accepting, right=of qb]         (qa) {};
		
		\path[->] 
		(qb) edge[bend left] node[] {$\ta$} (qa)
		(qa) edge[bend left] node[] {$\tb$} (qb)
		;
		\end{tikzpicture}
		\hspace{15mm}
		\begin{tikzpicture}[node distance=10mm,on grid,>=stealth',auto, 
		state/.style={rectangle,draw=black,inner sep=0pt,minimum size=4mm}]
		\node         (start)                 				{};
		\node[state]  (a)     [right=of start] 				{$\ta$};
		\node[state]  (b)     [below=of a]     	{$\tb$};
		\node 				(end)   [right=of a]	   	{};
		
		\draw[fill=black] (start) circle (0.5mm);
		\draw[fill=black] (end) circle (0.5mm);
		\draw 						(end) circle (1mm);
		
		\path[->]
		(start.center) edge node[] {} (a)
		(start.center) edge[bend left=60] node {} (end)
		(a) edge node{} (end)
		(b) edge node{} (end)
		(a) edge[bend left] node {} (b)
		(b) edge[bend left] node {} (a)
		%		 	 	(a1) edge[loop above] node {} (a1)
		%		 	 	(b) edge[loop below] node {} (b)
		%		 	 	(a1) edge[bend left=15] node {} (b)
		%		 	 	(b) edge[bend left=15] node {} (a1)
		%		 	 	(a1) edge node[ above, right, pos=0.1] {$\close(x)$} (d)
		%		 	 	(b) edge node[below, right,  pos=0.1] {$\close(x)$} (d)
		%		 	 	(d) edge node {} (x)
		%		 	 	(x) edge node{} (end.west);
		;
		\end{tikzpicture}
	\end{center}
	
	If we consider the occurrence graph notation, we see that this $\DFA$ cannot be obtained from a deterministic regular expression (at least not using the Glushkov construction, which -- in the absence of variables -- is identical to the construction from the proof of Theorem~\ref{thm:glushkov}). Note that the states for $\ta$ and $\tb$ belong to the same strongly connected component. Hence, if there is an $\alpha\in\DREG$ such that this automaton is $\glush(\alpha)$, then $\alpha$ must contain a subexpression $\beta^+$ with the occurrences $\ta_{(i)}$ and $\tb_{(j)}$ that correspond to these states. Then $\ograph{\beta}$ must contain  edges from $\ta_{(i)}$ and $\tb_{(j)}$ to $\snk$, and from $\src$ to $\ta_{(i)}$. Using the rule for Kleene plus from the proof of Theorem~\ref{thm:glushkov}, we see that $\ograph{\alpha}$ must contain an edge from $\ta_{(i)}$ to itself. This is a contradiction.  Of course, this argument only shows that this $\DFA$ cannot be obtained from a deterministic regular expression; but it can be generalized to show that there is no $\alpha\in\DREG$ with $\lang(\alpha)=L$ (see e.\,g.~\cite{bru:one}, and in particular~\cite{car:gen}, which explains how to apply the technique from~\cite{bru:one} on this language).
	
	We shall now use a similar line of reasoning to obtain a contradiction from the assumption that there is an  $\alpha\in\DRX$ with $\lang(\alpha)=L$. As explained above, the $\DTMFA$ $\glush(\alpha)$ must consist of a chain and a cycle (by definition, each state of $\glush(\alpha)$ is reachable; and from each state, we can reach an accepting state). This means that  $G_{\markpos{\alpha}}$ contains a cycle $v_1, \ldots, v_n$ for some $n\geq 1$ and $v_i\in (\markpos{\Sigma}\cup\markpos{\Xi})$ such that there is an edge from $v_i$ to $v_{i+1}$ for $1\leq i < n$ and from $v_{n}$ to $v_1$. Hence, each $v_i$ has at most two outgoing edges: One to the next node in the cycle, and (if it is an accepting state) one to the sink node $\snk$. Furthermore, exactly one $v_i$ has an incoming edge from outside the cycle, let this be $v_1$. 
	
	From the construction of $G_{\markpos{\alpha}}$, this cycle must have been generated from a Kleene plus in $\alpha$. But this allows us to conclude that only exactly one $v_i$ can have an edge to $\snk$; and furthermore, that this must be $v_n$. This can be concluded from the following reasoning: If there existed nodes $v_i,v_j$ with $i\neq j$, and both have an edge to $\snk$, then the construction for Kleene plus would require edges from both $v_i$ and $v_j$ to $v_1$, which would break the cycle structure. Likewise, if $i\neq n$, then there must be an edge from $v_i$ to $v_1$, and from $v_i$ to $v_{i+1}$, which is a contradiction to our previous observations.
	
	Hence,  each iteration of the cycle must consume exactly one terminal letter (otherwise, we would skip over words of $L$), alternating between $\ta$ and $\tb$. Thus,  $v_i\in\markpos{\Xi}$ must hold for all $1\leq i \leq n$ (as $v_i = a_{(j)}$ with $a_{(j) }\in\markpos{\Sigma}$ would consume $a$ in every iteration, which would contradict the fact that the iterations alternate between consuming $\ta$ and $\tb$). 
	%\todomcom{Added explanation you asked for. Enough?}
	
	Now assume that we enter an iteration that consumes $\ta$ (the same reasoning shall hold for $\tb$). Then no variable that is recalled can contain $\tb$, and no variable can be bound to $\tb$, as otherwise, the iteration would consume more than $\ta$. But in the next iteration, the same variables are recalled, and as neither of them contains $\tb$, the iteration cannot consume $\tb$. Therefore, we arrive at a contradiction, and conclude that there is no $\alpha\in\DRX$ with $\lang(\alpha)=L$.
\end{proof}

The language $L$ from Lemma~\ref{lem:drxvsdrmfa} is also known to be a non-deterministic regular language (see e.\,g.\ \cite{bru:one}). Our proof can be seen as taking the idea behind the characterization of deterministic regular languages from~\cite{bru:one}, applying it to the specific language $L$, and also taking variables into account. While this accomplishes the task of proving  that deterministic regex share some of the limitations of deterministic regular expressions, the approach does not generalize (at least not in a straightforward manner). In particular, deterministic regex can express regular languages that are not deterministic regular, and are also quite similar to $L$:
\begin{example}\label{ex:drxvsdrmfa}
	Let $L\df \{(\ta\tb)^{\frac{3}{2}i}\mid i\geq 0\}$. Then $L$ is generated by the non-deterministic regular expression $(\mathtt{ababab})^* (\emptyword\ror (\mathtt{aba}))$, and one can show that $L$ is not a deterministic regular language by using the BKW-algorithm~\cite{bru:one} (also~\cite{cze:dec,lu:dec}) on the minimal $\DFA$~$M$ for~$L$. But for 
\begin{equation*}
\alpha \df \ta\bind{y}{\tb}\bind{x}{\ta}  \bigl(\bind{z}{\rr{y}}\bind{y}{\rr{x}}\bind{x}{\rr{z}}\bigr)^*\,, 
\end{equation*}
we have $\alpha\in\DRX$ and $\lang(\alpha)=L$ (to see this, note that in the iterations of the Kleene-star, the variable contents alternate between $z = \tb$, $y = \ta$, $x = \tb$ and $z = \ta$, $y = \tb$, $x = \ta$).
\end{example}
The ``shifting gadget'' that is used in Example~\ref{ex:drxvsdrmfa} can be extended to show a far more general result for unary languages. Considering that $\langcl(\DREG)\subset\langcl(\REG)$ holds even over unary alphabets  (cf.\ Losemann et al.~\cite{los:clo}),
 the following result might seem surprising:
\begin{theorem}\label{thm:drxunary}
	For every regular language $L$ over a unary alphabet, $L\in\langcl(\DRX)$.
\end{theorem}

%\subsection{Proof of Theorem~\ref{thm:drxunary}}
\begin{proof}
	Assume that $L \in \langcl(\REG)$ with $L\subseteq\{\ta\}^*$. Our goal is to construct an $\alpha\in\DRX$ with $\lang(\alpha)=L$. For technical reasons, we assume that $\emptyword\notin L$ (this is no problem, as for any $\alpha\in\DRX$ with $\emptyword\notin\lang(\alpha)$, $(\alpha\ror\emptyword)\in\DRX$). If $L$ is finite, $L\in\langcl(\DREG)$, and hence $L\in\langcl(\DRX)$. (As we shall see, our construction can also be used for finite languages, by replacing $\alpha_{\mathsf{cycle}}$ in $\alpha^{\mathsf{chain}}_k$ below with $\emptyword$. But to streamline the argument, we only consider infinite $L$.)
	
	Let $M$ be a $\DFA$ with $\lang(M)=L$. Assume that  all states of $M$ are reachable, and that from each state, an accepting state can be reached. Then $M$ has the  form as shown in Figure~\ref{lollipopFigure}.
		%\todomcom{Moved the explanation into the figure.}
	\begin{figure}
	\begin{center}
		\begin{tikzpicture}[node distance=16mm,on grid,>=stealth',auto, 
		state/.style={circle,draw=black,inner sep=1pt,minimum size=8mm}]
		
		\node[state,initial,initial by arrow,initial text={}]         (p0) {$p_0$};
		\node[state]  (p1)     [right=of start]   {$p_1$};
		\node[]  (dots1)     [right=of p1]    {$\cdots$};
		\node[state]	(pm)	[right=of dots1]		 	{$p_m$};
		\node[state]	(q1)	[above right=of pm]		 	{$q_1$};
		\node[state]	(q2)	[right=of q1]		 	{$q_2$};
		\node[]	(dots2)	[below right=of q2]		 	{$\vdots$};
		\node[state]	(qnm1)	[below right=of pm]		 	{$q_{n-1}$};
		\node[state]	(qnm2)	[right=of qnm1]		 	{$q_{n-2}$};
		%	 	\node[state,accepting]	(x)	[right = of d] {};	 	
		%	 	
		\path[->] 
		(p0) edge node[] {$\ta$} (p1)
		(p1) edge node[] {$\ta$} (dots1)
		(dots1) edge node[] {$\ta$} (pm)
		(pm) edge[bend left] node[] {$\ta$} (q1)
		(q1) edge node[] {$\ta$} (q2)
		(q2) edge[bend left] node[] {$\ta$} (dots2.north)
		(dots2.south) edge[bend left] node[] {$\ta$} (qnm2)
		(qnm2) edge[] node[] {$\ta$} (qnm1)
		(qnm1) edge[bend left] node[] {$\ta$} (pm)
		;
		\end{tikzpicture}
	\end{center}
	\caption{Illustration of the unary DFA in the proof of Theorem~\ref{thm:drxunary}. Note that here, we do not distinguish between accepting and non-accepting states}
	\label{lollipopFigure}
	\end{figure}
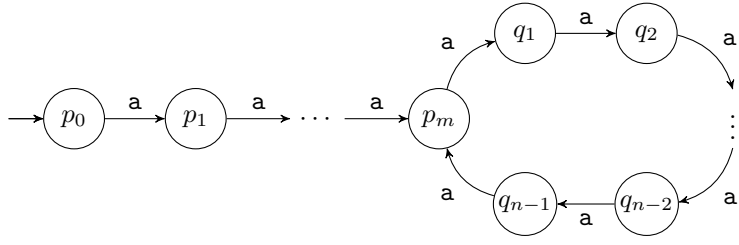
	We refer to the states $p_0$ to $p_m$ as the \emph{chain}, and to the states $q_1,\ldots,q_{n-1}$ and $p_m$ as the \emph{cycle}.
	Without loss of generality, we can assume that $p_m$ is accepting (the cycle contains at least an accepting state; and as the automaton does not have to be minimal, we can extend the chain of $p_i$ by unrolling the cycle until it starts with an accepting state). %For convenience, we also define $q_n\df p_m$. 
	
	Now there exists a number $k\geq 1$ and $c_1,\ldots,c_k\geq 1$ such that the words $\ta^{c_1}$, $\ta^{c_1+c_2}, \ldots, \ta^{c_1+\cdots+c_k}$ are exactly the words that are accepted in the chain (recall that we assume that $p_m$ is accepting, and that $p_0$ is not accepting, as $\emptyword\notin L$). Furthermore, there exists an $\ell\geq 1$ and $b_1,\ldots,b_{\ell}\geq 1$ such that the words $\ta^{b_1}, \ta^{b_1+b_2}, \ldots, \ta^{b_1+\cdots+b_{\ell}}$ are exactly the words that advance the cycle from $p_m$ to each of the accepting states. These conditions also imply $m=\sum_{i=1}^{k} c_i$ and $n=\sum_{i=1}^{\ell}b_i$.
	
	As an additional restriction, we assume that $b_1 \geq 2$. This is possible for the following reasons: If $b_i=1$ for all $i$, we can  replace the cycle with the deterministic regular expression $\ta^*$ and are done. Furthermore, if $b_1=1$, but there is a $b_i \geq 2$, we can unroll the cycle into the chain until $b_1\geq 2$ (for the ``new'' $b_1$).
	
	We define $\alpha\df \alpha_{1}^{\mathsf{chain}}$, where, for $1\leq i < k$, 
	\begin{align*}
		\alpha^{\mathsf{chain}}_i &\df \ta^{c_i}(\emptyword \ror \alpha^{\mathsf{chain}}_{i+1}),\\
		\alpha^{\mathsf{chain}}_k &\df \bind{x_{\ell}}{\ta}\ta^{c_{k}-1}(\emptyword \ror \alpha_{\mathsf{cycle}}). 
	\end{align*}
	Before we define $\alpha_{\mathsf{cycle}}$, note that if we disregard the words that can be generated by the subexpression $ \alpha_{\mathsf{cycle}}$, $\lang(\alpha)$ contains exactly the words that are accepted by the chain. Furthermore, if we assume that $\alpha_{\mathsf{cycle}}$ is a deterministic regex and that its language does not contain $\emptyword$, we can conclude $\alpha\in\DRX$. Finally, note that if we first enter $\alpha_{\mathsf{cycle}}$, the variable $x_{\ell}$ contains $\ta$.
	
	The central part of the construction is defining $\alpha_{\mathsf{cycle}}$ in such a way that it is deterministic and it simulates the cycle of the DFA.
	% by generating exactly the words  $\ta^{b_1}, \ta^{b_1+b_2}, \ldots, \ta^{b_1+\cdots+b_{\ell}}$. 
	We define 
	\begin{align*}
		\alpha_{\mathsf{cycle}}&\df \bigl(\alpha_{\mathsf{shift}} \cdot \alpha_{\mathsf{cont}} \bigr)^+,\\
		\alpha_{\mathsf{shift}}&\df \bind{x_0}{\rr{x_{\ell}}} \bind{x_{\ell}}{\rr{x_{\ell-1}}}\cdots \bind{x_2}{\rr{x_1}}\bind{x_1}{\rr{x_0}},\\
		\alpha_{\mathsf{cont}}&\df \rr{x_1}^{b_1-2}\cdot\rr{x_2}^{b_2-1}\cdot \cdots \cdot\rr{x_{\ell}}^{b_{\ell}-1}.
	\end{align*}
The idea behind this definition is as follows. The expression $\alpha_{\mathsf{cycle}}$ should be able to generate the words that are generated in the cycle of the DFA, i.\,e., all the words $\ta^{i}$, where $i = r \sum^{\ell'}_{j = 1} b_j$ for some $r \geq 1$ and $\ell'$ with $1 \leq \ell' \leq \ell$ (more precisely, $r$ is the iteration of the cycle and in the current iteration, the word terminates in the $(\ell')^{\text{th}}$ accepting state). This is done by producing, in every iteration of the expression $(\alpha_{\mathsf{shift}} \cdot \alpha_{\mathsf{cont}})^+$, another factor $\ta^{b_j}$ as follows: 
%
%$\ta^{b_1}, \ta^{b_1+b_2}, \ldots, \ta^{b_1+\cdots+b_{\ell}}$
%
%$\ta^{b_1}, \ta^{b_1+b_2}, \ldots, \ta^{b_1+\cdots+b_{\ell}}, \ta^{2b_1+\cdots+b_{\ell}}, \ta^{2b_1+2b_2+\cdots+b_{\ell}}, \ldots$
%The idea behind this definition is as follows: Before the first iteration of the Kleene plus, $x_{\ell}$ contains $\ta$, and all other variables default to $\emptyword$. Now, note that passing through $\alpha_{\mathsf{shift}}$ shifts the $\ta$ from $x_{\ell}$ to $x_1$, or from $x_i$ to $x_{i+1}$ for $1\leq i < \ell$. As this cyclic shift needs an extra variable, if $x_1$ is set to $\ta$, $x_0$ is also set to $\ta$ (which is overwritten in the next iteration of the plus).
%The idea behind this definition is as follows: 
Before the first iteration,
%of the Kleene plus, 
$x_{\ell}$ contains $\ta$, and all other variables default to $\emptyword$. Now, note that passing through $\alpha_{\mathsf{shift}}$, the $\ta$ from $x_{\ell}$ is shifted to $x_1$, or from $x_i$ to $x_{i+1}$ for $1\leq i < \ell$. This means that in every application of $\alpha_{\mathsf{cont}}$, exactly one of the variables $x_1, x_2, \ldots, x_{\ell}$ contains $\ta$, which is responsible for producing the next $\ta^{b_j}$ factor, while all the other variables are empty. However, this shifting mechanism needs an additional variable, namely $x_0$, which, in each iteration that charges $x_1$ with $\ta$, produces an additional occurrence of $\ta$ (this explains the assumption $b_1\geq 2$ above and that $\rr{x_1}$ is repeated for only $b_1-2$ times in $\alpha_{\mathsf{cont}}$ instead of $b_i-1$, as the other variable references). In summary, in the $i$-th iteration of the plus, the variable $x_j$ with $j\df((i-1)\bmod \ell) + 1$ is set to $\ta$, and if $j=1$, then $x_0$ is also set to $\ta$. All other variables are set to $\emptyword$. This means that $\alpha_{\mathsf{shift}}$ produces $\ta^2$ in the $i$-th iteration if $(i\bmod \ell) = 1$, and $\ta$ in all other iterations. 
%This is the reason we ensured that $b_1\geq 2$ above, as we now use $\alpha_{\mathsf{cont}}$ to produce the remaining letters. 
If $j=1$, then the $\rr{x_1}^{b_1-2}$ in $\alpha_{\mathsf{cont}}$ produces $\ta^{b_1-2}$, which means that in this iteration, $\alpha_{\mathsf{shift}} \cdot \alpha_{\mathsf{cont}}$ produces $\ta\ta\cdot \ta^{b_1-2}=\ta^{b_1}$ (recall that all other variables are set to $\emptyword$). If $j>1$, then  $\alpha_{\mathsf{cont}}$ used $\rr{x_j}^{b_j-1}$ to produce $\ta^{b_j-1}$, which means that $\alpha_{\mathsf{shift}} \cdot \alpha_{\mathsf{cont}}$ produces $\ta\cdot \ta^{b_j-1}=\ta^{b_j}$.
%
%This explains 
%
%As this cyclic shift needs an extra variable, if $x_1$ is set to $\ta$, $x_0$ is also set to $\ta$ (which is overwritten in the next iteration of the plus).
%	
%	Hence, in the $i$-th iteration of the plus, the variable $x_j$ with $j\df((i-1)\bmod \ell) + 1$ is set to $\ta$, and if $j=1$, then $x_0$ is also set to $\ta$. All other variables are set to $\emptyword$. This means that $\alpha_{\mathsf{shift}}$ produces $\ta^2$ in the $i$-th iteration if $(i\bmod \ell) = 1$, and $\ta$ in all other iterations. This is the reason we ensured that $b_1\geq 2$ above, as we now use $\alpha_{\mathsf{cont}}$ to produce the remaining letters. If $j=1$, then the $\rr{x_1}^{b_1-2}$ in $\alpha_{\mathsf{cont}}$ produces $\ta^{b_1-2}$, which means that in this iteration, $\alpha_{\mathsf{shift}} \cdot \alpha_{\mathsf{cont}}$ produces $\ta\ta\cdot \ta^{b_1-2}=\ta^{b_1}$ (recall that all other variables are set to $\emptyword$). If $j>1$, then  $\alpha_{\mathsf{cont}}$ used $\rr{x_j}^{b_j-1}$ to produce $\ta^{b_j-1}$, which means that $\alpha_{\mathsf{shift}} \cdot \alpha_{\mathsf{cont}}$ produces $\ta\cdot \ta^{b_j-1}=\ta^{b_j}$. 
	
	In conclusion, the $i$-th iteration of the Kleene plus in $\alpha_{\mathsf{cycle}}$ adds the word $\ta^{b_j}$ for $j\df((i-1)\bmod \ell) + 1$; which means that $\alpha_{\mathsf{cycle}}$ simulates the cycle. Hence, $\lang(\alpha)=L$. As  $\alpha_{\mathsf{cycle}}$  contains no disjunctions or Kleene plus (except for the surrounding plus),  it  is deterministic. As remarked above, this allows us to conclude $\alpha\in\DRX$.
\end{proof}

As, in the proof of Theorem~\ref{thm:drxunary}, a $\DFA$ with $n$ states is converted into a deterministic regex of length $O(n)$, this construction is even efficient.
We  summarize our observations (also see Figure~\ref{fig:hierarchyay}):
\begin{theorem}\label{thm:hierarchyay}
	$$\langcl(\DREG)\subset \langcl(\DRX) \subset \langcl(\DTMFArej)\subset \langcl(\DTMFA) \subset \langcl(\TMFA)=\langcl(\RX)$$
	
	The following pairs of classes are incomparable: $\langcl(\DRX)$ and $\langcl(\REG)$, $\langcl(\DRX)$ and $\langcl(\DTMFAacc)$, as well as $\langcl(\DTMFArej)$ and $\langcl(\DTMFAacc)$.   
\end{theorem}

\begin{figure}
	%Now we are real language theorists:
	\begin{center}
		\begin{tikzpicture}[on grid]
		\node (MFA) {$\langcl(\TMFA)=\langcl(\RX)$};
		\node[below=of MFA] (DTMFA) {$\langcl(\DTMFA)$};
		\node[below left= of DTMFA,xshift=-5mm] (DTMFArej) {$\langcl(\DTMFArej)$};
		\node[below right = of DTMFA,xshift=5mm] (DTMFAacc) {$\langcl(\DTMFAacc)$};
		\node[below left= of DTMFArej,xshift=-5mm] (DRX) {$\langcl(\DRX)$};
		\node[below right= of DTMFArej,xshift=5mm] (REG) {$\langcl(\REG)$};
		\node[below right= of DRX,xshift=3mm] (DREG) {$\langcl(\DREG)$};
		\path[->,>=stealth']
		(DREG) edge [] node {} (DRX)
		(DREG) edge [] node {} (REG)
		(DRX) edge [] node {} (DTMFArej)
		(REG) edge [] node {} (DTMFArej)
		(REG) edge [] node {} (DTMFAacc)	
		(DTMFArej) edge [] node {} (DTMFA)
		(DTMFAacc) edge [] node {} (DTMFA)	
		(DTMFA) edge [] node {} (MFA)	
		;
		\end{tikzpicture}
	\end{center}
	%			\begin{tikzpicture}[on grid]
	%			\node (DREG) {$\langcl(\DREG)$};
	%			\node[above right= of DREG,xshift=10mm,yshift=-5mm] (DRX) {$\langcl(\DRX)$};
	%			\node[below right=of DREG,xshift=10mm,yshift=5mm] (REG) {$\langcl(\REG)$};
	%			\node[below right=of DRX,anchor=west,yshift=5mm] (DTMFA) {$\langcl(\DTMFA)\subset\langcl(\TMFA)=\langcl(\RX)$};
	%			
	%			\path[]
	%			(DREG.north east) edge[draw=none] node [sloped, allow upside down] {$\subset$} (DRX.west)
	%			(DREG.south east) edge[draw=none] node [sloped, allow upside down] {$\subset$} (REG.west)
	%			(DRX) edge[draw=none] node [sloped, allow upside down,midway] {$\subset$} (DTMFA)
	%			(REG) edge[draw=none] node [sloped, allow upside down] {$\subset$} (DTMFA)
	%			;		
	%	\end{tikzpicture}
	\caption{The proper inclusions from Theorem~\ref{thm:hierarchyay}. Arrows point from sub- to superset.}\label{fig:hierarchyay}
\end{figure}
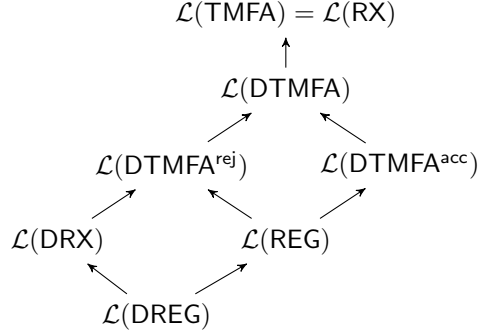
\begin{proof}
	This follows from our previous observations as follows:
	\begin{enumerate}
		\item $\langcl(\DREG)\subset \langcl(\DRX)$: The inclusion follows from the fact that our definition of determinism for regex is an extension of the notion of determinism for proper regular expressions. To see that the inclusion is proper, we recall any of the  non-regular $\DRX$-languages that we have seen, for example  $\{w\tc w\mid w\in\{\ta,\tb\}^*\}$ and $\{\ta^{n^2}\mid n\geq 0\}$ from Example~\ref{ex:drx}.
		\item $\langcl(\DRX)\subset \langcl(\DTMFArej)$: The inclusion follows from Theorem~\ref{thm:glushkov}, it is proper due to the language $\{(\ta\tb)^{\frac{1}{2}i}\mid i\geq 0\}$, see Lemma~\ref{lem:drxvsdrmfa}.
		\item $\langcl(\DTMFArej)\subset \langcl(\DTMFA)$: The inclusion holds by definition. It is proper as the inclusion $\langcl(\DTMFAacc)\subset \langcl(\DTMFA)$ also holds by definition, and as $\langcl(\DTMFArej)$ and $\langcl(\DTMFAacc)$ are incomparable (see below).
		\item $\langcl(\DTMFA) \subset \langcl(\TMFA)$: Again, the inclusion holds by definition. Languages that separate the two classes are for example  the language of all $ww$ (where $w$ is from a non-unary alphabet, see Example~\ref{ex:jump}), and the language $\{\ta^{4i+1}\mid i\geq 0\} \cup \{\ta^{4^i}\mid i\geq 1\}$	from Lemma~\ref{lem:unarySep}.
		\item $ \langcl(\TMFA)=\langcl(\RX)$ is the statement of Theorem~\ref{thm:TMFAisMFA}.
		\item $\langcl(\DRX)$ and $\langcl(\REG)$ are incomparable: Again, we can use $\{(\ta\tb)^{\frac{1}{2}i}\mid i\geq 0\}$ from Lemma~\ref{lem:drxvsdrmfa}, and a non-regular $\DRX$-language, like  $\{w\tc w\mid w\in\{\ta,\tb\}^*\}$.
		\item $\langcl(\DTMFArej)$ and $\langcl(\DTMFAacc)$  are incomparable: Due to Proposition~\ref{prop:crash},  over a unary alphabet, for every non-regular language $L\in\langcl(\DTMFArej)$, we have $L\in\langcl(\DTMFArej)\mdif\langcl(\DTMFAacc)$, and $\overline{L}\in\langcl(\DTMFAacc)\mdif\langcl(\DTMFArej)$. Hence, we can choose e.\,g.\ $L=\{\ta^{n^2}\mid n\geq 0\}$ (which, as shown in Example~\ref{ex:drx}, is in $\langcl(\DRX)\subset \langcl(\DTMFArej)$) and its complement to show the two classes to be incomparable. 
		\item $\langcl(\DRX)$ and $\langcl(\DTMFAacc)$ are incomparable: Since $\langcl(\REG) \subseteq \langcl(\DTMFAacc)$, the language $\{(\ta\tb)^{\frac{1}{2}i}\mid i\geq 0\}$ is in the class $\langcl(\DTMFAacc)$, but, due to Lemma~\ref{lem:drxvsdrmfa}, not in $\langcl(\DRX)$. Moreover, 
		$$L = \{\ta^{n^2}\mid n\geq 0\} \in \langcl(\DRX),$$ see Example~\ref{ex:drx}, but if $L \in \langcl(\DTMFAacc)$, then, due to Theorem~\ref{thm:glushkov}, also $L \in \langcl(\DTMFArej)\cap\langcl(\DTMFAacc)$, which, by Proposition~\ref{prop:crash}, leads to the contradiction $L \in \langcl(\REG)$.
		%This follows with  $\{(\ta\tb)^{\frac{1}{2}i}\mid i\geq 0\}$  and $\{\ta^{n^2}\mid n\geq 0\}$.
	\end{enumerate}
	This concludes the proof.
\end{proof}

%\begin{figure}
%	%Now we are real language theorists:
%	\begin{center}
%				\begin{tikzpicture}[on grid]
%				
%		\node[] (DREG) {$\langcl(\DREG)$};		
%		\node[above right= of DREG, xshift=10mm,yshift=-2mm] (DRX) {$\langcl(\DRX)$};	
%		\node[below right= of DREG, xshift=10mm,yshift=2mm] (REG) {$\langcl(\REG)$};
%		\node[right= of DRX.east,xshift=-5mm] (DTMFArej) {$\langcl(\DTMFArej)$};	
%		\node[right = of REG.east,xshift=-5mm] (DTMFAacc) {$\langcl(\DTMFAacc)$};	
%		\node[right=of DREG.east,xshift=45mm] (DTMFA) {$\langcl(\DTMFA)$};			
%		\node[right of=DTMFA, xshift=3mm,anchor=west] (MFA) {$\langcl(\TMFA)=\langcl(\RX)$};
%		
%		\path[->,>=stealth']
%		(DREG) edge [] node {} (DRX.west)
%		(DREG) edge [] node {} (REG.west)
%		(DRX) edge [] node {} (DTMFArej)
%		(REG) edge [] node {} (DTMFArej.south west)
%		(REG) edge [] node {} (DTMFAacc)	
%		(DTMFArej.east) edge [] node {} (DTMFA.north west)
%		(DTMFAacc.east) edge [] node {} (DTMFA.south west)	
%		(DTMFA) edge [] node {} (MFA)	
%		;
%		\end{tikzpicture}
%	\end{center}\caption{The proper inclusions from Theorem~\ref{thm:hierarchyay}. Arrows point from sub- to superset.}\label{fig:hierarchyay2}
%\end{figure}
We can also use the examples from this section to show that the classes $\langcl(\DRX)$ and $\langcl(\DTMFArej)$ are not closed under most of the commonly studied operations on languages:
\begin{theorem}\label{thm:closure}
	 $\langcl(\DRX)$ and $\langcl(\DTMFArej)$ are  not closed under the following operations: union, concatenation, reversal,  complement, homomorphism, and inverse homomorphism.  $\langcl(\DRX)$ is also not closed under  intersection, and intersection with $\DREG$-languages.
\end{theorem}

%\subsection{Proof of Theorem~\ref{thm:closure}}
\begin{proof}
	The proofs for the operations that apply to both classes follow the same basic scheme: We start with one (or more) $\DRX$-language(s), and show that applying the operation yields a language that is not a $\DTMFArej$-language:
	\begin{description}
		\item[Union:] We use $L_1\df \{\ta^{4i+1}\mid i\geq 0\}$ and $L_2\df \{\ta^{4^i}\mid i\geq 1\}$, which are defined by the deterministic regex $\alpha_1 \df \ta (\ta^4)^*$, and $\alpha_2\df \ta^2\cdot \bind{x}{\ta^2} \cdot\bigl( \bind{y}{\rr{x}\cdot \rr{x}}\cdot \bind{x}{\rr{y}\cdot\rr{y}}  \bigr)^*$ (see Example~\ref{ex:drxCrazy}). Then $L_1\cup L_2 \notin\langcl(\DTMFArej)$, as shown in Lemma~\ref{lem:unarySep}.
		\item[Concatenation:] Define the deterministic regexes $\alpha_3 \df \ta^+$ and $\alpha_4\df \bind{x}{\ta^*}\cdot\tb\cdot\rr{x}$. Then $\lang(\alpha_4)=\{\ta^i\tb\ta^i\mid i\geq 0\}$, and $\lang(\alpha_3)\cdot \lang(\alpha_4) = \{\ta^i \tb \ta^j \mid i>j\geq 0\}$. As shown in Example~\ref{ex:jump2}, $\lang(\alpha_3)\cdot \lang(\alpha_4)$ is not a $\DTMFArej$-language.
		\item[Reversal:] Let $\alpha_5 \df \bind{x}{\ta^*}\cdot\tb\cdot\rr{x}\cdot \ta^+$. Then $\alpha_5\in\DRX$, and $\lang(\alpha_5)=\{\ta^j \tb \ta^i \mid i>j\geq 0\}$. Reversing $\lang(\alpha_5)$ again gives us the language from Example~\ref{ex:jump2}.
		\item[Complement:] This follows directly from Proposition~\ref{prop:crash}. Consider e.\,g.\ $\{\ta^{n^2}\mid n\geq 0\}$.
		\item[Homomorphism:] Let $\alpha_6\df ((\tc\cdot \alpha_1)\ror(\td\cdot \alpha_2))$. Then $\alpha_6\in \DRX$, but $h(\lang(\alpha_6)) = L_1\cup L_2$ for the morphism $h$ that is defined by $h(x)\df x$ if $x\in\{\ta,\tb\}$ and $h(x)\df \emptyword$ if $x\in\{\tc,\td\}$.
		\item[Inverse homorphism:] Define a morphism $g$ by $g(\ta)\df g(\tb) \df  \ta$, and $g(\tc)\df \tb$. Then let $L_{7}\df g^{-1}(\lang(\alpha_4))=\{u\cdot\tc\cdot v\mid u,v\in\{\ta,\tb\}^*, |u|=|v|\}$. We use Lemma~\ref{lem:jump} to show that $L_{7}\notin\langcl(\DTMFArej)$. Assume to the contrary that it is, and choose $m\geq 1$. Then there exist $n\geq m$ and words $p_n, v_n$ that satisfy the conditions of Lemma~\ref{lem:jump}. We now distinguish the following cases: First, assume that $p_n$ does not contain the letter $\tc$. Then choose a $d\in\{\ta,\tb\}$ that is not the first letter of $v_n$, and define $u\df d\cdot \tc\cdot \ta^{|p_n|+1}$. Then $p_n u \in L_7$; but as $v_n$ is not a prefix of $u$, this contradicts Lemma~\ref{lem:jump}. Now assume that $p_n$ contains $\tc$. Then $p_n = w_1 \tc w_2$ with $w_1,w_2\in\{\ta,\tb\}^*$, and $|w_1| \geq  |w_2| + n\geq |w_2| + m$. Again, choose $d\in\{\ta,\tb\}$ such that $d$ is not the first letter of $v_n$, and define $u \df d \cdot \ta^{|w_1| - |w_2| -1}$. Then $p_n u =  w_1 \tc w_3$ for $w_3 = w_2 d\cdot \ta^{|w_1| - |w_2| -1}$, and $|w_3| = |w_1|$. Hence, $p_n u\in L_{7}$, but as $v_n$ is not a prefix of $u$, this contradicts Lemma~\ref{lem:jump}.
	\end{description}
	\subparagraph*{Intersection:} In order to show both claims on the intersection of $\lang(\DRX)$, it suffices to show that we can obtain a language that is not a $\DRX$-language by intersecting  two deterministic regular languages. Accordingly, we define deterministic regular expressions $\beta_1\df (\ta (\tb\ror\emptyword))^*$  and $\beta_2 \df \emptyword\ror\bigl(\ta(\tb(\ta\ror \emptyword))^*\bigr)$ (these expressions have been obtained by very minor modifications to the expressions that Caron, Han, Mignot~\cite{car:gen} use to show that $\lang(\DREG)$ is not closed under intersection).
	
	Let  $L_{8}=\{(\ta \tb)^{\frac{1}{2}i}\mid i\geq 0\}$. To show that $L_{8}= \lang(\beta_1)\cap \lang(\beta_2)$, we follow the approach from~\cite{car:gen},  and first consider $\glush(\beta_1)$ and the corresponding minimal incomplete $\DFA$:
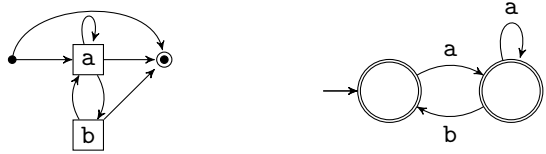
	\begin{center}
		\begin{tikzpicture}[node distance=10mm,on grid,>=stealth',auto, 
		state/.style={rectangle,draw=black,inner sep=0pt,minimum size=4mm}]
		\node         (start)                 				{};
		\node[state]  (a)     [right=of start] 				{$\ta$};
		\node[state]  (b)     [below=of a]     	{$\tb$};
		\node 				(end)   [right=of a]	   	{};
		
		\draw[fill=black] (start) circle (0.5mm);
		\draw[fill=black] (end) circle (0.5mm);
		\draw 						(end) circle (1mm);
		
		\path[->]
		(start.center) edge node[] {} (a)
		(start.center) edge[bend left=80] node {} (end)
		(a) edge node{} (end)
		(a) edge[loop above] node{} (a)
		(b) edge node{} (end)
		(a) edge[bend left] node {} (b)
		(b) edge[bend left] node {} (a)
		;
		\end{tikzpicture}
		\hspace{15mm}
		\begin{tikzpicture}[node distance=16mm,on grid,>=stealth',auto, 
		state/.style={circle,draw=black,inner sep=1pt,minimum size=8mm}]
		
		\node[state,initial,initial by arrow,initial text={},accepting]         (qb) {};
		\node[state,accepting, right=of qb]         (qa) {};
		
		\path[->] 
		(qb) edge[bend left] node[] {$\ta$} (qa)
		(qa) edge[loop above] node[] {$\ta$} (qa)
		(qa) edge[bend left] node[] {$\tb$} (qb)
		;
		\end{tikzpicture}
	\end{center}
	Likewise, we consider $\glush(\beta_2)$ and the corresponding minimal incomplete $\DFA$ (which merges the two states for $\ta$):
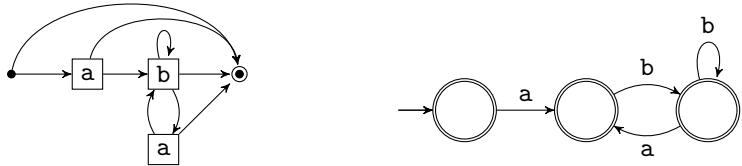
	\begin{center}
		\begin{tikzpicture}[node distance=10mm,on grid,>=stealth',auto, 
		state/.style={rectangle,draw=black,inner sep=0pt,minimum size=4mm}]
		\node         (start)                 				{};
		\node[state]  (a1)     [right=of start] 				{$\ta$};
		\node[state]  (b)     [right=of a1]     	{$\tb$};
		\node[state] (a2) [below=of b] {$\ta$}; 
		\node 				(end)   [right=of b]	   	{};
		
		\draw[fill=black] (start) circle (0.5mm);
		\draw[fill=black] (end) circle (0.5mm);
		\draw 						(end) circle (1mm);
		
		\path[->]
		(start.center) edge node[] {} (a1)
		(start.center) edge[bend left=80] node {} (end)
		(a1) edge[bend left=80] node {} (end)
		(a1) edge node{} (b)
		(a2) edge node{} (end)
		(b) edge node{} (end)
		(a2) edge[bend left] node {} (b)
		(b) edge[bend left] node {} (a2)
		(b) edge[loop above] node {} (b)
		;
		\end{tikzpicture}
		\hspace{15mm}
		\begin{tikzpicture}[node distance=16mm,on grid,>=stealth',auto, 
		state/.style={circle,draw=black,inner sep=1pt,minimum size=8mm}]
		
		\node[state,initial,initial by arrow,initial text={},accepting]         (q0) {};
		\node[state,accepting, right=of q0]         (qa) {};
		\node[state,accepting, right=of qa]         (qb) {};
		
		\path[->] 
		(q0) edge[] node[] {$\ta$} (qa)
		(qb) edge[bend left] node[] {$\ta$} (qa)
		(qa) edge[bend left] node[] {$\tb$} (qb)
		(qb) edge[loop above] node[] {$\tb$} (qb)
		;
		\end{tikzpicture}		
	\end{center}
	Now it is easily seen that $L_{8}= \lang(\beta_1)\cap \lang(\beta_2)$. From Lemma~\ref{lem:drxvsdrmfa}, we know that $L_{8}\notin\langcl(\DRX)$. Hence, the class of deterministic regex languages is not closed under intersection with deterministic regular languages, which also implies that it is not closed under intersection.
\end{proof}

We leave open whether $\langcl(\DTMFArej)$ is closed under intersection (with itself or with $\langcl(\DREG)$), but we conjecture that this is not the case. In this regard, note that while $\langcl(\TMFA)$ is closed under intersection with $\langcl(\REG)$ (as show in ~\cite{sch:cha}), it is open whether $\langcl(\TMFA)$ closed under intersection with itself. 
%\todomcom{The missing reference for the non-closure of the class under intersection was from a time when we still thought that your paper talks about the same class as CSY. For actual regex, closure under intersection is open.} 

We also leave open whether $\langcl(\DRX)$ and $\langcl(\DTMFArej)$ are closed under Kleene plus or star.
%!TEX root=det_SIAM.tex
\section{Static Analysis}\label{sec:static}
In this section, we examine a restriction  $\DRX$ and $\DTMFA$, which we motivate with the following observation:
As shown by Carle and Narendran~\cite{car:one}, the intersection emptiness problem for regex is undecidable\footnote{Although that proof refers to a subclass of our definition of regex, see Section~\ref{sec:choices}, it directly translates to $\RX$}. For $\DRX$, that proof cannot be used, but the result still holds (and by Theorem~\ref{thm:glushkov}, this extends to $\DTMFA$):
\begin{theorem}\label{thm:intersectundec}
	Given $\alpha,\beta\in\DRX$, it is undecidable whether $\lang(\alpha)\cap\lang(\beta)=\emptyset$.
\end{theorem}

%\subsection{Proof of Theorem~\ref{thm:intersectundec}}
\begin{proof}
	We show this with a reduction from Post's Correspondence Problem (PCP, for short). Let $(u_1,v_1),\ldots,(u_k,v_k)\in \Sigma^*\times\Sigma^*$, $k\geq 1$, be a PCP instance. Our goal is to construct $\alpha,\beta\in\DRX$ such that $\lang(\alpha)\cap\lang(\beta)\neq \emptyset$ if and only if there exists a sequence $i_1,\ldots,i_n$, $n\geq 1$ and $1\leq i_j \leq k$, such that $u_{i_1}\cdots u_{i_n} = v_{i_1}\cdots v_{i_n}$. To do so, we first introduce an alphabet $A\df\{a_1\ldots,a_k\}$ such that $A$, $\Sigma$, and $\{\sepA,\sepB,\sepC\}$ are pairwise disjoint (at the end of the proof, we discuss how this construction can be adapted to binary terminal alphabets). We then define
	\begin{align*}
		\alpha &\df \Bigl(\bigror_{i=1}^k a_i \sepA u_i \bind{x}{\Sigma^*} \sepA v_i \bind{y}{\Sigma^*} \sepB \rr{x} \sepA \rr{y} \sepC   \Bigr)^*,\\
		\beta &\df A\sepA \bind{z}{\Sigma^+} \sepA \rr{z} \sepB \Bigl( \bind{x}{\Sigma^+} \sepA \bind{y}{\Sigma^+} \sepC A\sepA \rr{x} \sepA \rr{y} \sepB   \Bigr)^* \sepA\sepC. 
	\end{align*}	
	To see that $\alpha$ is deterministic, note that the disjunction ranges over the letters from $A$. For $\beta$, we observe that after each iteration of the starred subexpression, we read either a letter from $\Sigma$, and start a new iteration, or $\sepA$, which means that this was the last iteration.
	
	We now claim that $w\in \lang(\alpha)\cap\lang(\beta)$ if and only if there exist an $n\geq 1$ and $i_1,\ldots,i_n$ with $1\leq i_j \leq k$ such that  $u_{i_1}\cdots u_{i_n} = v_{i_1}\cdots v_{i_n}$ and
	$w = w_1 \sepC w_2\sepC \cdots w_n \sepC$, where 
	$$w_j = a_{i_j} \sepA u_{i_j}\cdots u_{i_n} \sepA v_{i_j}\cdots v_{i_n} \sepB u_{i_{j+1}}\cdots u_{i_n} \sepA v_{i_{j+1}}\cdots v_{i_n}.$$
	Take note that $w_n$ always ends on $\sepB\sepA$.
	Informally explained, $w$ encodes how a  solution of the PCP instance is constructed, where the finished solution is in $w_1$, and the start of the construction is at $w_n$. Starting at $w_1$, the sequence of $w_j$ can be understood as splitting off pairs of prefixes $(u_j,v_j)$ from the solution, where each word $w_j$ also encodes which tuple $(u_j,v_j)$ is processed (by using the preceding symbol $a_j$ as a marker), and the words before and after the pair is split off (to the left and right of $\sepB$, respectively). 
	
	Here, $\alpha$ ensures that in each $w_j$, $u_j$ and $v_j$ are split off correctly, while $\beta$ ensures that the ``after'' words of $w_j$ are the ``before'' words of $w_{j+1}$. Hence, such a $w$ exists if and only if the instance of the PCP has a solution. As the existence of the latter is undecidable (see e.\,g.\ Hopcroft and Ullman~\cite{hop:int}), deciding  $\lang(\alpha)\cap\lang(\beta)\neq \emptyset$ is also undecidable.
	
	To adapt the construction to a binary alphabet (say, $\{\ta,\tb\}$), we use a morphism $h\colon (A\cup\Sigma\cup\{\sepA,\sepB,\sepC\})^*\to\{\ta,\tb\}^*$ that is defined as follows (we assume an arbitrary ordering on $\Sigma$ with $\Sigma=\{b_1,\ldots,b_{|\Sigma|}\}$):
	\begin{itemize}
		\item $h(a_i)\df \ta \tb^i \ta$ for all $a_i \in A$,
		\item $h(b_i)\df \ta \tb^{i} \ta$, %where we assume an arbitrary ordering on $\Sigma$ with $\Sigma=\{b_1,\ldots,b_{|\Sigma|}\}$,
		\item $h(\sepA)\df \tb\ta\tb$, $h(\sepB)\df \tb\ta^2\tb$, and $h(\sepC)\df \tb\ta^3\tb$.
	\end{itemize}
	If we apply $h$ to $\alpha$ and $\beta$ by applying $h$ to each terminal, we obtain regex $h(\alpha)$ and $h(\beta)$ such that $\lang(h(\alpha))\cap \lang(h(\beta))\neq \emptyset$ if and only if the instance of the PCP has a solution. The only problem is that these regex are not deterministic, as there are disjunctions that start with the same terminal letter. But each of these disjunctions can be rewritten into a deterministic disjunction by nesting the branches. For example, consider the disjunction $(a_1\ror a_2 \ror a_3)$. Using $h$, this becomes $(\ta\tb\ta\ror \ta\tb^2\ta \ror \ta\tb^3\ta)$, which is  not deterministic, but can be rewritten to the equivalent
	$(\ta\tb (\ta \ror (\tb  (\ta \ror \tb\ta)   )) )$.
	
	Now, note that if we apply this rewriting to the disjunctions  to $h(\alpha)$ and $h(\beta)$ (including the disjunctions that are hidden in shorthand notations $A$ and $\Sigma$), we obtain deterministic regex. In particular, note that Kleene plus and Kleene star are only used on elements of $A\cup\Sigma$, and are always followed by either $\sepA$ or $\sepC$. As the encodings of the former start with $\ta$, while the encodings of the latter start with $\tb$, rewriting the disjunctions is enough to ensure determinism.
\end{proof}

As a consequence, $\DTMFA$~intersection emptiness problem is also undecidable. Theorem~\ref{thm:intersectundec} applies even to very restricted $\DRX$, as no variable binding contains a reference to another variable, $|\var(\alpha)|=2$, and $|\var(\beta)|=3$. Hence, bounding the number of variables does not make the problem decidable. Instead, the key part seems to be that the variables occur under Kleene stars, which means that they can be reassigned an unbounded amount of times. Following similar observations, Freydenberger and Holldack~\cite{fre:doc} introduced the following concept: A regex is \emph{variable-star-free (vstar-free)} if each of its plussed sub-regexes contains neither variable references, nor variable bindings. Analogously, we call a  $\TMFA$  \emph{memory-cycle-free} if it contains no cycle with a \emph{memory transition} (a transition in a $\TMFA$ that is a memory recall, or that contains memory actions other than $\unchanged$).
Let $\fRX$ be the set of all vstar-free regex, and $\fDRX= \fRX\cap \DRX$. Let $\fTMFA$ be the set of all memory-cycle-free $\TMFA$, and define $\fDTMFA$, $\fTMFArej$,\ldots\ analogously. The proof of Theorem~\ref{thm:glushkov} allows us to conclude that $\glush(\alpha)\in\fDTMFA$ holds for every $\alpha\in\fDRX$. Likewise, we can use the proof of Theorem~\ref{thm:TMFAisMFA} to conclude $\langcl(\fTMFA)=\langcl(\fRX)$.
Note that for $\emptyword$-free $\fDTMFA$, the membership problem can be decided in time $O(|Q|+|w|)$, as the preprocessing step of Theorem~\ref{thm:dtmfamembership} is not necessary (as only a bounded number of variable references is possible in each run). Likewise, we can  drop the factor $k$ from Theorem~\ref{thm:drxmembership} when restricted to $\fDRX$.

As shown by Freydenberger~\cite{fre:splog}, it is decidable in $\PSPACE$ whether $\bigcap_{i=1}^n \lang(\alpha_i)=\emptyset$ for  $\alpha_1,\ldots,\alpha_n\in\fRX$. By combining the proof for this  with some ideas from another construction from~\cite{fre:splog}, we encode the intersection emptiness problem for $\fTMFA$ in the \emph{existential theory of concatenation with regular constraints} (a $\PSPACE$-decidable, positive logic on words, see Diekert~\cite{die:mak}, Diekert, Je{\.z}, Plandowski~\cite{die:fin}). This yields the following: 
\begin{theorem}\label{thm:intersect}
	Given $M_1, \ldots,M_n\in\fTMFA$,  we can decide $\bigcap_{i=1}^n \lang(M_i)=\emptyset$  in $\PSPACE$.
	The problem is  $\PSPACE$-hard, even  if restricted to $\lang(\alpha)\cap\lang(\beta)$,  $\alpha\in\fDRX$ and $\beta\in\DREG$ (if the size of $\Sigma$ is not bounded), or to $\lang(\alpha)\cap\lang(M)$,  $\alpha\in\fDRX$ and $M\in\DFA$.
\end{theorem}

\begin{proof}
	We begin with the first lower bound:  As shown by Martens, Neven, and Schwentick~\cite{mar:com} (Theorem 3.10), the intersection emptiness problem for deterministic regular expressions is $\PSPACE$-complete (if $|\Sigma|$ is not bounded; the paper does not discuss the unbounded case, and the proof cannot be adapted directly). This problem is defined as follows: Given $\beta_1,\ldots,\beta_n\in \DREG$ for some $n\geq 2$, is $\bigcap_{i=1}^{n}\lang(\beta_i)=\emptyset$? We use a new terminal letter $\sepA\notin\Sigma$, and define 
	\begin{align*}
		\alpha&\df \bind{x}{\Sigma^*}\sepA \bigl(\rr{x} \sepA\bigr)^{n-1}\\
		\beta &\df \beta_1 \sepA \beta_2 \sepA \cdots \beta_n \sepA.
	\end{align*}
	First, observe that $\alpha$ and $\beta$ are deterministic, as $\sepA\notin\Sigma$. Now  $w\in\lang(\alpha)$ holds if and only if $w = \bigl(\hat{w}\sepA\bigr)^n$ for some $\hat{w}\in \Sigma^*$, and  $w\in \lang(\beta)$ if and only there exist $w_1,\ldots,w_n\in\Sigma^*$ with $w_i\in\lang(\beta_i)$ and $w = w_1 \sepA w_2\sepA \cdots w_n\sepA$. Hence, $(\lang(\alpha)\cap\lang(\beta))\neq\emptyset$ if and only if $\bigcap_{i=1}^{n}\lang(\beta_i)\neq\emptyset$. As this problem is $\PSPACE$-complete, deciding $(\lang(\alpha)\cap\lang(\beta))\stackrel{?}{=}\emptyset$ is $\PSPACE$-hard.
	
	The second lower bound is a reduction from the  intersection emptiness problem for $\DFA$, which is defined as follows: Given $M_1,\ldots,M_n\in\DFA$ with $n\geq 2$, is there a $w\in\Sigma^*$ with  $w\in \lang(M_i)$ for all $1\leq i\leq n$? This problem is $\PSPACE$-complete (cf.\ Kozen~\cite{koz:low}). We take a new terminal symbol $\sepA\notin\Sigma$, define $\alpha$ as above, and choose $M$ to be the $\DFA$ for the language $\lang(M_1)\sepA\lang(M_2)\sepA \cdots \sepA \lang(M_n)$ (as $\sepA$ does not occur in the languages of the $\DFA$, this is trivially possible). The reasoning continues as above; but as the $\DFA$ can be defined on a binary alphabet, this proof does not require an unbounded alphabet.
	
	The upper bound takes more work, including further definitions. Our goal is to encode the intersection emptiness problem for $\fTMFA$ in $\ECreg$, the existential theory of concatenation with regular constraints, which we now introduce (for a more detailed definition and examples on $\ECreg$, see for example Freydenberger~\cite{fre:splog}).
	
	One of the basic elements of $\ECreg$-formulas are word equations: A \emph{pattern} is a word $\alpha\in(\Sigma\cup\Xi)^*$, and a \emph{word equation} is a pair of patterns $(\eta_L,\eta_R)$, which can also be written as $\eta_L=\eta_R$ (hence the name equation).   A \emph{pattern substitution} is a homomorphism $\sigma\colon (\Xi\cup\Sigma)^*\to\Sigma^*$ with $\sigma(a)=a$ for all $a\in\Sigma$.  It is a \emph{solution} of a word equation $(\eta_L,\eta_R)$ if $\sigma(\eta_L)=\sigma(\eta_R)$, and we write this as $\sigma\models (\eta_L,\eta_R)$. Less formally, a pattern substitution replaces all variables with terminal words (where multiple occurrences of the variable have to be substituted in the same way), and it is a solution of an equation if both sides have the same terminal word as a result.
	
	The other basic building block are \emph{constraint symbols}: For every $\eNFA$~$A$ and every $x\in \Xi$, we can use a constraint symbol $\const{A}{x}$. A pattern substitution $\sigma$ satisfies $\const{A}{x}$ if $\sigma(x)\in\lang(A)$. We write this as $\sigma\models\const{A}{x}$.
	
	The \emph{existential theory of concatenation with regular constraints} $\ECreg$ is obtained by combining word equations and constraint symbols using $\land$, $\lor$ and existential quantification over variables. Semantics  are defined canonically: We have $\sigma\models (\varphi_1\land\varphi_2)$ if  $\sigma\models (\varphi_1)$ and $\sigma\models\varphi_2$; and $\sigma\models (\varphi_1\lor\varphi_2)$ if $\sigma\models (\varphi_1)$ or $\sigma\models\varphi_2$. Finally, $\sigma\models (\exists x\colon \varphi)$ if there exists a $w\in\Sigma^*$ such that $\sigma_{[x\to w]}\models \varphi$, where the pattern substitution $\sigma_{[x\to w]}$ is defined by $\sigma_{[x\to w]}(x)\df w$, and $\sigma_{[x\to w]}(y)\df\sigma(y)$ if $y\neq x$. In slight abuse of notation, we also write $w\models \varphi(x)$ if $\sigma\models \varphi(x)$ holds for the pattern substitution $\sigma(x)\df w$.
	
	For example, let $\varphi(x)\df \exists y\colon \bigl( (x=y\tb y) \land \const{A}{y} \bigr)$, where $A$ is an $\NFA$ with $\lang(A)=\{\ta^*\}$. Then $w\models\varphi(x)$ if and only if $w=\ta^n \tb\ta^n$ for some $n\geq 0$.	
	
	Given an $\ECreg$-formula $\varphi$, deciding the existence of a pattern substitution $\sigma$ with $\sigma\models\varphi$ is $\PSPACE$-complete, cf.\ Diekert\cite{die:mak}.
	
	We first prove the claim only for automata with rejecting trap states (as we shall see further down, the case for accepting trap states requires only a small modification).
	Before we proceed to the main idea of the construction, we first take a closer look at the accepting runs of  memory-cycle-free $\TMFA$.
	
	Let $M\in\fTMFArej$ with $M=(Q,\Sigma,\delta,q_0,F)$ and memories $\{1,\ldots,k\}$, and consider any accepting run of $M$. As $M$ is memory cycle free, whenever it takes a memory transition from a state $p$ to a state $q$, we know that $p$ cannot occur anywhere else in the run. Otherwise, it would be possible to repeat the memory transition from $p$ to $q$ arbitrarily often, which would contradict the assumption that $M$ is memory cycle free. Hence, we know that every accepting run of $M$ can use at most $|Q|-1$ memory transitions. 
	
	This allows us to condense any accepting run of $M$ by considering only its memory transitions. Formally,   for some $0\leq \ell < |Q|$, we define  a \emph{condensed run} $\kappa = (\vec{q},\vec{p},\vec{\tau})$ of length $\ell$ as follows: 
	\begin{enumerate}
		\item $\vec{q}=(q_0,\ldots,q_{\ell})$ is a sequence of states, where $q_0$ is the starting state of $M$, 
		\item $\vec{p}=(p_0,\ldots,p_{\ell})$ is a sequence of states, with $p_{\ell}\in F$,
		\item $\vec{\tau}= (\tau_1,\ldots,\tau_{\ell})$ is a sequence of memory transitions, where for each $1\leq i \leq \ell$, either
		\begin{align*}
			\tau_i &= (p_i, x_i, q_{i+1},  s_{i,1}, \ldots, s_{i,k}) \text{ with $x_i\in\{1,\ldots, k\}$, or }\\
			\tau_i &= (p_i, b_i, q_{i+1}, s_{i,1}, \ldots, s_{i,k}) \text{ for some $b_i\in(\Sigma\cup\{\emptyword\})$}
		\end{align*}
		and $s_{i,j}\in\{\open, \close, \reset, \unchanged\}$, $1\leq j\leq k$. In the second case, we also require that there is at least one $j$ with $s_{i,j}\neq \unchanged$.
		\item $p_i$ is reachable from $q_i$ without using memory transitions for all $0\leq i \leq \ell$.
	\end{enumerate}
	The intuition behind this construction is that we condense the  run to a sequence of states  $(q_0, p_0, q_1, p_1, \ldots, q_{\ell}, p_{\ell})$ that only contains the starting state $q_0$, a final state $p_{\ell}$, and the states before and after each memory transition (as $M$ runs from $q_i$ to $p_i$ only without memory transitions, and each memory transition $\tau_i$ takes the automaton from $p_i$ to $q_{i+1}$).
	
	As we shall see, each condensed run $\kappa$ can be converted in polynomial time into an $\ECreg$ formula $\varphi_{\kappa}(w)$ that defines exactly the language of all $w\in\lang(M)$ for which there is an accepting run of $M$ that can be condensed to $\kappa$. In particular, the parts of the run between each pair of states $q_i$ and $p_i$ (which involve no memory transitions) shall be handled by appropriate regular constraints.
	
	Now, given $M_1,\ldots,M_n\in \fTMFArej$, we proceed as follows to decide whether $\bigcap_{i=1}^{n}\lang(M_i)=\emptyset$. First, we guess a condensed run $\kappa_i$ for each $M_i$ (as the length of each sequence is bounded by the number of states of $M_i$ and as $\PSPACE=\NPSPACE$, this is allowed). Next, we convert each $\kappa_i$ into an $\ECreg$-formula $\varphi_i(w)$ (as we shall see, this is possible in polynomial time). Finally, we combine these into the formula $\varphi(w)\df \bigland_{i=1}^{n}\varphi_i(w)$, and decide whether $\varphi$ is satisfiable (as mentioned above, this is possible in $\PSPACE$, see Diekert~\cite{die:mak}). As $\varphi$ is satisfiable if and only if there exists a $w\in \bigcap_{i=1}^{n}\lang(M_i)$, this proves the claim for $\fTMFArej$ (as mentioned above, we shall discuss the case of accepting failure states at the end of the proof).
	
	We now discuss how to construct $\varphi_{\kappa}$ from $\kappa$. Consider a word $w\in\lang(M)$, and the condensed run $\kappa$ of length $\ell$ for any accepting run of $M$ on $w$. Then $w$ can be decomposed into $w= u_0 v_1 u_1 \cdots v_{\ell} u_{\ell}$ with $u_i, v_i\in\Sigma^*$ such that $u_i$ is the word that $M$ consumes when processing from $q_i$ to $p_i$ (without using memory tranisitons), and $v_i$ is the word that is consumed when processing from $p_{i-1}$ to $q_i$ (using the memory transition $\tau_i$). This is illustrated by the following picture:
	\begin{center}
		\begin{tikzpicture}[on grid, node distance=16mm]
		
		\node (q0) {$q_0$};
		\node[right=of q0] (p0) {$p_0$};
		\node[right=of p0] (q1) {$q_1$};
		\node[right=of q1] (p1) {$p_1$};
		\node[right=of p1] (dots) {\hphantom{h}$\cdots$\hphantom{h}};
		\node[right=of dots] (plm1) {$p_{\ell-1}$};
		\node[right=of plm1] (ql) {$q_{\ell}$};
		\node[right=of ql] (pl) {$p_{\ell}$};		
		\path[->]
		(q0) edge[] node[above] {$u_0$} (p0)
		(p0) edge[] node[above] {$v_1$} node[below] {$\tau_1$} (q1)
		(q1) edge[] node[above] {$u_1$} (p1)			
		(p1) edge[] node[above] {$v_2$} node[below] {$\tau_2$} (dots)
		(dots) edge[] node[above] {$u_{\ell-1}$} (plm1)
		(plm1) edge[] node[above] {$v_{\ell}$} node[below] {$\tau_{\ell}$} (ql)
		(ql) edge[] node[above] {$u_{\ell}$} (pl)			
		;
		\end{tikzpicture}
	\end{center}
	Following this intuition, we define
	\begin{multline*}
		\varphi_{\kappa}(w)\df \exists u_0,\ldots,u_{\ell},v_1,\ldots,v_{\ell}\colon\\
		(w=u_0 v_1 u_1 \cdots v_{\ell} u_{\ell})  \land \bigland_{i=0}^{\ell}\const{M_{q_i,p_i}}{u_i} \land \bigland_{i\in T} (v_i = b_i) \land \bigland_{i\in R} (v_i = \eta_i),
	\end{multline*}
	where the following holds:
	\begin{enumerate}
		\item for $p,q\in Q$, $M_{q,p}$ is the $\eNFA$ that is obtained from $M$ by removing all memory transitions, using $q$ as starting and $p$ as only finite state,
		\item $T\subseteq  \{1,\ldots,\ell\}$ is the set of all $i$ such that $\tau_i$ is not a memory recall  (i.\,e., $\tau_i$ consumes a terminal symbol $b_i$),
		\item $R\subseteq \{1,\ldots,\ell\}$ is the set of all $i$ such that $\tau_i$ is a memory recall  (i.\,e., $\tau_i$ recalls memory $x_i$),
		\item 	for each $i\in R$, we define $\eta_i$ to describe  the current content for $x_i$ (as the memory actions are completely determined by the $\tau_j$, this is directly possible by checking the $\tau_j$ with $j\leq i$). There are three possible cases:
		\begin{enumerate} 
			\item if $x_i$ was never changed in a memory transition $\tau_j$ with $j < i$, then its value defaults to $\emptyword$, and we define $\eta_i \df \emptyword$,
			\item if $x_i$ was reset in a memory transition $\tau_j$ with  $j< i$, and not changed between $\tau_j$ and $\tau_i$, we define $\eta_i \df \emptyword$,
			\item otherwise, there exist well-defined $1\leq j < j' \leq i$ such that $x_i$ was opened in transition $\tau_j$ and closed in $\tau_{j'}$, and not changed between $\tau_j$ and $\tau_i$. Hence, we define $\eta_i \df v_j u_j \cdots u_{j'}$.
		\end{enumerate}
	\end{enumerate}
	The constraints $\const{M_{q_i,p_i}}{u_i}$ check that each $u_i$ conforms to a sequence of transitions that takes $M$ from $q_i$ to $p_i$, without using memory transitions. The conjunction over the $i\in T$ ensures that the memory actions that are not memory recalls consume terminals correctly, and the conjunction over the $i\in R$ ensures that each memory recall refers to the right part of the consumed input. Hence, for all $w\in\Sigma^*$, $w\models\varphi_{\kappa}$ if and only if $w\in\lang(M)$, and there is an accepting run of $M$ on $w$ that can be condensed to $\kappa$. It is easily seen that $\varphi_{\kappa}$ can be constructed in polynomial time (and, hence, its size is polynomial in $|Q|$): We need to construct $\ell+1\leq |Q|$ automata $M_{q_i,p_i}$, each of which has at most $|Q|$ states and at most $|Q|^2$ transitions. Determining each $\eta_i$ is also possible in time $O(|Q|)$, be checking the previous transitions~$\tau_j$.
	
	This concludes the proof for the case of $M_i\in \fTMFArej$. For the case where the failure state is accepting, we need to add a small extension: Instead of only considering condensed runs for runs that reach a final state, we also need to consider the runs that end in a memory recall that fails. Hence, we consider a condensed run $\kappa$ of length $\ell< |Q|$, such that $\tau_{\ell}$ is a memory recall transition for a variable $x\in\{1,\ldots,k\}$. Then we construct $\varphi_{\kappa}$ almost as explained above. The only difference is that we replace the equation $(v_{\ell} = \eta_{\ell})$ in the conjunction over the elements of $R$ with the following formula:
	\begin{multline*}
		\Bigl(\exists z\colon (\eta_{\ell} = v_{\ell}z)  \land \const{A}{v_{\ell}} \land \const{A}{z}    \Bigr) \\ \lor \Bigl(\biglor_{a\in\Sigma} \biglor_{b\in \Sigma\mdif\{a\}}\exists y,z_1,z_2\colon (v_{\ell} = yaz_1) \land (\eta_{\ell} = ybz_2)  \Bigr),
	\end{multline*}
	where $A$ is the minimal $\DFA$ for $\lang(A)=\Sigma^+$. The left part of this disjunctions describes all cases where $v_{\ell}$ is non-empty and a proper prefix of the content of $x$; the right part describes all cases where $v_{\ell}$ and the content of $x$ differ at at least one position (recall that the content of $x$ at this point of the run is represented by $\eta_{\ell}$). Hence, this formula describes all cases where this condensed run ends in a memory recall failure.
	
	This shows that the approach also works for $M\in\fTMFAacc$. Hence, given $M_1,\ldots,M_k\in\fTMFA$, we can decide whether the intersection of all $\lang(M_i)$ is empty by guessing a condensed run $\kappa_i$ for each $M_i$. If $M_i\in\fTMFArej$, we only need to consider runs that end in final states; if $M_i\in\fTMFAacc$, we also need to consider runs that end in memory recall failures. Either way, the length of $\kappa_i$ is bounded by the number of states in $M_i$. We then transform in polynomial time each $\kappa_i$ into an $\ECreg$-formula $\varphi_i$, and combine these into $\varphi \df \bigland_{i=1}^n \varphi_i$. Then $w\models \varphi$ if and only if $w\in\lang(M_i)$ for all $i$; and as satisfiability of $\ECreg$-formulas can be decided in $\PSPACE$, intersection emptiness is also decidable in $\PSPACE$. As we already showed hardness at the very beginning of this proof, we conclude that the problem is $\PSPACE$-complete.
\end{proof}
%\todomcom{Integrated the proposition.}
The unbounded size of $\Sigma$ comes from the  $\PSPACE$-hardness of the intersection emptiness problem for $\DRX$ by Martens~et~al.~\cite{mar:com}, which has the same requirement. Using the existential theory of concatenation for the upper bound might seem conceptually excessive, considering how complicated even the satisfiability problem for word equations is  (see Diekert~\cite{die:mak,die:mor} for a detailed and a recent survey).  But even intersection emptiness for $\fDRX$ is at least as hard as the satisfiability problem for word equations: 
\begin{proposition}\label{prop:wehard}
	Given a word equation $\eta$ over $\Sigma$, we can construct in linear time $\alpha_L,\alpha_R\in\fDRX$ over $\Sigma\cup\{\sepA\}$ such that $\lang(\alpha_L)\cap\lang(\alpha_R)\neq \emptyset$ holds if and only if $\eta$ has a solution.
\end{proposition}
\begin{proof}
	Let $\eta=(\eta_L,\eta_R)$, and assume that $x_1,\ldots,x_k$ are the variables that occur in $\eta$. Let $\sepA$ be a new terminal letter, $\sepA\notin\Sigma$, and define
	\begin{align*}
	\alpha_L \df \bind{x_1}{\Sigma^*} \sepA \bind{x_2}{\Sigma^*}\sepA \cdots \bind{x_k}{\Sigma^*}\sepA \beta_L,\\
	\alpha_R \df \bind{x_1}{\Sigma^*} \sepA \bind{x_2}{\Sigma^*}\sepA \cdots \bind{x_k}{\Sigma^*}\sepA \beta_R,
	\end{align*}
	where $\beta_L$ and $\beta_R$ are obtained from $\eta_L$ and $\eta_R$ (respectively) by replacing each occurrence of a variable $x_i$ with the reference $\rr{x_i}$. As $\sepA\notin\Sigma$, and as $\beta_L$ and $\beta_R$ consist only of a chain of terminals and variable references, $\alpha_L,\alpha_R\in\DRX$. Furthermore, $w\in \lang(\alpha_L)\cap \lang(\alpha_R)$ holds if and only if there is a homomorphism  $\sigma\colon (\Sigma\cup\Xi)^* \to \Sigma^*$ with $\sigma(a)=a$ for all $a\in\Sigma$ such that 
	\begin{align*}
	w &= \sigma(x_1)\sepA \sigma(x_2)\sepA \cdots \sigma(x_k)\sepA \sigma(\eta_L)\\
	&= \sigma(x_1)\sepA \sigma(x_2)\sepA \cdots \sigma(x_k)\sepA \sigma(\eta_R),
	\end{align*} 
	which holds if and only if there is a solution $\sigma$ of $\eta$. 
\end{proof}

We now combine the proofs of Theorems~\ref{closureComplementTheorem} and~\ref{thm:intersect}, and observe:
\begin{theorem}\label{thm:inc}
	Given  $M_1,M_2\in\fDTMFA$, $\lang(M_1)\subseteq \lang(M_2)$ can be decided in $\PSPACE$.
\end{theorem}

%\subsection{Proof of Theorem~\ref{thm:inc}}
\begin{proof}
	We use the simple fact that $L_1\subseteq L_2$ holds if and only if $L_1\cap (\Sigma^* \mdif L_2)=\emptyset$. Due to Proposition~\ref{makeCompleteProposition}, we can assume that $M_2$ is  complete (see Section~\ref{sec:DTMFA}).
	
	While we could use Theorem~\ref{thm:intersect} together with Theorem~\ref{closureComplementTheorem} to show that the inclusion problem is decidable, the proof of Theorem~\ref{closureComplementTheorem} uses a construction that can lead to an exponential blowup in the number of states. The reason for this is that even in a complete $\DTMFA$, we cannot simply toggle the acceptance behaviour of states, as the automaton might continue its computation by recalling memories that contain $\emptyword$.
	
	But as we shall see, it is possible to adapt the proof of Theorem~\ref{thm:intersect} to handle this as well.
	First, note that we do not need to consider how to handle memory recall failures, as this is already part of the proof (we can simply add or remove the modifications that we discussed for $\DTMFAacc$). The first modification is that the algorithm now guesses a condensed run $\kappa$ that ends in an state $p_{\ell}$ that is not accepting. But to ensure that we can treat this state as an accepting state, we need to ensure that no accepting state can be reached from it. Instead of putting this into the formula, we make an additional guess: For each variable $x\in\{1,\ldots,k\}$, we also guess a language $L_x$ such that $L_x = \{\emptyword\}$ or $L_x=\Sigma^+$. Formally, in addition to $\kappa$ and $\ell$, the algorithm guesses a function $f\colon \{1,\ldots,k\}\to \{M_{\emptyword}, M_{\Sigma^+} \}$, where $M_{\emptyword}$ and $M_{\Sigma^+}$ are $\NFA$ with $\lang(M_{\emptyword})=\{\emptyword\}$ and $\lang(M_{\Sigma+})=\Sigma^+$.
	
	It then checks whether it is possible to reach an accepting state from $q$, using only $\emptyword$-transitions and memory recalls for variables $x$ with $f(x)=M_{\emptyword}$. If that is the case, the algorithm rejects the guess. Otherwise, it constructs a formula $\varphi_{\kappa,f}$, which is obtained from $\varphi_{\kappa}$ by adding the following formula to the conjunction:
	$$
	\exists y_1,\ldots,y_k\colon \bigland_{x\in\{1,\ldots,k\}} \bigl((y_x = \hat{\eta}_x)\land \const{f(x)}{y_i}\bigr),
	$$
	where each $\hat{\eta}_{x}$ is chosen to represent the content of the variable $x$ in $p_{\ell}$, like the $\eta_j$ in the proof of Theorem~\ref{thm:intersect}. Hence, this formula checks whether the contents of the variables when reaching the state $p_{\ell}$ conform to the guessed function $f$.
	
	Hence, for every word that would be rejected by $M_2$, we can guess appropriate $\ell$, $\kappa$, and $f$, which allows us to decide the intersection emptiness of $\lang(M_1)$ and $(\Sigma^*\mdif \lang(M_2))$ in $\PSPACE$ as in the proof of Theorem~\ref{thm:intersect}. Hence, inclusion is decidable in $\PSPACE$.
\end{proof}

Obviously, this implies that equivalence for $\fDTMFA$ is decidable in $\PSPACE$, and, furthermore, this also holds for $\fDRX$, which is an interesting contrast to non-deterministic $\fRX$: As shown by Freydenberger~\cite{fre:ext}, equivalence (and, hence, inclusion and minimization) are undecidable for $\fRX$ (while~\cite{fre:ext} does not explicitly mention the concept, the regex in that proof are vstar-free, as discussed in~\cite{fre:doc}). 
Hence, Theorem~\ref{thm:inc} also yields a minimization algorithm for $\fDRX$ and $\fDTMFA$ that works in $\PSPACE$ (enumerate all smaller candidates and check equivalence). We leave open whether this is optimal, but observe that even for $\DREG$, minimization is $\NP$-complete, see Niewerth~\cite{nie:thesis}. 

\section{A Relaxation of Determinism}\label{sec:relax}

%Next, we discuss a potential extension (or rather relaxation) of determinism. 
One could argue that Definition~\ref{def:drx} is overly restrictive; e.\,g., consider  $\alpha\df \bind{x}{\ta^+}\bind{y}{\tb^+}\tc (\rr{x}\ror\rr{y})$. Then $\alpha$ is not deterministic; but as the contents of $x$ and $y$ always start with $\ta$ or $\tb$ (respectively), deterministic choices between $\rr{x}$ and $\rr{y}$ are possible by looking at the current letter of the input word. Analogous observations can be made for $\TMFA$. These observations lead to the following definition of $\ell$-determinism.\par
Let $\ell \geq 1$ and let $u, v \in \Sigma^*$. The words $u$ and $v$ are \emph{$\ell$-prefix equivalent}, denoted by $u \prefequi_{\ell} v$, if $u$ is a prefix of $v$, $v$ is a prefix of $u$ or their longest common prefix has a size of at least $\ell$. By $u \nprefequi_{\ell} v$, we denote that $u$ and $v$ are not $\ell$-prefix equivalent. Note that in order to check for two words $u$ and $v$ whether or not $u \prefequi_k v$, it is sufficient to compare the first $\min\{k, |u|, |v|\}$ symbols of $u$ and $v$. \par
Based on this, we define the notion of \emph{$\ell$-deterministic} $\TMFA$ (for short: $\ellDTMFA$) as a relaxation of the criteria of $\DTMFA$: In contrast to the latter, an $\ell$-deterministic $M \in \TMFA(k)$ can have states $q$ with multiple outgoing memory recall-transitions, as long as 

\begin{enumerate}
	\item these recall distinct memories, and
	\item for every reachable configuration $(q, v, (u_1, r_1), \ldots, (u_k, r_k))$ of $M$, $u_i \nprefequi_{\ell} u_j$ holds for all $i\neq j$ that appear on the recall transitions of $q$.
\end{enumerate}
 For technical reasons, we define $0$-$\DTMFA$ to coincide with $\DTMFA$. \par
%More precisely, we define the notion of \emph{$\ell$-deterministic} $\TMFA$ as a relaxation of the criteria of $\DTMFA$: In contrast to the latter, an $\ell$-deterministic $\TMFA$ can have states $q$ with multiple memory recall-transitions, as long as these recall distinct memories, and if $q$ is reached in some computation, then for each pair of these recalled memories, the contents differ in the first $\ell$ positions. 
Next, we note that this relaxation from determinism to $\ell$-determinism does not increase the expressive power of $\DTMFA$.
%(intuitively, storing the length $\ell$ prefixes of the memory contents allows  making $\ell$-deterministic memory recall transitions deterministic):
\begin{proposition}\label{prop:ldet}
%	Let $\ell\geq 1$. For every $\ell$-deterministic $M \in \DTMFA$, there is an $M' \in \DTMFA$ with $\lang(M) = \lang(M')$.
$\DTMFA = \bigcup_{\ell \geq 0}\ellDTMFA$.
\end{proposition}

%\subsection{Proof of Proposition~\ref{prop:ldet}}

\begin{proof}
The inclusion $\DTMFA \subseteq \bigcup_{\ell \geq 0}\ellDTMFA$ holds by definition. In order to show the converse inclusion, let $M \in \ellDTMFA(k)$, for some $\ell \geq 1$. We transform $M$ into an equivalent $\DTMFA$ $M'$ as follows. We implement $k$ auxiliary memories (called \emph{state-memories} in the following) in the finite state control, which can store words of length at most $\ell$, i.\,e., we replace every state $q$ by states $[q, m_1, m_2, \ldots, m_k]$, where $m_i \in \Sigma^*$, $|m_i| \leq \ell$, $1 \leq i \leq k$. The general idea is that $M'$ simulates $M$ in such a way that whenever $M$ reaches a configuration $(q, v, (u_1, r_1), \ldots, (u_k, r_k))$, then $M'$ reaches the configuration with state $[q, m_1, m_2, \ldots, m_k]$ and memory configurations $(u'_i, r_i)$, $1 \leq i \leq k$, such that, for every $i$, $1 \leq i \leq k$, $u_i = m_i u'_i$ and if $u'_i \neq \varepsilon$, then $|m_i| = k$. This can be achieved as follows.\par
	Initially, all memories and state-memories are empty and closed (to this end, the finite state control contains a flag for each state-memory, indicating whether or not it is open). If $M$ recalls memory $i$, then $M'$ consumes the content of the state-memory $i$ from the input, symbol by symbol, and then applies a memory recall instruction on memory $i$ (note that memory $i$ might be empty). If the consumption of the content of state-memory $i$ fails, i.\,e., it is not a prefix of the remaining input, then we move to the state $\trapstate$.\par
	Whenever $M$ opens memory $i$, $M'$ empties the state-memory $i$ and marks it as open, but does not yet open memory $i$. The scanned input is now stored as follows. If a single symbol is read and the state-memory currently stores a word of length at most $\ell - 1$, then this symbol is appended to the state-memory (furthermore, if the new symbol exhausts the state-memory's capacity, then memory $i$ is opened), and if the state-memory already stores a word of length $\ell$, then the symbol is automatically stored in the open memory $i$. \par
	On the other hand, if $M$ consumes a prefix $u$ of the input by a memory recall instruction for some memory $j$, i.\,e., in $M'$, the state-memory $j$ stores some $u'$ and memory $j$ stores some $u''$ with $u = u' u''$, then this is simulated by $M'$ as follows. We start consuming $u'$ symbol by symbol and store every symbol in state-memory $i$. If this is possible without exhausting the capacity of state-memory $i$ (i.\,e., state-memory $i$ now stores a word of length at most $\ell-1$), then $|u'| < \ell$, which implies $u'' = \varepsilon$ and we are done. On the other hand, if the consumption of $u'$ exhausts the state-memories capacity, i.\,e., $u' = v' v''$, where $v'$ is the largest prefix that fits in state-memory $i$ (note that $v' = u$ is possible), then we open memory $i$ and fill it with $v'' u''$ by first consuming $v''$ symbol by symbol and then consulting memory $j$.\par
	We implement the modifications from above in such a way that whenever in $M$ there is a nondeterministic choice of the form that, for some state $q$ and several $i_1, i_2, \ldots, i_s$, $1 \leq i_j \leq k$, $1 \leq j \leq s$, each $\delta(q, i_j)$, $1 \leq j \leq s$, is defined (note that, since $M$ is $\ell$-deterministic, these are the only possible non-deterministic choices), then this is implemented in $M'$ by $s$ many $\varepsilon$-transitions from the states $[q, m_1, m_2, \ldots, m_k]$. Since the modifications from above do not require any nondeterminism, there is a one-to-one correspondence between the nondeterministic choices of $M$ and $M'$. We further note that, for every $j$, $1 \leq j \leq s$, the $\varepsilon$-transition for consulting memory $i_j$ is followed by a path of states, in which the content of state-memory $i_j$ is consumed symbol by symbol, followed by a recall of memory $i_j$ (and, simultaneously, for every open memory $i$, the state-memory is filled with the consumed symbols until it is full and then memory $i$ is opened). The memory recall performed by this path of states either fails, which can happen in the phase where the content of the state-memory is matched with the input or in the actual recall of the memory, or it successfully simulates the memory recall. We shall now describe how the nondeterministic choices of $M'$ can be removed.\par
	Instead of nondeterministically choosing one of these paths, we carry them out in parallel as follows. We start consuming a prefix of the remaining input and compare it, symbol by symbol, with the contents of the state-memories $i_j$, $1 \leq j \leq s$. Whenever the next input symbol does not match the next symbol of a state-memory $i_j$, we mark this memory as \emph{inactive} and ignore it from now on. If all memories are inactive, we change to state $\trapstate$ and if there is exactly one active memory $i_j$ left, we conclude the consultation of this memory (i.\,e., we match the remaining part of the state-memory $i_j$ with the input and then consult memory $i_j$). In particular, we note that if a state-memory has been completely and successfully matched with a prefix of the input and there is a another memory still active, then the contents of these memories are $\ell$-prefix equivalent, which is a contradiction to the $\ell$-determinism of $M$. Consequently, we encounter the situation that either all memories are inactive or that exactly one active one is left, before a state-memory is completely matched with a prefix of the input. Obviously, this procedure is completely deterministic and it results in an equivalent automaton. 
\end{proof}

%For the sake of the argument, let $\alpha\in\RX$ be $\ell$-deterministic if and only if $\glush(\alpha)$ is. 

Next, we consider the problem to decide whether a given $\TMFA$ is $\ell$-deterministic for some given $\ell \geq 1$, and we shall see that this is a hard problem (even for $\fTMFA$). Moreover, extending the notion of $\ell$-determinism to $\RX$ by defining $\alpha\in\RX$ to be $\ell$-deterministic if and only if $\glush(\alpha)$ is, we shall see that deciding $\ell$-determinism is also hard for $\RX$ (and even $\fRX$).

%Moreover, extending the notion of $\ell$-determinism to $\RX$ by defining $\alpha\in\RX$ to be $\ell$-deterministic if and only if $\glush(\alpha)$, 
%
%(note that one of the main results of this paper is that determinism of a  )

\begin{proposition}\label{prop:ldetcomplexity}
	For every $\ell\geq 1$, deciding whether a $\TMFA$ is $\ell$-deterministic is $\PSPACE$-complete. The problem is $\coNP$-complete if the input is restricted to $\fTMFA$. These lower bounds hold even if we restrict the input to $\RX$ and $\fRX$, respectively.
\end{proposition} 

%\subsection{Proof of Proposition~\ref{prop:ldetcomplexity}}\label{app:ldetcomplexity}
\begin{proof}
	Before we proceed to the actual proof, we  briefly discuss why it is possible to treat non-deterministic regex as an input for the problem, considering that the number of transitions in $\glush(\alpha)$ can be exponential (in the number of variables of $\alpha$). While this is true in general, the non-deterministic regex that have these blowups are also not $\ell$-deterministic: As soon as an $M\in\TMFA$ has more than one transition from one state to another, it is not $\ell$-deterministic. Hence, we can use an algorithm that decides $\ell$-determinism for $\TMFA$ to decide $\ell$-determinism for $\RX$ by converting every input $\alpha\in\RX$ into $\glush(\alpha)$ according to the proof of Theorem~\ref{thm:glushkov}, but aborting if $\ograph{\alpha}$ contains nodes $u$ and $v$ with at least two edges from $u$ and $v$ (if these occur, $\alpha$ can be rejected as not $\ell$-deterministic, regardless which $\ell$ was chosen).
	
	\subparagraph*{Upper bounds:}
	In order to  prove the upper bounds, let $M\in\TMFA$ and $\ell\geq 1$. Assume that $M$ is not deterministic, but only violates the criteria by having states with multiple outgoing memory recall transitions for different variables (if any other violation of the criteria occurs, $M$ cannot be $\ell$-deterministic). Now, $M$ is not $\ell$-deterministic if and only if there exists a state $q$  in $M$ that has outgoing memory recall transitions for two different variables $x$ and $y$, and there is a run of $M$ that reaches $q$ while $x$ and $y$ contain words $w_x$ and $w_y$ (respectively) such that $w_x \prefequi_{\ell} w_y$. We show this property can be decided in $\PSPACE$ in general, and in $\NP$ if $M$ is memory-cycle-free. The claim of the Proposition follows then directly if $M$ is memory-cycle-free, and from the closure of $\PSPACE$ under complementation in the general case.
	
	We first consider the general case: The $\PSPACE$ algorithm guesses a state $q$ that has two outgoing memory recall transitions for variables $x$ and $y$. It then guesses its way from $q_0$ through the automaton, while storing for each variable $z$ of $M$ (not just $x$ and $y$) the first $\ell$ letters of the stored word $w_z$ (in order to determine these for a memory $z$, it suffices to know all terminal edges that are traversed while $z$ is open, and at most $\ell$ letters of each memory $z'$ that is referenced while $z$ is open). If the algorithm reaches $q$ while $w_x\prefequi_{\ell} w_y$, the algorithm correctly identifies $M$ as not $\ell$-deterministic. Hence, this can be decided in $\PSPACE$.
	
	For the memory-cycle-free case, we combine this with the condensed runs from the proof of Theorem~\ref{thm:intersect}. The $\NP$-algorithm first guesses a condensed run $\kappa$ of $M$ that ends at $q$ with outgoing memory recall transitions for $x$ and $y$. In order to determine the first $\ell$ letters of each variable, it then guesses a prefix $u_i$ of length at most $\ell$ for each transition from a $q_i$ to a $p_i$, and  checks whether there is a word in $\lang(M_{q_i,p_i})$ that has $u_i$ as a prefix (where the $\eNFA$ $M_{q_i,p_i}$ is obtained as in the proof of Theorem~\ref{thm:intersect}: Remove all memory transitions from $M$, and take $q_i$ as starting and $p_i$ as only accepting state). It then computes the first $\ell$ letters of $w_x$ and $w_y$, which suffice to determine whether $w_x\prefequi_{\ell} w_y$. If $w_x\prefequi_{\ell} w_y$, the algorithm correctly identifies $M$ as not $\ell$-deterministic. Hence, for memory-cycle-free $\TMFA$, the absence of $\ell$-determinism  can be decided in $\NP$, which means that $\ell$-determinism can be decided in $\coNP$.
	
	\subparagraph*{Lower bound for $\fTMFA$ and $\fRX$:} We prove this claim with a reduction from  the 3-satisfiability problem, which is well-known to be $\NP$ complete (cf.\ Garey and Johnson \cite{gar:com}). Let $\varphi$ be a formula in 3-conjunctive normal form, with variables $V = \{v_1,\ldots,v_k\}$, $k\geq 1$,  where $\varphi\df \bigland_{i=1}^n \varphi_i$ with $\varphi_i = (\lambda_{i,1} \lor \lambda_{i,2} \lor \lambda_{i,3})$, and $\lambda_{i,j}\in \{v,\neg v\mid v\in V\}$ for all $1\leq i \leq n$ and $1\leq j \leq 3$.
	
	Our goal is to construct a $\beta\in\fRX$ that is not $\ell$-deterministic if and only if there is an assignment to the variables in $V$ that satisfies $\varphi$. As the latter problem is $\NP$-complete, deciding whether a vsf-regex is $\ell$-deterministic is $\coNP$-hard. To this end, we first construct an $\alpha\in\fDRX$ that has a variable $z$ that can only contain $\emptyword$ if $\varphi$ has a satisfying assignment, and that otherwise contains a word from $\{\ta\}^+$.
	
	We then define $\beta\df \alpha\cdot \bind{z_1}{\tb^{\ell}}\bind{z_2}{\rr{z}\:\tb^{\ell}}(\rr{z_1}\ror \rr{z_2})$. Note that $\beta$ is not $\ell$-deterministic if and only if $\emptyword$ can be assigned to $z$; as otherwise, $z$ always contains some word from $\ta^{+}\tb$, which means that $z_1$ and $z_2$ already differ on the first letter.
	
	We implement this by modeling each variable $v_i\in V$ of $\varphi$ with two variables $x_i$ and $\hat{x_i}$ in $\alpha$, where an assignment of $1$ to $v_i$ is modeled by setting $x_i$ to $\emptyword$ and $\hat{x}_i$ to $\ta$, while  assigning $0$ is modeled by setting $x_i$ to $\ta$ and $\hat{x}_i$ to $\emptyword$.  Keeping this in mind, we define $\alpha \df \alpf{init} \cdot \alpf{sat}$, where
	\begin{align*}
		\alpf{init} &\df \alpf{init}^1 \cdots \alpf{init}^k,\\
		\alpf{init}^i &\df \bigl(   (\ta \bind{x_i}{\emptyword}\bind{\hat{x}_i}{\ta}  ) \ror    (\tb \bind{x_i}{\ta}\bind{\hat{x}_i}{\emptyword}  )    \bigr)\\
		\intertext{for $1\leq i \leq k$,  as well as }
		\alpf{sat} &\df \alpf{sat}^1\cdots \alpf{sat}^n \cdot \bind{z}{\rr{y_1}\cdots\rr{y_n}},\\
		\alpf{sat}^i &\df   (\ta \cdot \alpf{lit}^{i,1}) \ror  \Bigl(\tb  \bigl(  (\ta \cdot \alpf{lit}^{i,2})     \ror  (\tb \cdot \alpf{lit}^{i,3})   \bigr)            \Bigr),\\
		\alpf{lit}^{i,j}&\df \begin{cases}
			\bind{y_i}{\rr{x_l}} & \text{if $\lambda_{i,j}=v_l$,}\\
			\bind{y_i}{\rr{\hat{x}_l}} & \text{if $\lambda_{i,j}=\neg v_l$}
		\end{cases}
	\end{align*}
	for $1\leq i \leq n$, and $1\leq j\leq 3$. 
	
	Now, observe that $\alpha$ is deterministic, as each part of a disjunction starts with a unique first letter ($\ta$ or $\tb$); and $\alpha$ is obviously vstar-free. To see that $\alpha$ can assign $\emptyword$ to $z$ if and only if $\varphi$ has a satisfying assignment, we read $\alpha$ from left to right: First, $\alpf{init}$ ensures that for each pair of variables $x_i$ and $\hat{x}_i$, exactly one is bound to $\emptyword$, and the other to $\ta$ (recall that setting $x_i$ to $\emptyword$ corresponds to assigning $1$ to $v_i$). Next, for each clause $\varphi_i$,  $\alpf{sat}^i$ stores the value of one of the literals $\lambda_{i,j}\in\{v_l,\neg v_l\}$ under the chosen assignment  $y_i$, by recalling the appropriate $x_l$ or $\hat{x}_l$. Thus, $y_i$ can only contain $\emptyword$ if the assignment satisfies $\varphi_i$. Finally, all $y_i$ are concatenated, and the result is stored in $z$. Hence, $z$ can only contain $\emptyword$ if all clauses $\varphi_i$ are satisfied, which means that $\varphi$ is satisfied. Likewise, each satisfying assignment can be used to make the appropriate choices in the $\alpf{init}^i$ and $\alpf{sat}^j$ such that $z$ contains $\emptyword$.
	
	Hence, as explained above, $\beta$ is $\ell$-deterministic if and only if $\varphi$ has no satisfying assignment, which means that deciding whether a vstar-free regex is $\ell$-deterministic is $\coNP$-hard. As we already showed the matching upper bound, the problem is $\coNP$-complete.
	
	\subparagraph*{Lower bound for $\TMFA$ and $\RX$:} We show this with a reduction from the intersection emptiness problem for $\DFA$, which is defined as follows: Given  $M_1,\ldots,M_n\in\DFA$ for some $n\geq 2$, is there a $w\in\Sigma^*$ with  $w\in \lang(M_i)$ for all $1\leq i\leq n$? This problem is $\PSPACE$-complete (cf.\ Kozen~\cite{koz:low}). 
	
	As in the case for vstar-free regex, we first construct an $\alpha\in\DRX$ that has a variable $z$ such that it is possible to assign $\emptyword$ to $z$ if and only if the intersection of the $\lang(M_i)$ is not empty (and which is set to a word from $\{\ta\}^+$ otherwise), and then  define 
	$$\beta\df \alpha\cdot \bind{z_1}{\tb^{\ell}}\bind{z_2}{\rr{z}\:\tb^{\ell}}(\rr{z_1}\ror \rr{z_2}).$$
	Again, $\beta$ is not $\ell$-deterministic if and only if $\emptyword$ can be assigned to $z$; as otherwise, $z$ always contains some word from $\ta^{+}\tb$.
	
	Consider  $M_1,\ldots,M_n\in\DFA$ with $M_i = (\Sigma,Q_i,q_{i,0}, \delta_i,F_i)$. In order to simplify the construction, we assume $Q_i = \{q_{i,0},\ldots, q_{i,m}\}$ for some $m\geq 1$, and  $\Sigma\supseteq \{\ta,\tb,\tc_0,\ldots,\tc_{\max(m,n)}\}$. We shall discuss in the proof how the construction can be adapted to a binary alphabet, but using an unbounded alphabet is simpler.
	
	The main idea of the construction is that each state $q_{i,j}$ is represented by a variable $x_{i,j}$, which can take either $\ta$ or $\emptyword$ as values. The regex $\alpha$ contains a subexpression $\alpf{iter}$ which uses a Kleene star to simulate  all $M_i$ in parallel on the same input. In particular, it ensures that $x_{i,j}$ can be set to $\emptyword$ if and only if $M_i$ can enter state $q_{i,j}$ at the current point of the parallel simulation. Using this, we shall see that it is possible to set $z$ to $\emptyword$ if and only if all $M_i$ can reach an accepting state at the same time. We define
	$$\alpha \df \alpf{init}\cdot \bigl(\ta \cdot \alpf{iter}  \bigr)^* \cdot \tb \cdot \alpf{acc}$$
	Before we define the subexpressions of $\alpha$, we observe that the use of the  Kleene star does not affect determinism, as the terminals $\ta$ and $\tb$ signal whether there should be another iteration of the star or not (respectively).
	The subexpressions of $\alpha$ are defined as follows:
	\begin{align*}
		\alpf{init}&\df \alpf{init}^1\cdots \alpf{init}^n\,,\\
		\alpf{init}^i &\df \bind{x_{i,0}}{\emptyword} \bind{x_{i,1}}{\ta}\cdots \bind{x_{i,m}}{\ta}\,,
	\end{align*}		
%			\intertext{for all $1\leq i\leq n$. This represents that each automaton $M_i$ is in its starting state $q_{0,i}$. Furthermore, to simulate the behaviour of the automata $M_i$, we define }
for all $1\leq i\leq n$. This represents that each automaton $M_i$ is in its starting state $q_{0,i}$. Furthermore, to simulate the behaviour of the automata $M_i$, we define 
	\begin{align*}
		\alpf{iter}&\df \bigl(   (\ta\cdot\alpf{step}^{\ta}) \ror (\tb\cdot\alpf{step}^{\tb})     \bigr)\cdot\alpf{switch},\\
		\alpf{step}^{d}&\df  \alpf{step}^{d,1} \cdots \alpf{step}^{d,n},\\
		\alpf{step}^{d,i}&\df \bigror_{0\leq j \leq m} \Bigl(  \tc_j\cdot \bigl(   \bigror_{\substack{0\leq l \leq m,\\ \delta_i(q_{l,i},d)=q_{j,i}}}  \tc_l\cdot  \bind{\hx_{i,j}}{\rr{x_{i,l}}}      \bigr)\cdot \alpf{dump}^{i,j}     \Bigr),\\
		\alpf{dump}^{i,j}&\df \bind{\hx_{i,0}}{\ta} \cdots \bind{\hx_{i,j-1}}{\ta} \bind{\hx_{i,j+1}}{\ta}   \cdots \bind{\hx_{i,m}}{\ta},\\
		\alpf{switch}&\df \alpf{switch}^1 \cdots \alpf{switch}^n,\\
		\alpf{switch}^i&\df \bind{x_{i,0}}{\rr{\hx_{i,0}}}\cdots \bind{x_{i,m}}{\rr{\hx_{i,m}}}
	\end{align*}
	for $d\in\{\ta,\tb\}$, $1\leq i \leq n$,  and $0\leq j \leq m$. Each $\alpf{step}^{d,i}$ picks a pair of states $q_{i,j}$ and  $q_{i,l}$ of $M_i$, such that $\delta(q_{i,l},d)=q_{i,j}$. Less formally, $q_{i,j}$ is the successor state of $q_{i,l}$ on input $d$. The temporary variable $\hx_{i,j}$ is then set to the content of $x_{i,l}$, while all other temporary variables $\hx_{i,j'}$ with $j'\neq j$ are set to $\ta$, using $\alpf{dump}^{i,j}$. 
	
	Hence, each iteration of $\alpf{step}^d$ can set $\hx_{i,j}$ to $\emptyword$ if and only if $q_{i,j}$ is the successor state on input $d$ for a state $q_{i,l}$ such that $x_{i,l}$ contains $\emptyword$. In other words, each iteration of $\alpf{iter}$ uses a subexpression $\alpf{step}^d$ to simulates all $M_i$ in parallel on the input letter $d\in\{\ta,\tb\}$,  and $\alpf{switch}$ sets each $x_{i,j}$ to the same content as its corresponding temporary variable $\hx_{i,j}$. 
	
	As an aside, note that it is possible to adapt the construction to a binary terminal alphabet. To do so, one replaces the disjunctions over the terminals $\tc_j$ and $\tc_l$ with nested disjunctions over $\ta$ and $\tb$, as in the expressions $\alpf{sat}^i$ in the proof for the lower bound for $\fRX$ above. 
	
	Regardless of the number of terminal letters, we define the remaining subexpressions as follows:
	\begin{align*}
		\alpf{acc}&\df \alpf{acc}^1 \cdots \alpf{acc}^n \cdot \bind{z}{\rr{y_1}\cdots \rr{y_n}},\\
		\alpf{acc}^i&\df \bigror_{\substack{0\leq j \leq m, \\ q_{i,j}\in F_i}} \tc_j\cdot \bind{y_i}{\rr{x_{i,j}}}
	\end{align*}
	for all $1\leq i\leq n$. Again, this disjunction can be adapted to a binary terminal alphabet, as described above. 
	
	It is possible to set $z$ to $\emptyword$ if and only if every $y_i$ can be set to $\emptyword$. In turn, this is possible if and only if for every $M_i$, there is an accepting state $q_{j,i}$ such that $x_{j,i}$ can be set to $\emptyword$. As established above, $\alpf{iter}$ ensures that this is only possible if these states can be reached by simulating all $M_i$ in parallel on the same input. Hence, $z$ can be set to $\emptyword$ if and only if the intersection of all $\lang(M_i)$ is not empty.
	
	As discussed above,  $\alpha$ is deterministic (we discussed the use of the Kleene star above, and all branches disjunctions start with characteristic letters). Hence, $\beta$ is  $\ell$-deterministic if and only if the intersection of the $\lang(M_i)$ is  empty. As $\beta$ can obviously be constructed in polynomial time, this shows that deciding whether a regex is $\ell$-deterministic is $\PSPACE$-hard. As we already established the matching upper bound, this concludes the whole proof.
\end{proof}
Furthermore, while this definition of $\ell$-determinism is only concerned with choices between different variables, it is also possible to adapt the notion of $1$-determinism to include the distinction between a variable and a terminal. For example, the expression $\bind{x}{\ta^+}\tb (\tb\ror \rr{x})^*$ is not deterministic; but as the content of $x$ always starts with $\ta$, such cases could be considered $1$-deterministic. Propositions~\ref{prop:ldet} and~\ref{prop:ldetcomplexity} can be directly adapted to this extended notion of $1$-determinism.

%Hence, while we can decide efficiently whether a $\TMFA$ or a regex is deterministic, detecting $\ell$-determinism is costly, even for $\ell=1$. The same holds if we adapt the definition to distinguish between variables and terminals (see Section~A.21 in the full version of the paper~\cite{fre:drx}).
%\todomcom{Removed paragraph that became useless due to things being copied from Appendix. We really both need to do a final careful proofreading run.}
%!TEX root=det.tex
\section{Conclusions and Further Directions}\label{sec:conclusions}
Based on $\TMFA$, an automaton model for regex, we extended the notion of determinism from regular expressions to regex. Although the resulting language class cannot express all regular languages, it is still rich; and by using a generalization of the Glushkov construction, deterministic regex can be converted into a $\DTMFA$, and the membership problem can then be solved quite efficiently. Although we did not discuss this, the construction is also compatible with the Glushkov construction with counters by Gelade, Gyssens, Martens~\cite{gel:reg}. Hence, one can add  counters to $\DRX$ and $\DTMFA$ without affecting the complexity of  membership. 
 
Many challenging questions remain open, for example: Can the more advanced results for $\DREG$ be adapted to $\DRX$, i.\,e., can $\glush(\alpha)$ be computed more efficiently (as in ~\cite{bru:reg,pon:new}), or is it even possible, like in~\cite{gro:det}, to avoid computing $\glush(\alpha)$? 
Is effective minimization possible for $\DTMFA$ or $\DRX$?   
Is it decidable whether a $\DTMFA$ defines a $\DRX$-language?
Are inclusion and equivalence decidable for $\DRX$ or $\DTMFA$? 
Can determinism be generalized to larger classes of regex without making the membership problem intractable?

 \section*{Acknowledgements}
The authors thank Wim Martens for helpful feedback,  Matthias Niewerth, for pointing out that~$v_n$ must be a factor of~$p_n$ in the jumping lemma, and Martin Braun, for creating a library and tool for $\DRX$ and $\DTMFA$ (available at~\cite{bra:moar}).

%\bibliographystyle{plain}
%\bibliography{det_arxiv}

\begin{thebibliography}{10}

\bibitem{abigail}
Abigail.
\newblock Re: Random number in perl.
\newblock Posting in the newsgroup comp.lang.perl.misc, October 1997.
\newblock Message-ID slrn64sudh.qp.abigail@betelgeuse.wayne.fnx.com.

\bibitem{aho:alg}
Alfred~V. Aho.
\newblock Algorithms for finding patterns in strings.
\newblock In Jan van Leeuwen, editor, {\em Handbook of Theoretical Computer
  Science}, volume~A, chapter~5, pages 255--300. Elsevier, Amsterdam, 1990.

\bibitem{ang:fin}
Dana Angluin.
\newblock Finding patterns common to a set of strings.
\newblock {\em J. Comput. Syst. Sci.}, 21:46--62, 1980.

\bibitem{bar:ext-c}
Pablo Barcel{\'{o}}, Carlos~A. Hurtado, Leonid Libkin, and Peter~T. Wood.
\newblock Expressive languages for path queries over graph-structured data.
\newblock In {\em Proc. {PODS} 2010}, 2010.

\bibitem{bex:lea}
Geert~Jan Bex, Wouter Gelade, Frank Neven, and Stijn Vansummeren.
\newblock Learning deterministic regular expressions for the inference of
  schemas from {XML} data.
\newblock {\em {ACM} Trans. Web}, 4(4):14, 2010.

\bibitem{bra:moar}
Martin Braun.
\newblock moar -- {D}eterministic {R}egular {E}xpressions with
  {B}ackreferences, 2016.
\newblock Accessed February~2018.

\bibitem{w3c:dtd}
Tim Bray, Jean Paoli, C.~M. Sperberg-McQueen, Eve Maler, and Fran\c{c}ois
  Yergeau.
\newblock Extensible markup language {XML} 1.0 (fifth edition). {W3C}
  recommendation.
\newblock Technical Report \url{https://www.w3.org/TR/2008/REC-xml-20081126/},
  W3C, November 2008.

\bibitem{bru:reg}
Anne Br{\"{u}}ggemann{-}Klein.
\newblock Regular expressions into finite automata.
\newblock {\em Theor. Comput. Sci.}, 120(2):197--213, 1993.

\bibitem{bru:one}
Anne Br{\"{u}}ggemann{-}Klein and Derick Wood.
\newblock One-unambiguous regular languages.
\newblock {\em Inf. Comput.}, 142(2):182--206, 1998.

\bibitem{cam:afo}
Cezar C\^ampeanu, Kai Salomaa, and Sheng Yu.
\newblock A formal study of practical regular expressions.
\newblock {\em Int. J. Found. Comput. Sci.}, 14:1007--1018, 2003.

\bibitem{car:one}
Benjamin Carle and Paliath Narendran.
\newblock On extended regular expressions.
\newblock In {\em Proc. LATA 2009}, 2009.

\bibitem{car:gen}
Pascal Caron, Yo{-}Sub Han, and Ludovic Mignot.
\newblock Generalized one-unambiguity.
\newblock In {\em Proc. DLT 2011}, pages 129--140, 2011.

\bibitem{chr:fin}
Marek Chrobak.
\newblock Finite automata and unary languages.
\newblock {\em Theor. Comput. Sci.}, 47(3):149--158, 1986.

\bibitem{chr:fin-err}
Marek Chrobak.
\newblock Errata to: ``{F}inite automata and unary languages''.
\newblock {\em Theor. Comput. Sci.}, 302(1):497--498, 2003.

\bibitem{cox:reg}
Russ Cox.
\newblock Regular expression matching can be simple and fast (but is slow in
  {J}ava, {P}erl, {PHP}, {P}ython, {R}uby, \ldots).
\newblock \url{https://swtch.com/~rsc/regexp/regexp1.html}, 2007.
\newblock Accessed February~2018.

\bibitem{cze:dec}
Wojciech Czerwinski, Claire David, Katja Losemann, and Wim Martens.
\newblock Deciding definability by deterministic regular expressions.
\newblock In {\em Proc. {FOSSACS} 2013}, pages 289--304, 2013.

\bibitem{die:mak}
Volker Diekert.
\newblock Makanin's {A}lgorithm.
\newblock In {\em Algebraic Combinatorics on Words\/} \cite{lot:alg},
  chapter~12.

\bibitem{die:mor}
Volker Diekert.
\newblock More than 1700 years of word equations.
\newblock In {\em Proc. CAI 2015}, 2015.

\bibitem{die:fin}
Volker Diekert, Artur Je{\.z}, and Wojciech Plandowski.
\newblock Finding all solutions of equations in free groups and monoids with
  involution.
\newblock {\em Inf. Comput.}, 251:263--286, 2016.

\bibitem{fag:spa}
Ronald Fagin, Benny Kimelfeld, Frederick Reiss, and Stijn Vansummeren.
\newblock Document spanners: {A} formal approach to information extraction.
\newblock {\em J. {ACM}}, 62(2):12, 2015.

\bibitem{FerSch2015}
Henning Fernau and Markus~L. Schmid.
\newblock Pattern matching with variables: A multivariate complexity analysis.
\newblock {\em Inform. Comput.}, 242:287--305, 2015.

\bibitem{FerSchVil2015}
Henning Fernau, Markus~L. Schmid, and Yngve Villanger.
\newblock On the parameterised complexity of string morphism problems.
\newblock {\em Theory Comput. Syst.}, 59(1):24--51, 2016.

\bibitem{fre:ext}
Dominik~D. Freydenberger.
\newblock Extended regular expressions: Succinctness and decidability.
\newblock {\em Theory Comput. Sys.}, 53(2):159--193, 2013.

\bibitem{fre:splog}
Dominik~D. Freydenberger.
\newblock A logic for document spanners.
\newblock In {\em Proc. ICDT 2017}, 2017.

\bibitem{fre:doc}
Dominik~D. Freydenberger and Mario Holldack.
\newblock Document spanners: {F}rom expressive power to decision problems.
\newblock {\em Theory of Computing Systems}, 2017.

\bibitem{fre:exp}
Dominik~D. Freydenberger and Nicole Schweikardt.
\newblock Expressiveness and static analysis of extended conjunctive regular
  path queries.
\newblock {\em J. Comput. Syst. Sci.}, 79(6):892 -- 909, 2013.

\bibitem{w3c:xsd}
Shudi~(Sandy) Gao, C.~M. Sperberg-McQueen, and Henry~S. Thompson.
\newblock {W3C} {XML} schema definition language ({XSD}) 1.1 part 1:
  Structures.
\newblock Technical Report
  \url{https://www.w3.org/TR/2012/REC-xmlschema11-1-20120405/}, W3C, April
  2012.

\bibitem{gar:com}
Michael~R. Garey and David~S. Johnson.
\newblock {\em Computers and Intractability}.
\newblock W. H. Freeman and Company, 1979.

\bibitem{gel:reg}
Wouter Gelade, Marc Gyssens, and Wim Martens.
\newblock Regular expressions with counting: Weak versus strong determinism.
\newblock {\em {SIAM} J. Comput.}, 41(1):160--190, 2012.

\bibitem{goy:reg}
Jan Goyvaerts.
\newblock Regular expressions tutorial.
\newblock \url{http://www.regular-expressions.info/tutorial.html}, 2016.
\newblock Accessed February~2018.

\bibitem{gro:det}
Beno{\^{\i}}t Groz and Sebastian Maneth.
\newblock Efficient testing and matching of deterministic regular expressions.
\newblock {\em J. Comput. Syst. Sci.}, 89:372 --399, 2017.

\bibitem{hop:int}
John~E. Hopcroft and Jeffrey~D. Ullman.
\newblock {\em Introduction to Automata Theory, Languages, and Computation}.
\newblock Addison-Wesley, 1979.

\bibitem{w3c:xfun}
Michael Kay.
\newblock {XPath and XQuery Functions and Operators 3.0 W3C Recommendation}.
\newblock Technical Report
  \url{https://www.w3.org/TR/2014/REC-xpath-functions-30-20140408/}, W3C, April
  2014.

\bibitem{kle:rep}
S.~C. Kleene.
\newblock Representation of events in nerve nets and finite automata.
\newblock In C.~E. Shannon, J.~McCarthy, and W.~R. Ashby, editors, {\em
  Automata Studies}, pages 3--42. Princeton University Press, Princeton, NJ,
  1956.

\bibitem{koz:low}
Dexter Kozen.
\newblock Lower bounds for natural proof systems.
\newblock In {\em Proc. FOCS 1977}, 1977.

\bibitem{lat:def}
Markus Latte and Matthias Niewerth.
\newblock Definability by weakly deterministic regular expressions with
  counters is decidable.
\newblock In {\em Proc. {MFCS} 2015}, 2015.

\bibitem{los:clo}
Katja Losemann, Wim Martens, and Matthias Niewerth.
\newblock Closure properties and descriptional complexity of deterministic
  regular expressions.
\newblock {\em Theor. Comput. Sci.}, 627:54--70, 2016.

\bibitem{lot:alg}
M.~Lothaire.
\newblock {\em Algebraic Combinatorics on Words}, volume~90 of {\em
  Encyclopedia of mathematics and its applications}.
\newblock Cambridge University Press, 2002.

\bibitem{lu:dec}
Ping Lu, Joachim Bremer, and Haiming Chen.
\newblock Deciding determinism of regular languages.
\newblock {\em Theory Comput. Syst.}, 57(1):97--139, 2015.

\bibitem{mar:com}
Wim Martens, Frank Neven, and Thomas Schwentick.
\newblock Complexity of decision problems for {XML} schemas and chain regular
  expressions.
\newblock {\em {SIAM} J. Comput.}, 39(4):1486--1530, 2009.

\bibitem{mou:bio}
D.~W. Mount.
\newblock {\em Bioinformatics: Sequence and Genome Analysis}.
\newblock Cold Spring Harbor Laboratory Press, Woodbury, NY, 2004.

\bibitem{mur:tax}
Makoto Murata, Dongwon Lee, Murali Mani, and Kohsuke Kawaguchi.
\newblock Taxonomy of {XML} schema languages using formal language theory.
\newblock {\em ACM TOIT}, 5(4):660--704, 2005.

\bibitem{nie:thesis}
Matthias Niewerth.
\newblock {\em Data Definition Languages for XML Repository Management
  Systems}.
\newblock PhD thesis, TU Dortmund, 2015.

\bibitem{pon:new}
Jean{-}Luc Ponty, Djelloul Ziadi, and Jean{-}Marc Champarnaud.
\newblock A new quadratic algorithm to convert a regular expression into an
  automaton.
\newblock In {\em Proc. {WIA} '96}, 1996.

\bibitem{sch:ins}
Markus~L. Schmid.
\newblock Inside the class of regex languages.
\newblock {\em Int. J. Found. Comput. Sci.}, 24(7):1117--1134, 2013.

\bibitem{sch:cha}
Markus~L. Schmid.
\newblock Characterising {REGEX} languages by regular languages equipped with
  factor-referencing.
\newblock {\em Inform. Comput.}, 249:1--17, 2016.

\bibitem{spe:xml}
C.~M. Sperberg-McQueen and H.~Thompson.
\newblock {XML} schema.
\newblock \url{http://www.w3.org/XML}, 2005.

\bibitem{var:fro}
M.~Y. Vardi.
\newblock {\em From monadic logic to PSL}.
\newblock Pillars of Computer Science. Springer, Berlin, 2008.

\bibitem{wal:pro}
L.~Wall, T.~Christiansen, and J.~Orwant.
\newblock {\em Programming Perl}.
\newblock O’Reilly Media, Sebastopol, CA, 2000.

\end{thebibliography}

\end{document}